\documentclass[letterpaper,twocolumn,10pt]{article}
\usepackage[dvipsnames]{xcolor}
\usepackage{usenix2019_v3}
\usepackage{graphicx}
\usepackage{bm}
\usepackage{amsmath,amssymb,amsfonts}
\usepackage{textcomp}
\usepackage{comment}
\usepackage{bbding}
\usepackage{placeins}
\usepackage{mdframed}
\usepackage{tcolorbox}
\usepackage{framed}
\usepackage{makecell}
\usepackage{wasysym}
\usepackage{authblk}
\usepackage{algorithm}
\usepackage{algorithmicx}
\usepackage{algpseudocode}
\usepackage{cite}
\usepackage{caption}
\usepackage{float}
\usepackage{multicol}
\usepackage{multirow}
\usepackage{enumerate}
\usepackage{subfig}
\usepackage{hyperref}
\PassOptionsToPackage{hyphens}{url}
\usepackage{color}
\usepackage{diagbox}

\newtheorem{theorem}{Theorem}
\newtheorem{definition}{Definition}
\newtheorem{proof}{Proof}

\begin{document}

\date{}

\author[1]{Dongli Liu}
\author[1,\Envelope]{Wei Wang}
\author[1,2,3,4,\Envelope]{Peng Xu}
\author[1,5]{Laurence T. Yang}
\author[6]{Bo Luo}
\author[7]{Kaitai Liang}
\affil[1]{\textit{Huazhong University of Science and Technology}}
\affil[2]{\textit{Hubei Key Laboratory of Distributed System Security, School of Cyber Science and Engineering}}
\affil[3]{\textit{JinYinHu Laboratory}}
\affil[4]{\textit{State Key Laboratory of Cryptology}}
\affil[5]{\textit{St. Francis Xavier University}}
\affil[6]{\textit{The University of Kansas}}
\affil[7]{\textit{Delft University of Technology}}
\affil[\Envelope]{\small\textit{Corresponding authors: viviawangwei@hust.edu.cn, xupeng@hust.edu.cn}}

\title{{$\textit{d}$}-DSE: Distinct Dynamic Searchable Encryption Resisting Volume Leakage in Encrypted Databases}

\maketitle

\begin{abstract}
Dynamic Searchable Encryption (DSE) has emerged as a solution to efficiently handle and protect large-scale data storage in encrypted databases (EDBs). 
Volume leakage
poses a significant threat, as it enables adversaries to reconstruct search queries and potentially compromise the security and privacy of data.
Padding strategies are common countermeasures for the leakage, but they significantly increase storage and communication costs. 
In this work, we develop a new perspective to handle volume leakage.  
We start with distinct search and further explore a new concept called \textit{distinct} DSE (\textit{d}-DSE). 

We also define new security notions, in particular Distinct with Volume-Hiding security, as well as forward and backward privacy, for the new concept. 
Based on \textit{d}-DSE, we construct the \textit{d}-DSE designed EDB with related constructions for distinct keyword (d-KW-\textit{d}DSE), keyword (KW-\textit{d}DSE), and join queries (JOIN-\textit{d}DSE) and update queries in encrypted databases. 
We instantiate a concrete scheme \textsf{BF-SRE}, employing Symmetric Revocable Encryption. 
We conduct extensive experiments on real-world datasets, such as Crime, Wikipedia, and Enron, for performance evaluation.  
The results demonstrate that our scheme is practical in data search and with comparable computational performance to the SOTA DSE scheme (\textsf{MITRA}*, \textsf{AURA}) and padding strategies (\textsf{SEAL}, \textsf{ShieldDB}).  
Furthermore, our proposal sharply reduces the communication cost as compared to padding strategies, with roughly $6.36$ to \textcolor{black}{53.14}x advantage for search queries. 

\end{abstract}

\setcounter{footnote}{0}
\section{Introduction}

Encrypted databases (EDBs)\cite{Popa2011CyptDB} have been developed to support privacy-preserving data storage and search services.  
They allow clients to outsource sensitive encrypted data to a server and then send search queries to the server who can return the corresponding data. 
These databases can adopt Searchable Symmetric Encryption (\textcolor{black}{SE})\cite{song2000practical,curtmola2011searchable} and its Dynamic version\cite{Kamara2012dynamic} (DSE) to offer various types of privacy-preserving queries, such as keyword \cite{demertzis2018efficient}, range\cite{faber2015rich}, SQL\cite{kamara2018sql} queries, and update. 
Note more related works for EDBs and DSE are presented in Appendix \ref{RDSSE} and \ref{REDB}, respectively.  

Although easily inheriting Forward Privacy and Backward Privacy (\textcolor{black}{FP\&BP}) guarantees from DSE \cite{bost2016ovarphiovarsigma,bost2017forward}, those EDBs are still vulnerable to volumetric attacks \cite{DBLP:conf/ndss/BlackstoneKM20,zhang2023high}. 
The adversary can infer the underlying keywords from search queries\cite{DBLP:conf/ndss/BlackstoneKM20,lambregts2022val,zhang2023high} or even reconstruct the whole EDB \cite{grubbs2018pump,kellaris2016generic} by exploiting \emph{volume leakage} (i.e., the length of the response).

Padding strategies \cite{bost2017thwarting,demertzis2020seal,vo2021shielddb} are popularly applied to mitigate the impact of volume leakage in EDBs. 
They work by adding "dummy" data to the original database to ensure uniform volume for each keyword. 
Thus, from the viewpoint of adversaries, volume leakage with regard to queries is now indistinguishable.
This strategy, however, yields significant costs (see Tab. \ref{Distinct Search}). 
It requires extra storage and computational resources to establish and update the EDB, but also consumes huge communication bandwidth for query response. 
Note we provide more details on padding in Appendix \ref{RPDS}.

\begin{table*}[htb]
    
    \caption{Comparison on related DSE, padding, and our scheme.
    $N$ is the total number of keyword/value pairs.
    $W$ and $F$ are the total number of keywords and files, respectively.
    In Column \textbf{5}, $\checkmark$ indicates that padding strategies can be deployed in the schemes.
    \textcolor{black}{For keyword \textbf{w}}, $a_w$ is the total number of update operations, $n_w$ is the current number of pairs, $d_w$ is the number of deletion, $d_{max}$ is the supported maximum number of deletions, and $s_w^\prime$ is the number of legal search tokens.
    \textcolor{black}{For \textbf{SEAL}}, $x$ means the adjustable padding's parameter, and $n_w^*$ denotes the number of the pairs currently containing $w$ after padding.
    \textcolor{black}{For \textbf{ShieldDB}}, $n_r$ and $n_b$ separately represent the real and padding pairs of the streamed keyword $w_i$, $|B|$ means the padding dataset length, and $f_w$ and $f_{w_i}$ represent the maximum frequency and the streamed keyword $w_i$ frequency, respectively.
    \textcolor{black}{Backward privacy (\textbf{BP})} contains type-1/2/3 level (\uppercase\expandafter{\romannumeral1}/\uppercase\expandafter{\romannumeral2}/ \uppercase\expandafter{\romannumeral3}), including the special case
    $\rm \uppercase\expandafter{\romannumeral2}^A,\rm \uppercase\expandafter{\romannumeral3}^R,\rm \uppercase\expandafter{\romannumeral2}^D$ extended from \cite{bost2017forward}.
    }\vspace{-3mm}
    \label{Distinct Search}
    \scalebox{0.67}{
    \begin{tabular}{|p{4.2cm}|c|c|c|c|c|c|c|c|c|}
        \hline
         \multirow{2}{*}{\textbf{Schemes}}  & \textbf{Distinct} &\multicolumn{3}{c|}{\textbf{Computation}} & \multicolumn{3}{c|}{\textbf{Communication}} & \multirow{2}{*}{\textbf{Client Storage}} & \multirow{2}{*}{\textbf{BP}}\\
          & \textbf{Search} &\textbf{Search} & \textbf{Update} & \textbf{Padding} & \textbf{Search} &  \textbf{\textcolor{black}{Update}}& \textbf{\textcolor{black}{Comm. Round}} & & \\
         \hline
         \textbf{MONETA}\cite{bost2017forward}  & \XSolidBrush &$O\left(a_w\log N+\log ^3N\right)$ & $O\left(\log ^2N\right)$ & $\checkmark$ &$O\left(a_w\log N+\log ^3N\right)$ &  $O\left(\log ^3N\right)$ & 3 & $O\left(1\right)$ & \uppercase\expandafter{\romannumeral1}\\
         \hline
         $\textbf{DIANA}_{del}$\cite{bost2017forward}  & \XSolidBrush & $O\left(a_w\right)$ & $O\left(\log a_w\right)$ & $\checkmark$ & $O\left(n_w+d_w\log a_w\right)$  & $O\left(1\right)$ & 2 & $O\left(W\log F\right)$ & \uppercase\expandafter{\romannumeral3}\\
         \hline
         \textbf{JANUS}\cite{bost2017forward}  & \XSolidBrush & $O\left(n_w \cdot d_w\right)$ & $O\left(1\right)$ & $\checkmark$ & $O\left(n_w\right)$  & $O\left(1\right)$ & 1 & $O\left(W\log F\right)$ &  \uppercase\expandafter{\romannumeral3}\\
         \hline
         \textbf{JANUS++}\cite{sun2018practical}   & \XSolidBrush & $O\left(n_w \cdot d_{max}\right)$ & $O\left(d_{max}\right)$ & \checkmark & $O\left(n_w\right)$  & $O\left(1\right)$ & 1 & $O\left(W\log F\right)$ & \uppercase\expandafter{\romannumeral3}\\
         \hline
         \textbf{AURA}\cite{sun2021practical}   & \XSolidBrush & $O\left(n_w\right)$ & $O\left(1\right)$ & \checkmark & $O\left(n_w\right)$  & $O\left(1\right)$ & 1 & $O\left(W \cdot d_{max}\right)$ & $\rm \uppercase\expandafter{\romannumeral2}^A$ \\
         \hline
         \textbf{MITRA}\cite{ghareh2018new}   & \XSolidBrush & $O\left(a_w\right)$ & $O\left(1\right)$ & \checkmark & $O\left(a_w\right)$  & $O\left(1\right)$ & 2 & $O\left(W \log F\right)$ & $\rm \uppercase\expandafter{\romannumeral2}$\\
         \hline
         \textbf{SD}$\mathbf{_a}$\cite{demertzis2020dynamic}   & \XSolidBrush & $O\left(a_w+ \log N\right)$ & $O\left(\log N\right)$ & \checkmark & $O\left(a_w+ \log N\right)$  & $O\left(\log N\right)$ & 2 & $O\left(1\right)$ & $\rm \uppercase\expandafter{\romannumeral2}$\\
         \hline
         \textbf{SD}$\mathbf{_d}$\cite{demertzis2020dynamic}   & \XSolidBrush & $O\left(a_w+ \log N\right)$ & $O\left(\log^3 N\right)$ & \checkmark & $O\left(a_w+ \log N\right)$  & $O\left(\log^3 N\right)$ & 2 & $O\left(1\right)$ & $\rm \uppercase\expandafter{\romannumeral2}$\\
         \hline
         \textbf{ROSE}\cite{Xu2022rose}  & \XSolidBrush & $O\left(\left(n_w+s^\prime_w +1 \right)d_w\right)$ & $O\left(1\right)$ & \checkmark & $O\left(n_w\right)$  & $O\left(1\right)$ & 2 & $O\left(W\log F\right)$ & $\rm \uppercase\expandafter{\romannumeral3}^R$ \\
         \hline
         \hline
        \textbf{SEAL}*\cite{demertzis2020seal}   & \XSolidBrush & $O(x\cdot n_w^*)$ & \XSolidBrush & $O(x\cdot N)$ & $O(x\cdot n_w^*)$  & \XSolidBrush & 1 & $O(x\cdot N)$& \XSolidBrush\\
        \hline
        $\textbf{ShieldDB's DSE}^{\star}$\cite{vo2021shielddb}   & \XSolidBrush & $O\left(n_r+n_b\right)$ &  $O\left(n_r+n_b\right)$ & $O(|B|\left(f_w-f_{w_i}\right))$ & $O\left(n_r+n_b\right)$  & $O\left(n_r+n_b\right)$ & 1 & $O\left(|B|\right)$ & --\\ 
         \hline
         \hline
         \textbf{BF-SRE} (Ours)  & \Checkmark &  $O\left(n_w\right)$ & $O\left(1\right)$ & \checkmark & $O\left(n_w\right)$  & $O\left(1\right)$ & 1 & $O\left(W \cdot d_{max}\right)$ & $\rm \uppercase\expandafter{\romannumeral2}^D$\\
         \hline
    \end{tabular}
    }
    
    \scalebox{0.64}{
    \begin{tabular}{l}
    *: \textsf{SEAL} only focuses on static database (indicating it cannot support update operations) while providing values retrieval with keyword. \textcolor{black}{In \textsf{SEAL} we store the client's states locally.}\\
    $\star$: \textsf{ShieldDB} is a dynamic document database. Its deletion is done through re-encryption, which does not clearly specify the backward privacy.\\
    \end{tabular}}
    \vspace{-15pt}
\end{table*}

\textbf{A new perspective - starting from \textit{distinct search}.}
In line with the design of secure EDBs \cite{fuller2017sok}, we should minimize the leakage during search queries. 
It is natural to see that the SQL syntax `SELECT DISTINCT'\footnote{In ISO/IEC 9075-2:2016, the `DISTINCT' predicate can apply with other syntax like `SELECT', `JOIN', `GROUP BY' to retrieve distinct values from a database table.} can be used to eliminate duplicated search responses, which, to some extent, reduces the volume leakage.  
But we cannot simply apply distinct search in EDBs due to the fact that it still leaks information to the server.  
In the context of SQL queries, a distinct search retrieves distinct (unique) values from the corresponding columns in a database table, which allows us to eliminate "duplicate" values. 
The search process traverses only one copy of each value and ignores its duplicates. 
However, if the number of those "ignored" values (i.e., the number of duplicates) is leaked, it can reveal information about the quantities of all relevant values, leading to volume leakage. 
To mitigate this threat, it is crucial to design a mechanism that prevents the server from being aware of the repetition of values. 
{We find that we can use the dummy tags to group all relevant values into an unsearchable dataset so that the server, in the Search stage, can only "see" a single copy of results, thereby concealing the volume of repetitive values and reducing the potential for volume leakage. }

\textit{A trivial solution to transform DSE for volume-hiding.}\label{trival_sol} Recall that DSE focuses on searching file\textcolor{black}{-}identifiers instead of values.
It may not be straightforward to replace the distinct identifiers with the unique/repeatable values.
Here, we attempt to provide a \textit{trivial transformation} to make DSE support "distinct search." \\
$\bullet$ In the Setup or Update stage, the client maps each value to a unique identifier and then initializes or updates the keyword-value pairs by a DSE instance. \\
$\bullet$ To perform the distinct search, the client utilizes the query protocol of the DSE to retrieve these identifiers. Afterward, the client restores the values based on the mapping of the identifiers and then extracts the distinct values locally.

{We see that this transformation necessitates the client to maintain a local mapping table for translation between identifiers and values during each search. 
Furthermore, its security and efficiency heavily rely on the underlying DSE instance. 
We note that a regular DSE scheme could still be vulnerable to volume leakage {\color{black}(e.g., \cite{xu2023leakage})}, or it applies expensive padding strategies incurring significant bandwidth and storage costs {(e.g., \cite{vo2021shielddb,demertzis2020seal})}.  
This observation indicates that a trivial transformation could not be the best solution for both security and efficiency. 
}

In response to this, we propose the \emph{Distinct Dynamic Searchable Symmetric Encryption (\textit{d}-DSE)} that enables clients to securely search for distinct values with volume-hiding. 
In this concept, we make use of the Distinct State, the Distinct Classifier, and the Distinct Constraint to determine distinct values and eliminate the volume difference. 
{
\textcolor{black}{We instantiate a \textit{d}-DSE scheme \textsf{BF-SRE} for non-interactive deletion and efficient search.}
We present a brief comparison between \textsf{BF-SRE} and related \textcolor{black}{(D)SE schemes} in Tab.  \ref{Distinct Search}.

\textit{Extension from "distinct" to "diverse" search.} 
We notice that many works have demonstrated the feasibility of incorporating DSE to support secure queries in SQL syntax, such as keyword\cite{demertzis2020seal}, range\cite{faber2015rich}, and join  \cite{hahn2019joins,shafieinejad2022equi}. 
Fortunately, we find out that the distinct search can be integrated into secure SQL as well.
{\color{black}
Our core idea is to process a query by first obtaining the \textit{distinct values} using the distinct search and subsequently restoring  the quantity of each distinct value (i.e., \textit{value's quantity}) by small client storage.
Please see Sec. \ref{AppinEDB} for more details. 
}
}

The above descriptions briefly introduce our technical roadmap from the beginning by distinct search, constructing secure distinct search (i.e., \textit{d}-DSE), and fulfilling SQL search for EDBs. 
We summarize key contributions as follows. \\
$\bullet$ \textbf{Definition and models of \textit{d}-DSE}.
We underline the definition and models for secure distinct search. 
In particular, we propose a so-called Distinct with Volume-Hiding (DwVH) security that captures how a secure distinct search can mitigate volume leakage. 
We formalize the EDB context of distinct search, analyze security risks, and define the advantage in the models. 
\\
{\color{black}
$\bullet$ \textbf{\textit{d}-DSE designed EDB}. 
We leverage \textit{d}-DSE to build EDBs with volume-hiding property.
We first illustrate the EDB system enhanced by the \textit{d}-DSE query model.
Then we present three new constructions: d-KW-\textit{d}DSE (for distinct keyword queries), KW-\textit{d}DSE (for keyword queries), JOIN-\textit{d}DSE (for join queries), and the corresponding row addition and deletion operations.
We comprehensively expose the leakage functions \cite{DBLP:conf/ccs/CashGPR15,jutla2022efficient} under these constructions and interpret the pattern leakage from the EDB perspective. 
Finally, our analysis clarifies that they resist both SOTA passive \cite{xu2023leakage} and active \cite{zhang2023high} volumetric attacks.
}
\\ 
$\bullet$ \textbf{Concrete \textcolor{black}{scheme} for \textit{d}-DSE}. 
To handle the distinct feature \textcolor{black}{in \textit{d}-DSE designed EDB}, we utilize the Bloom Filter (\textcolor{black}{BF}) \cite{sun2021practical}, producing Distinct State to control the tag generation for values. 
In this sense, we can effectively match the distinct values in the search. 
Based on the above, we propose our \textcolor{black}{scheme}, called \textsf{BF-SRE}, by applying Symmetric Revocable Encryption (SRE). 
The scheme enables the client to process the deletion locally, and thus, it is non-interactive and flexible to row update in encrypted databases.
We also provide formal security analysis to demonstrate that the \textcolor{black}{scheme} satisfies the \textcolor{black}{FP\&BP} as well as the DwVH security. 
\\
$\bullet$ \textbf{Evaluations}.
{\color{black}
We extensively perform evaluations on our proposals and prior work under diverse real-world datasets, including the Crime reports, Wikipedia, and the Enron email datasets. 
Concretely, under equivalent security parameters, we compare \textsf{BF-SRE} with the SOTA \textsf{MITRA}*\cite{ghareh2018new}, \textsf{AURA} \cite{sun2021practical}, \textsf{SEAL}\cite{demertzis2020seal}, and \textsf{ShieldDB}\cite{vo2021shielddb} in terms of the time and communication costs.
Our evaluation demonstrates that \textsf{BF-SRE} stands as a competitive solution  compared to the aforementioned schemes. 
For example, on the Enron dataset, its costs of time and communication for searching the highest-volume keyword are 8.25s and 83.82KB, outperforming others by a factor of approximately 29.27x and 30.54x, respectively. 
}

\section {Backgrounds}

We first describe the encrypted database that supports recording the keyword/value pairs.
\textcolor{black}{In this context, we introduce the parties in distinct search with the corresponding interactions.}
After that, we propose the threat in distinct search.

\subsection{Encrypted Database Description}\label{EDB_mem}
{%
We use EDB to represent an encrypted database in the context of DSE.
This database stores repeatable keyword/value pairs $\left(w, v\right)$ 
{\color{black}
(e.g., for relational table structure, an attribute's value in one column as a keyword with a foreign-key in another as a value)
}
in sequence and supports secure addition, deletion, and search operations.
Specifically, the client can search keywords to retrieve the matching values, add new keyword/value pairs, and delete all specified repetitive $\left(w,v\right)$ pairs.
Note that the search operation incorporates the distinct search predicate {\color{black} $TypeDB\left(w\right)=\left\{v_1,...,v_n| v\in (w,v)\ and\ \forall\ i,j \in [1,n]\ s. t.\ v_i=v_j \iff i=j \right\}$} to retrieve all matched distinct values given a keyword $w$.
\textcolor{black}{We call the number of distinct values in $TypeDB\left(w\right)$ as the \textit{value's type}.}
}

\subsection{Parties}

\textcolor{black}{Two parties involve in the distinct search: } \\
    $\bullet$ Client: The client initializes the EDB by security parameter and outsources it to the server.
    The client sends an update token generated by the input $\left(w,v,op = \{add,del\} \right)$ to add or delete the keyword/value pairs $\left(w,v\right)$ on the EDB.
    The client sends the search token generated by the keyword $w$ to find the corresponding distinct values. \\
    $\bullet$ Server: The server hosts the EDB outsourced by the client.
    The server updates the EDB by the update token.
    In processing a search token, the server searches the distinct values corresponding to the keyword $w$ and returns the result to the client.

\subsection{The Threat in Distinct Search}

{
\color{black}  A scheme for secure distinct search should protect the client's outsourced EDB against the probabilistic polynomial time (PPT) adversary from the \textit{passive}\cite{xu2023leakage,oya2021hiding,oya2022ihop} and \textit{active}\cite{zhang2023high} attacks.
Akin to \cite{chen2023power}, we propose the adversaries' observations during various client events that could reveal volumetric information for the aforementioned attacks.
}
\begin{table}[htb]
    \centering
    \caption{The adversaries' passive and active observations.}
    \vspace{-5pt}
    \scalebox{0.65}{
    \centerline{
    \begin{tabular}{|l|l|l|}
        \hline
        \textbf{Client Event}   & \textbf{Passive observations} & \textbf{Active observations}\\
        \hline
        \multirow{2}{*}{\makecell[l]{\textbf{Setup}:\\ 1.Set up Secret Key;\\2.Outsource the EDB. }} & \multicolumn{2}{c|}{ \makecell[l]{1.The initial EDB.}}\\
        \cline{2-3} & \makecell[l]{1.Prior volume knowledge;} &  \makecell[l]{1.The deceptive set of\\keyword/value pairs;\\2. Partial queries volume \\ before injection.}\\
        \hline
        \multirow{2}{*}{\makecell[l]{\textbf{Update:} \\ 1.Update ciphertexts\\in the EDB. }} & \multicolumn{2}{l|}{ \makecell[l]{1.The updated EDB.}} \\
        \cline{2-3} & \makecell[l]{1.Updated queries \\ in ciphertexts.} & \makecell[l]{1. Updated queries\\generated  from \\ deceptive set.} \\
        \hline
        \makecell[l]{\textbf{Search:} \\ 1.Require distinct search.} & \makecell[l]{ 1.Client search queries; \\2.Search process;\\3.The volume of search \\ queries.} & \makecell[l]{ 1. The after-injection \\volume of queries.}\\
        \hline
    \end{tabular}
    }}
    \label{Observe}
\end{table}

From Tab. \ref{Observe}, we state that the philosophy of volume-hiding countermeasure is to decouple the volumetric relationship speculated from the Setup and Search stage.
The passive and active adversaries first get the baseline of volume through prior knowledge (e.g., outdated keyword frequency\cite{oya2021hiding}) and partial queries observed before injection, respectively.
Then they both consult the observation in the Update stage to achieve maximum fit in the Search stage, such as minimizing objective
functions\cite{xu2023leakage,oya2022ihop} and manufacturing binary volume\cite{zhang2023high}.
To mitigate the fitting, padding strategies produce uniform volume through dummy data at initialization.
Our perspective is to prevent the search leakage containing predicable volume distribution, which is our security goal for distinct search.

We also require Forward Privacy and Backward Privacy (FP\&BP)\cite{bost2016ovarphiovarsigma,bost2017forward} to protect the update  and distinct search function from serious privacy disclosure.
To capture the leakage from distinct search, the security for volume leakage\cite{patel2019mitigating} should be defined.

\section{Distinct DSE}
We define Distinct Dynamic Searchable Symmetric Encryption (\textit{d}-DSE){\color{black}, which includes its scheme definition and security guarantees, to tackle the distinct search threat.}

\subsection{Notations}
$\lambda\in\mathbb{N}$ is the security parameter.
$r\stackrel{\$}{\gets}\mathcal{R}$ means randomly sampling $r$ from the space $\mathcal{R}$.
$a||b$ is the concatenation between $a$ and $b$.
$\left\{0,1\right\}^\lambda$ is a $\lambda$-bit length string.
$\left\{0,1\right\}^*$ is an arbitrary-bit length string.
$F$ is a secure Pseudo Random Function (PRF).
\textcolor{black}{$\Sigma$ is the Dynamic Searchable Symmetric Encryption (DSE) scheme.}
The frequently used notations \textcolor{black}{\&} concepts, the symmetric encryption, Bloom Filter (BF)\cite{sun2021practical}, and Symmetric Revocable Encryption (SRE)\cite{sun2021practical}, are introduced in Appendix \ref{Ntandtools}.

\subsection{The \textit{d}-DSE Scheme}

{\color{black}
\begin{definition}[The Distinct DSE] \label{def_scheme_DDSE}
\itshape
A Distinct Dynamic Searchable Encryption is a triple $\left(\textsf{Setup}, \textsf{Search}, \textsf{Update}\right)$ consisting of three protocols:
    \begin{enumerate}[$\bullet$]
    \item $\textsf{Setup}\left(\lambda\right)$ is a protocol that takes as input the {\rm security parameter} $\lambda$. It invokes the {\rm DSE setup protocol $\Sigma.setup$} and generates the \textbf{\rm Distinct State $\sigma^D$}, outputting $K, \mathbf{st}=\{\sigma,\sigma^D\}$ for the client and outsourcing $\mathrm{EDB}$ for the server respectively, where $K$ is the {\rm master secret key}, $\mathrm{EDB}$ is the {\rm encrypted database}, and $\mathbf{st}$ is the {\rm client's internal state}.
    \item $\textsf{Search}\left(K,\mathbf{st},w;\mathrm{EDB}\right)$ is a protocol between the client with inputs $K$, $\mathbf{st}$, and a search query restricted to a {\rm keyword} $w$, and the server with input the $\mathrm{EDB}$.
    It invokes the {\rm DSE search protocol} $\Sigma.search$ with input modified by the Distinct State $\sigma^D$.
    At the end, the client gets a search result set included \textbf{ distinct values} from the $\mathrm{EDB}$.
    \item $\textsf{Update}\left(K,\mathbf{st},op,{\rm in};\mathrm{EDB}\right)$ is a protocol between the client with $K$ and $\mathbf{st}$ as above, and an {\rm operation} $op$  with its input ${\rm in}$, where $op$ is from the set \{add, del\} {\rm (i.e. addition, deletion)} and ${\rm in}$ is parsed as the {\rm keyword/value pair} $\left(w,v\right)$; and the server with input $\mathrm{EDB}$. 
    It invokes the {\rm DSE update protocol} $\Sigma.update$ with input modified by the Distinct State $\sigma^D$.
    At the end, the client updates its internal $\mathbf{st^\prime}$, and the server renews the $\mathrm{EDB}$.
	\end{enumerate}

\textbf{Correctness.}
Except with negligible probability, the \textit{d}-DSE scheme is correct if the \textsf{Search} protocol returns (current) correct results for the keyword being searched.
For the formalism, we follow the case where the client should not delete a keyword with a retrieval value that is not present in $\mathrm{EDB}$.
\end{definition}
}

\textbf{Remark on the \textit{d}-DSE search protocol.}
The retrieval searched by \textit{d}-DSE and DSE is different\cite{bost2017forward}.
\textcolor{black}{In DSE, the client eventually needs to obtain the outsourced data through the file-identifier, and the volume leakage is not considered by Forward Privacy and Backward Privacy (FP\&BP) \cite{xu2023leakage}.
In our context, the client is allowed to search the EDB to retrieve outsourced distinct values by a keyword.
To this end,} we must refine the leakage on retrieving values.

\subsection{Security Notions for Distinct Search}

{\color{black}To address the leakage ($\mathcal{L}_D$) in \textit{d}-DSE, we refer to
the DSE security. 
The security contains FP\&BP \cite{bost2016ovarphiovarsigma,bost2017forward} under the Adaptive Security model \cite{curtmola2011searchable}, but it falls short in preventing volume leakage. 
Therefore, we define the Distinct with Volume-Hiding (DwVH) security. 
Note that we introduce the Sim-adaptive Security model instead of the Adaptive Security model to perceive FP\&BP and the DwVH security.   
}

We assume that the client is \emph{honest} and should prevent the disclosure of sensitive information, e.g., underlying keywords and encrypted data. 
The server, \textcolor{black}{which is \emph{honest-but-curious}, should follow the protocols' instructions yet passively exposes} some information. 
Similar to the $\textit{Real}$ and $\textit{Ideal}$ formulation \cite{cash2014dynamic,bost2017forward}, \textcolor{black}{t}he Sim-adaptive security notions of \textit{d}-DSE are defined as follows. 

\begin{definition}[Sim-adaptive Security of \textit{d}-DSE]
	    {
	    \itshape
		Assume a \textit{d}-DSE scheme is $\mathcal{L}$-adaptively secure iff for all sufficiently large security parameters $\lambda\in\mathbb{N}$ and PPT adversary $\mathcal{A}$, there is a set of efficient simulators $\mathcal{S}$ with a set of leakage functions $\mathcal{L}$ that has:
		\begin{equation*}
        \footnotesize
		\begin{split}
		\left|\mathbb{P}\left[{\rm Real}_\mathcal{A}^{\rm d-DSE}\left(\lambda\right)=1\right]-\mathbb{P}\left[{\rm Ideal}_{\mathcal{A},\mathcal{S},\mathcal{L}}^{\rm d-DSE}\left(\lambda\right)=1\right]\right|
		=negl\left(\lambda\right),
		\end{split}
		\end{equation*}
		
		where the \textcolor{black}{games} ${\rm Real}_\mathcal{A}^{\rm \rm d-DSE}\left(\lambda\right)$ and ${\rm Ideal}_{\mathcal{A},\mathcal{S},\mathcal{L}}^{\rm \rm d-DSE}\left(\lambda\right)$ are:
            \begin{enumerate}[$\bullet$]
			\item Game ${\rm Real}_\mathcal{A}^{\rm d-DSE}\left(\lambda\right)$: the adversary controls the client to run real protocols.
			Firstly, it triggers the protocol \textsf{Setup} and gets the encrypted database $ \mathrm{\mathrm{EDB}}$.
			Secondly, with the parameters of its choice, it adaptively triggers \textsf{Update}, \textsf{Search}, and then obtains the real transcript list $Q = \left(q_1,q_2,..q_n\right)$.
			Finally, it outputs a bit $b$ decided from the real sensitive information $\left({\mathrm{EDB}},q_1,...q_n\right)$.

			\item Game ${\rm Ideal}_{\mathcal{A},\mathcal{S},\mathcal{L}}^{\rm d-DSE}\left(\lambda\right)$: all of the real protocols are replaced by the set of simulators $\mathcal{S}$.
			The adversary first triggers the simulator $\mathcal{S}_\mathsf{setup}$ and obtains the simulated database $ {\mathrm{EDB}}^\mathcal{S}$. 
			Secondly, with chosen parameters, it adaptively triggers $\mathcal{S}_\mathsf{Update}$, $\mathcal{S}_\mathsf{Search}$, and then gets the simulated transcript $Q^\mathcal{S}=(q_1^\mathcal{S},q_2^\mathcal{S},..q_n^\mathcal{S})$.
			Finally, it outputs a bit $b$ decided from the simulated sensitive information $({\mathrm{EDB}}^\mathcal{S},q_1^\mathcal{S},q_2^\mathcal{S},..q_n^\mathcal{S})$.
			\end{enumerate}
		}
\end{definition}

\textbf{Common Leakage Functions.}
Akin to \cite{bost2017forward}, we use the $sp\left(w\right)$ to represent the \textcolor{black}{\textit{search pattern}, i.e., which queries belong to the keyword $w$}, and the
\textsf{UpHist}$\left(w\right)$ to represent the history of all updates on keyword $w$.

\textbf{Forward Privacy \textcolor{black}{(FP)}.} 
As \textit{d}-DSE holds \textcolor{black}{FP}, adversaries should not utilize the previous transcript to imply the newly added information before executing the \textsf{Search} protocol.
Based on FP in \cite{bost2016ovarphiovarsigma}, \textit{d}-DSE must further preserve the value's information.
This requirement prevents adversaries from deducing \textcolor{black}{the equality among values, which will expose} more information than just the distinct values.

\begin{definition}[Forward privacy]
    {
    \itshape 
    An $\mathcal{L}$-adaptively-secure \textit{d}-DSE scheme has forward privacy if its \textsf{Update} leakage function $\mathcal{L}_D^{Upt}$ can be written as:
    \begin{equation*}\mathcal{L}_D^{Upt}\left(w,v,op\right)=\mathcal{L}^\prime\left(op\right),
    \end{equation*}
    where $\mathcal{L}^\prime$ is a state-less function.
    }
\end{definition}

\textbf{Backward Privacy \textcolor{black}{(BP)}.}
As \textit{d}-DSE holds \textcolor{black}{BP}, adversaries should not use new transcripts to imply previous information after executing \textsf{Search}.

{ 
To capture the \textcolor{black}{BP} in \textit{d}-DSE, we first propose the leakage function ${\rm TimeDTS\left(w\right)}$.
${\rm TimeDTS}\left(w\right)$ defines the list of all encrypted values $v$ matching keyword $w$, excluding the deleted and repetitive ones, together with the timestamp $u$ of when they were added in the database:
{\color{black}
\begin{equation*}
        \footnotesize
	    \begin{aligned}
		 {\rm TimeDTS}\left(w\right)&= \{ \textcolor{black}{(u,v)}| (u,w,v,add) \in Q\ and \ \forall\ u^\prime, (u^\prime,w,v,del) \notin Q \\
	     &and\ \forall\ (u,w,v_i,add) \in Q, (u,w,v_j,add) \in Q, \\
          & s. t.\ v_i=v_j \iff i=j\}.\\
	    \end{aligned}
\end{equation*}
}
\label{BP-dDSE}

${\rm Update}\left(w\right)$ defines the leakage of which timestamps have update queries on the keyword $w$.
With the input $w$, ${\rm Update}\left(w\right)$ returns the set of the timestamps, in which each of them represents an update on $w$:
\begin{equation*}
        \footnotesize
	\begin{aligned}
		{\rm Update}\left(w\right)= \{ u\ |\ (u,w,v,add)\in Q\ or\ (u,w,v,del) \in Q \}.\\
	\end{aligned}
\end{equation*}

With the above functions, we define the \textcolor{black}{BP} of \textit{d}-DSE:
\begin{definition}[Backward privacy]
    {\itshape
    An $\mathcal{L}$-adaptively-secure \textit{d}-DSE scheme has backward privacy if its \textsf{Update}, \textsf{Search} leakage functions $\mathcal{L}_D^{Upt}$, $\mathcal{L}_D^{Srch}$ can be written as:
    \begin{equation*}
        \begin{aligned}
        \footnotesize
    \mathcal{L}_D^{Upt}\left(w,v,op\right)&=\mathcal{L}^\prime\left(w,op\right) \\
        \mathcal{L}_D^{Srch}\left(w\right)&=\mathcal{L}^{\prime\prime}\left(sp(w),{\rm TimeDTS\left(w\right)},{\rm Update\left(w\right)}\right),
        \end{aligned}
    \end{equation*}
    where $\mathcal{L}^\prime,\mathcal{L}^{\prime\prime}$ are state-less functions.
    }
\end{definition}

We say that the aforementioned backward privacy is appropriate for \textit{d}-DSE.
That is inspired by the type-2 level of DSE\cite{bost2017forward}. 
Besides, the type-3 level discloses that multiple $\left(w,v\right)$ additions can be associated with the corresponding deletion operations, revealing the number of identical values that may cause volume leakage.
}

\subsection{Distinct with Volume-Hiding Security}
We analyze the \textit{d}-DSE leakage in consideration of volumetric attacks from the server. 
Previous works \cite{patel2019mitigating,wang2022practicalvol} assume that the number of values associated with any single keyword should not be revealed, \textcolor{black}{excluding the maximum volume}.
Accordingly, the leakage in \textit{distinct search} should only infer the distinct values searched by keywords, and that cannot be \textcolor{black}{pairwised} with prior and current volume knowledge.

Inspired by the volume-hiding \cite{patel2019mitigating}, we define the DwVH security, which aims to prevent adversaries from distinguishing the "signatures" of EDB.
We say that a signature is defined as a sequence of uploaded keyword/value pairs $S= \{w,l(w),t(w)\}$, \textcolor{black}{where $l(w)$ and $t(w)$ denote the number of distinct values (i.e., \textit{value's type}) and the sum of value's quantities associated with keyword $w$, respectively}.
Under our security notion, the adversary is allowed to make two signatures \textcolor{black}{$S_0,S_1$}  for the challenger who randomly selects one of them \textcolor{black}{$S_b, b\in \left\{0,1\right\}$} to generate the EDB.
Afterward, the adversary collects leakages from update and search operations.
Finally, it decides which signature is used to construct the EDB.
We say that a \textit{d}-DSE scheme is DwVH secure if this signature (chosen by the challenger) is indistinguishable from the adversaries, \textcolor{black}{i.e., negligible advantages on the leakage from distinct values and \textit{d}-DSE operations to restore volume information.}

{\color{black}
We let $n$ denote the initial total number of keyword/value pairs in the EDB, and 
$s$ represents the step $s= 1,2,...poly(\lambda)$. 
The DwVH security is defined as follows.
}

\begin{definition}[DwVH Security]
    {\itshape
    An $\mathcal{L}$-adaptive secure \textit{d}-DSE is called DwVH secure if for all $n\geq 1$ and all adversary $\mathcal{A}$ that execute at most $s$ steps, the probability that $\mathcal{A}$ outputs 1 in $\mathbf{DwVHGame}_{\mathcal{A}}^\mathcal{L}\left((n,s),0\right)$ is identical to that in $\mathbf{DwVHGame}_{\mathcal{A}}^\mathcal{L}\left((n,s),1\right)$.
    The $\mathbf{DwVHGame}_{\mathcal{A}}^\mathcal{L}\left((n,s),b\right)$ is:
    }
\end{definition}

\floatname{algorithm}{Game}
\begin{algorithm}
{
    \scriptsize    \caption{$\mathbf{DwVHGame}_{\mathcal{A}}^\mathcal{L}\left((n,s),b\right)$:}    \label{DwVHGame}
    \begin{algorithmic}[1]
        \Statex \underline{$Prepare:$}
        \State $\mathcal{A}$ generates two signatures $S_0 = \left\{w,l_0(w),t_0(w)\right\}_{w\in\mathcal{W}}$ and $S_1 = \left\{w,l_1(w),t_1(w)\right\}_{w\in\mathcal{W}}$, such that:
        \Statex \quad
        \vspace{-6pt}
        \Statex \quad $l_0(w)_{w\in\mathcal{W}}<t_0(w)_{w\in\mathcal{W}}\leq n$;  $l_1(w)_{w\in\mathcal{W}}<t_1(w)_{w\in\mathcal{W}}\leq n$; 
        \Statex \quad $l_0(w)_{w\in\mathcal{W}}=l_1(w)_{w\in\mathcal{W}}$; $\Sigma_{w\in\mathcal{W}}t_b(w) = n$
	    \State $\mathcal{A}$ sends $S_0$ and $S_1$ to the challenger $\mathcal{C}$. 
        \State $\mathcal{C}$ computes keyword/value pairs by choosing $t_b(w)$ values for each keyword $w$, and each keyword $w$ corresponds to $l_b(w)$ distinct values.
        \State $\mathcal{C}$ updates all pairs and sends the corresponding leakages $\mathcal{L}_D^{Upt}$ to the adversary $\mathcal{A}$.
    \end{algorithmic}

    \begin{algorithmic}[1]
        \Statex \underline{$Queries:$}
        \State $\mathcal{A}$ adaptively executes $s$ step. In each step, it can perform:
        \Statex \quad$\bullet$ Update query: $\mathcal{A}$ adaptively performs the $s$-th update query about the step's keyword $w_s$, $\mathcal{C}$ computes $\mathcal{L}_D^{Upt}$ and sends it back.
        \Statex \quad$\bullet$ Search query: $\mathcal{A}$ adaptively performs the $s$-th search query about $w_s$, and $\mathcal{C}$  computes and returns $\mathcal{L}_D^{Srch}$.
    \end{algorithmic}

    \begin{algorithmic}[1]
        \Statex \underline{$Guess:$}
        \State $\mathcal{A}$ guess the $b$ input and outputs a bit $b^\prime \in \left\{0,1\right\}$.
    \end{algorithmic}
}
\end{algorithm}

Compared to the volume-hiding definitions \cite{patel2019mitigating}, \textit{d}-DSE can fight against adversaries who enable challengers to adaptively perform update and search operations.
{The aim of the DwVH security is to show that the leakages from the update and distinct search do not reveal which signature is chosen to establish the EDB.
}
In other words, the adaptive adversaries are consistent with our consideration in the \textcolor{black}{FP\&BP} \textit{d}-DSE so that they just leverage the \textit{d}-DSE leakage of update and distinct search.

\section{High Level of \textit{d}-DSE Designed EDB}\label{AppinEDB}
{\color{black}
This section illustrates the practical application of \textit{d}-DSE in bolstering foundational queries in EDB systems. 
We introduce the \textit{d}-DSE query model, which encompasses update, distinct keyword, keyword, and join queries, effectively covering the spectrum of fundamental queries pertinent to encrypted relational databases.
Based on our \textit{d}-DSE security model, we also conduct a leakage analysis to identify potential volumetric attacks for EDB. 
\textcolor{black}{Note that we summarize the frequently used notations in Tab. \ref{notation_tables} (see Appendix \ref{Ntandtools}).}

\subsection{Apply to EDB system}

\begin{figure}[htb]
    \vspace{-10pt}
    \centerline{ \scalebox{1}{\includegraphics[width=8.5cm,height=7cm]{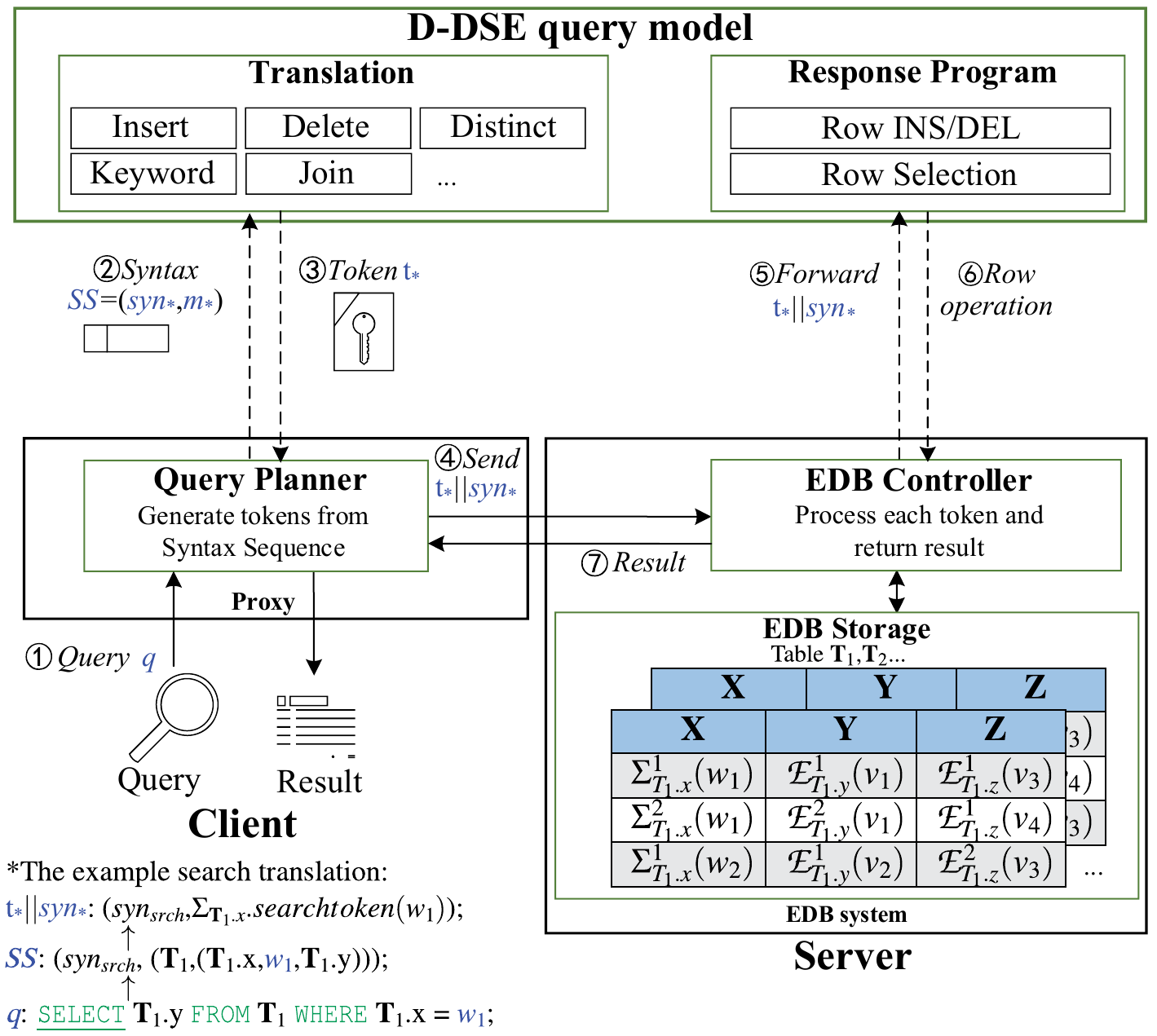}}}
    \vspace{-5pt}
    \caption{{\color{black} The high-level of \textit{d}-DSE design EDB.
    $\Sigma_{\textbf{T}_1.x}^\texttt{c}(w_1)$ and $\mathcal{E}_{\textbf{T}_1.y}^\texttt{c}(v_1)$ denote that the keyword $w_1$ and value $v_1$ in table column $\textbf{T}_1.x$ and $\textbf{T}_1.y$ are encrypted by  $\Sigma$ and $\mathcal{E}$ with a counter $\texttt{c}$, respectively.
    }}
    \label{Apphigh}
    \vspace{-5pt}
\end{figure}

{\color{black}
Like \cite{Popa2011CyptDB,poddar2016arx}, Fig. \ref{Apphigh} depicts that \textit{d}-DSE based EDB leverages the \textit{EDB-proxy} architecture, including:

\begin{enumerate}[$\bullet $ ]

\item  \textbf{EDB Storage \textcolor{black}{($\mathcal{EDB}_S$)}} organizes encrypted data into the table collection $\mathcal{T} = \{\textbf{T}_1,...\textbf{T}_N\}$ \textcolor{black}{\textit{on disk}}. 
Each $\textbf{T}_*$ (e.g., $\textbf{T}_1$.x) is represented by columns $\textbf{T}_*.\star$ and their respective attribute types (e.g., constant-size string, integer, etc.).
Tables include rows, where each row corresponds to the encrypted record $\textbf{r}=(\Sigma_{\textbf{T}_*.\star}^\texttt{c}(w),\mathcal{E}_{\textbf{T}_*.\star}^\texttt{c}(v),...)$ manufactured by the \textit{d}-DSE query model. 

\item  \textbf{\textit{d}-DSE query model} includes a series of \textit{d}-DSE constructions that support the following: 
(1) the \textit{translation} from the functional equivalent syntax $SS = (syn_*,m_*)$ to the token ${\rm t}_*$, where $syn_*$ specifies the construction name, $m_*$ includes the message of $\textbf{T}_*$, $\textbf{T}_*.\star$, and numeric or text data, and ${\rm t}_*$ contains the encrypted data \textcolor{black}{generated from $SS$};
(2) the \textit{response program} for the EDB controller to perform the related construction operations (i.e., row insertion/deletion and selection) on \textcolor{black}{$\mathcal{EDB}_S$} via  the concatenation of the encrypted data and the name ${\rm t}_*||syn_*$.
We construct update, (distinct) keyword, and join query for the \textit{d}-DSE query model (see the following subsections).

\item \textbf{Query Planner} processes queries and forwards results.
For a query $q$, the planner does the following: 
\textcircled{1} extracts it to $SS$ (see the example in Fig. \ref{Apphigh}); 
\textcircled{2} forwards $SS$ to the translation;
\textcircled{3} receives the ${\rm t}_*$;
\textcircled{4} forwards the ${\rm t}_*||syn_*$ to the EDB Controller and waits for response;

\item  \textbf{EDB Controller} processes the ${\rm t}_*||syn_*$ and replies the corresponding results.
\textcolor{black}{Specifically, it follows operations from the response program to update rows in $\mathcal{EDB}_S$ and to copy encrypted data from $\mathcal{EDB}_S$ to the \textit{controller's memory} for selection. We call the replicated encrypted data as $\texttt{EDB}_C$.} 
The EDB Controller works as:
\textcircled{5} forwards ${\rm t}_*||syn_*$ to the response program;
\textcircled{6} follows the operations to update/select rows;
\textcircled{7} returns the operation results.

\end{enumerate}
}

We leverage the EDB-proxy architecture as described in \cite{Popa2011CyptDB,poddar2016arx} to map the original SQL syntax into a series of \textit{d}-DSE tokens. 
This approach enables us to thoroughly analyze potential leakages within \textcolor{black}{the N-size $SS$ sequence (i.e., $\mathbf{SS} = (SS_1,..,SS_N) $)} through the \textit{d}-DSE query model.
Our \textit{d}-DSE query model can attain a \textit{one-for-all} volume-hiding capability within EDB systems. 
We note that hereafter, we omit the subscript when describing a query in a table (e.g., $\textbf{T}$).

\subsection{Update Queries}

We start with the basic update queries on table $\textbf{T}$, e.g.
\begin{equation*}
    \footnotesize
    \begin{aligned}
    \textcolor{ForestGreen}{\underline{\texttt{INSERT INTO}}} \ \ \textbf{T}\ \ (\textbf{T}.x,\textbf{T}.y)\ \ & \textcolor{ForestGreen}{\texttt{VALUE}}\ (w,v) \\
    & \hookrightarrow SS_{ins} = (syn_{ins},(\textbf{T},(\textbf{T}.x,w,\textbf{T}.y,v))) \\       
    \textcolor{ForestGreen}{\underline{\texttt{DELETE FROM}}} \ \ \textbf{T}\  \textcolor{ForestGreen}{\texttt{WHERE}} \ \textbf{T}.x =  w\  & \textcolor{ForestGreen}{\texttt{AND}} \ \textbf{T}.y =  v \\ &\hookrightarrow SS_{del} = (syn_{del},(\textbf{T},(\textbf{T}.x,w,\textbf{T}.y,v))).        
    \end{aligned}
\end{equation*}

We can convert this syntax $SS_{ins}/SS_{del}$ to the \textit{d}-DSE \textsf{Update} protocol, which works for \textbf{T} (initialized by \textit{d}-DSE \textsf{Setup} protocol) and receives keyword, value, and `INSERT'/`DELETE' command as the $(w,v,op)$ input from the Query Planner.
{\color{black} We generate the update token through the \textsf{Update} protocol to append/delete the encrypted keyword and value in the corresponding row in \textbf{T}, i.e., $\left(\Sigma_{\textbf{T}.x}(w),\mathcal{E}_{\textbf{T}.y}(v)\right)$, for subsequent distinct keyword queries.}
{\color{black} Note that we use the client's state \textbf{st} to mark deleted rows on $\mathcal{EDB}_S$ and then process deletion in batch for queries on $\texttt{EDB}_C$.  
}

{\color{black} Note that one should not separately analyze the leakage from the N-size update sequence (i.e., $\textbf{op}_N$$=((u_1,op_1)..(u_N,op_N))$, where $u_*$ is the timestamp) since it is further refined by the leakage (Update(w), see Sec. \ref{BP-dDSE}) from search queries in the \textcolor{black}{EDB system} (see Sec \ref{d-KW-dDSE}).}

\subsection{Distinct Keyword Queries} \label{d-KW-dDSE}

A distinct keyword query retrieves distinct values from the column $\mathbf{T}.y$ of table $\textbf{T}$ that corresponds to keyword $w$ in column $\mathbf{T}.x$, i.e: 
\begin{equation*}
    \footnotesize
    \begin{aligned}
    \textcolor{ForestGreen}{\underline{\texttt{SELECT DISTINCT}}}\ \textbf{T}.y\  \textcolor{ForestGreen}{\texttt{FROM}}\  \textbf{T}\  & \textcolor{ForestGreen}{\texttt{WHERE}}\ \textbf{T}.x = w \\ &\hookrightarrow SS_{Dsrch} = (syn_{Dsrch},(\textbf{T},(\textbf{T}.x,w,\textbf{T}.y))).        
    \end{aligned}
\end{equation*}

Like \cite{demertzis2020seal,kamara2018sql}, the \textit{d}-DSE \textsf{Search} protocol can implement $SS_{Dsrch}$ by viewing $w$ in $\textbf{T}.x$ as keywords and data in $\textbf{T}.y$ as values.
$d$-DSE should allocate a local memory as Distinct State $\sigma^D$ to control the \textsf{Update} and \textsf{Search} protocols. 
{\color{black} For generating tokens ${\rm t}_*$ from $SS$, the \textsf{Update} and \textsf{Search} protocols should separately employ the Distinct Classifier and Distinct Constraint to pre-process $SS$ in \textit{translation} (Fig. \ref{Apphigh}):\\
$\bullet $ \textit{Distinct Classifier} identifies whether the $(w,v,op)$ from update queries is represented as distinct in retrieval, and subsequently, it tags the modified input.\\
$\bullet $ \textit{Distinct Constraint} generates the constrained key for distinct search queries that limits the retrieval of distinct values from $\textbf{T}.y$ based on $w$ in $\textbf{T}.x$.
}

We refer to this construction as d-KW-\textit{d}DSE (distinct KeyWord queries from \textit{d}-DSE).
\footnote{\textcolor{black}{We give a detailed design for the $d$-DSE construction in Appendix \ref{High-d-DDSE}.}}

We notice that finding distinct values in a table requires a huge storage cost to store their state (i.e., distinct or repetitive).
At the same time, the rows in a table do not support records that represent delete operations (e.g., the addition \textcolor{black}{\&} deletion list of file-identifiers \cite{ghareh2018new}).
To this end, we apply the Bloom Filter (BF) to minimize local storage costs and the Symmetric Revocable Encryption (SRE) to enable non-interactive deletion. 
Our scheme for d-KW-\textit{d}DSE is elaborated in Sec. \ref{BF-SRE}.

\textbf{Leakage analysis on d-KW-\textit{d}DSE.} 
We represent a N-size sequence of distinct keyword queries as $\mathbf{q}_{dis} = \left(\textbf{i},\textbf{w}\right)$, including the sequence of table name ${\rm i}\in \textbf{i}$ for $\textbf{T}_{\rm i}$ and searched keyword $w_* \in \textbf{w}$, respectively.
Akin to \cite{hahn2019joins}, in the EDB perspective, the leakage function \cite{curtmola2011searchable,DBLP:conf/ccs/CashGPR15,cash2013highly} of d-KW-\textit{d}DSE is:
\begin{equation*}
\footnotesize
\begin{aligned}
\bm{\hat{L}}_{d-KW-\textit{d}DSE}\left(\mathbf{q}_{dis}\right)=[({\rm i}_1,{\rm t}_{w_1},\mathcal{L}_D^{Srch}(w_1)...({\rm i}_N,{\rm t}_{w_N},\mathcal{L}_D^{Srch}(w_N)],\\
\end{aligned}
\end{equation*}
where ${\rm t}_{w_*}$ represents the token from $w_*$.

Similar to \cite{demertzis2020seal}, the replacement of the \textsf{Search} protocol is carried out in a black-box manner.
We gain the ability to detect and examine all the leakages present in the EDB, which stem from the search leakage in \textit{d}-DSE.

To the best of our knowledge, passive\cite{xu2023leakage} and active\cite{zhang2023high} attacks both source patterns revealed from the leakage function.
We propose volumetric information \cite{xu2023leakage} for d-KW-\textit{d}DSE:
\begin{enumerate}[$\bullet $]
    \item \textbf{Update length pattern}, denoted as $ulen(w)=|{\rm Update\left(w\right)}|$ in $\mathcal{L}_D^{Srch}(w)$, outputs the number of updates made on keyword $w$.
	\item \textbf{Distinct response length pattern}, denoted as $drlen(w)=|{\rm TimeDTS}\left(w\right)|$, outputs \textcolor{black}{value's type}  matching keyword $w$.
    \item \textbf{Query equality pattern} (a.k.a search pattern),  denoted as $qeq(w_i,w_j) = \textbf{1}(w_i=w_j)$, indicates whether two queries are targeting the same keyword.
    The predicate $\textbf{1}(*)$ outputs 1 for $w_i=w_j$ and 0  otherwise.
\end{enumerate}

We state that d-KW-\textit{d}DSE does not have the insert/delete length pattern due to the type-2 BP of \textit{d}-DSE.
The file volume and response similarity pattern are both related to file retrieval instead of the {\color{black}constant-size distinct values} protected by FP\&BP.

\textbf{Potential attacks on d-KW-\textit{d}DSE}. 
d-KW-\textit{d}DSE does not possess file volume (size) pattern leveraged by BVA \cite{zhang2023high} in injection attacks.
For LAA, we find it difficult to perform similar volumetric attacks as in DSE\cite{xu2023leakage}.
Based on DwVH security, the value's quantity can express arbitrary  distribution and integrate into the $ulen(w)$ of each keyword.
Meanwhile, $drlen(w)$ does not reveal the sum of the value's quantities matching $w$, leading to the absence of the linear relationship between $drlen(w)$ and $ulen(w)$ \cite{xu2023leakage}.

\subsection{Keyword Queries} \label{KW-dDSE}

The most common keyword query retrieves values from column \textbf{T}.y for a keyword $w$ in \textbf{T}.x:
\begin{equation*}
    \footnotesize
    \begin{aligned}
    \textcolor{ForestGreen}{\underline{\texttt{SELECT}}}\ \textbf{T}.y\  \textcolor{ForestGreen}{\texttt{FROM}}\  \textbf{T}\  & \textcolor{ForestGreen}{\texttt{WHERE}}\ \textbf{T}.x = w \\
    & \hookrightarrow SS_{srch} = (syn_{srch},(\textbf{T},(\textbf{T}.x,w,\textbf{T}.y))).
    \end{aligned}
\end{equation*}

The $SS_{srch}$ is implemented by d-KW-\textit{d}DSE along with certain client computations. 
Specifically, the client allocates a hash table to map the keyword $w$ with the vector $\mathbf{d}$ constructed from $(w,v,op)$ input.
The dimensions of $\mathbf{d}$ record the value's quantities in the value's numerical (or lexicographical for string) order.
We can first obtain the encrypted distinct values on column \textbf{T}.y through d-KW-\textit{d}DSE, decrypt them, and finally restore each value's quantity through $\mathbf{d}$ in their aforementioned order.
The order from d-KW-\textit{d}DSE can also help us to insert/auto-increment/delete a dimension in $\mathbf{d}$ for new/repetitive/deleted $(w,v)$ input when updating \textcolor{black}{encrypted data}. 
We call this construction as KW-\textit{d}DSE (KeyWord queries from \textit{d}-DSE), which is a practical approach with $O\left(W\right)$ local storage cost {\color{black} to clients} \cite{bost2017forward,bost2016ovarphiovarsigma,ghareh2018new}.
{\color{black} Note that this approach enables record recovery by treating duplicated records copied from multiple columns' data as values.}

\textbf{Leakage analysis on KW-\textit{d}DSE.} The KW-\textit{d}DSE and d-KW-\textit{d}DSE are identical except for the hash map on the honest client. 
For a keyword query sequence $\mathbf{q}_{kw} = \left(\textbf{i},\textbf{w}\right)$, the leakage function of KW-\textit{d}DSE is:
\begin{equation*}
{
\footnotesize
\begin{aligned}
\bm{\hat{L}}_{KW-\textit{d}DSE}\left(\mathbf{q}_{kw}\right) =\bm{\hat{L}}_{d-KW-\textit{d}DSE} \left(\mathbf{q}_{kw}\right).\\    
\end{aligned}
}   
\end{equation*}

\textbf{Potential attacks on KW-\textit{d}DSE}. The potential attacks are identical to those on d-KW-\textit{d}DSE, as the equation captures the same leakage function. 

\subsection{Join Queries}\label{Join_app} 

We implement the join query\cite{hahn2019joins} between the foreign-key and primary-key, which is the fundamental query type for relational databases.
A simple join query of two tables $\textbf{T}_1$/$\textbf{T}_2$ on the foreign/primary key $\textbf{T}_1.z=\textbf{T}_2.z$ returns all values on $\textbf{T}_2.y$ from $\textbf{T}_1$ and $\textbf{T}_2$ that agree on $\textbf{T}_1.z=\textbf{T}_2.z \And \textbf{T}_1.x=w$, i.e.,
\begin{equation*}
    \footnotesize
    \begin{aligned}
    \textcolor{ForestGreen}{\underline{\texttt{SELECT}}}\ \textbf{T}_2.y\  & \textcolor{ForestGreen}{\texttt{FROM}}\  \textbf{T}_1  \textcolor{ForestGreen}{\underline{\texttt{JOIN}}}\ \textbf{T}_2\  \textcolor{ForestGreen}{\texttt{ON}}\ \textbf{T}_1.z = \textbf{T}_2.z\ \textcolor{ForestGreen}{\texttt{WHERE}}\ \textbf{T}_1.x = w. \\
    & \hookrightarrow SS_{join} = (syn_{join},(\textbf{T}_1,\textbf{T}_2,(\textbf{T}_1.x,w,\textbf{T}_1.z),(\textbf{T}_2.z,0,\textbf{T}_2.y))).
    \end{aligned}
\end{equation*}

It is straightforward to use multi KW-\textit{d}DSE to capture the $SS_{join}$. 
We encrypt $\textbf{T}_2$ with one KW-\textit{d}DSE $T_2^{KW-\textit{d}DSE}$ between column $\textbf{T}_2$.y and $\textbf{T}_2$.z and $\textbf{T}_1$  with another $T_1^{KW-\textit{d}DSE}$ between $\textbf{T}_1$.z and $\textbf{T}_1$.x.
For the join query, we first invoke $T_1^{KW-\textit{d}DSE}$ to get values $\mathbf{v}$ in $\textbf{T}_1$.z corresponding to $w$ in $\textbf{T}_1$.x. Then we utilize $T_2^{KW-\textit{d}DSE}$ to search values on $\textbf{T}_2$.y by results from $T_1^{KW-\textit{d}DSE}$.
We call this construction as JOIN-\textit{d}DSE (JOIN queries from \textit{d}-DSE).

\textbf{Leakage analysis on JOIN-\textit{d}DSE.} We find that the leakage of JOIN-\textit{d}DSE on the EDB is caused by the process from $T_1^{KW-\textit{d}DSE}$ and $T_2^{KW-\textit{d}DSE}$. 
{\color{black}Since a join query consists of multi KW-\textit{d}DSE processes (referred to the join equation)}, the leakage function of JOIN-\textit{d}DSE for a sequence of join queries $\mathbf{q}_{join} = \left(\textbf{i},\textbf{j},\textbf{w}\right)$ is:
\begin{equation*}
{
\footnotesize
\begin{aligned}
\bm{\hat{L}}_{JOIN-\textit{d}DSE}\left(\mathbf{q}_{join}\right) = \bm{\hat{L}}_{KW-\textit{d}DSE} \left( (\textbf{i},\textbf{w}){||}_{l}^{|\textbf{v}|}(\textbf{j},v_l) \right),\\
\end{aligned}
}
\end{equation*}
where $v_l$ is the value of $l$ dimension in \textbf{v}, \textbf{j} is the sequence of the second table name (i.e., $\textbf{T}_2$), and ${||}_{l}^{|\textbf{v}|}$ denotes concatenating the sequence of each value performing KW-\textit{d}DSE on $\textbf{T}_\textbf{j}$.

Different from $\bm{\hat{L}}_{KW-\textit{d}DSE}$, $\bm{\hat{L}}_{JOIN-\textit{d}DSE}$ has an apparent search sequence between $\textbf{T}_\textbf{i}$ and $\textbf{T}_\textbf{j}$.
To perform KW-\textit{d}DSE on $\textbf{T}_\textbf{j}$, the keyword $w$ in $\bm{\hat{L}}_{JOIN-\textit{d}DSE}$ can query the same $v$, revealing that they have previously searched the same value in $\textbf{T}_\textbf{i}$.

\textbf{Potential attacks on JOIN-\textit{d}DSE}. Although JOIN-\textit{d}DSE resists volumetric attack inherited from KW-\textit{d}DSE, it additionally leaks a `generalized' access pattern \cite{DBLP:conf/ccs/CashGPR15} and reveals the co-occurrence matrix \cite{ning2021leap,lambregts2022val} in {\color{black}certain situations (e.g., the plain dataset is partially exposed)}. 
To reduce the leakage, one could use the oblivious join \cite{Kra2020eobli,chang2022towards} through indistinguishable join access, which heavily depends on Oblivious RAM (ORAM) and Bitonic-Sorted Network. 

}

\section{The BF-SRE Scheme}\label{BF-SRE}

We describe the \textit{d}-DSE \textcolor{black}{scheme \textsf{BF-SRE} based on the Bloom Filter (BF), Symmetric Revocable Encryption (SRE), and the Forward Private DSE (FP-DSE) scheme for d-KW-\textit{d}DSE}. 
Inspired by \textsf{AURA} \cite{sun2021practical}, the retrieval of \textsf{BF-SRE} is exactly the encrypted distinct values, which brings benefits for subsequent token execution. 
The security analysis depicts that \textsf{BF-SRE} \textcolor{black}{attains FP\&BP} and DwVH security. 
The \textcolor{black}{scheme} also achieves non-interactive deletion and the sub-linear search.
{\color{black}

\textbf{Symmetric Revocable Encryption (SRE)\cite{sun2021practical}}: 
(1) \textsf{SRE}.$KGen\left(\lambda\right)$ outputs the master secret key $msk = (sk,D)$ from the security parameter $\lambda$, where $sk$ and $D$ are the secret key and revoke structure, respectively;
(2) \textsf{SRE}.$Enc\left(msk,s,t\right)$ outputs the ciphertext $ct$ from $msk$, message $s$, and the tag $t$;
(3) \textsf{SRE}.$Comp\left(D,t\right)$ compresses $t$ into the revoke structure $D$ and outputs the new one $D^\prime$; 
(4) \textsf{SRE}.$ckRev\left(sk,D\right)$ takes as input $sk$ and $D$ and outputs $sk_R$; 
(5) \textsf{SRE}.$Dec(sk_R,ct)$ decrypts $ct$ via $sk_R$ if the related tag is not in $D$ for $sk_R$.

}

\begin{figure}[htb]
    \centering
    \vspace{-10pt}
    \centerline{\scalebox{0.5}{{\includegraphics[width=18cm,height=5cm]{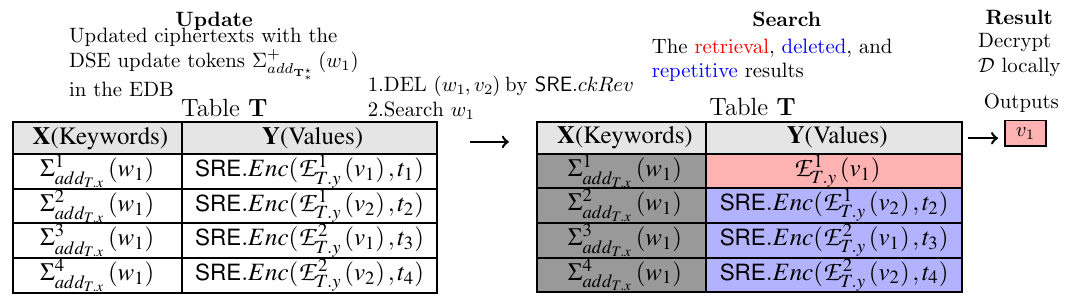}}}}
    \caption{\textsf{BF-SRE}: the scheme overview. \textcolor{black}{We use the blue color to represent the same revocation.}}
    \vspace{-10pt}
    \label{d-DSEHigh}	
\end{figure}

Fig. \ref{d-DSEHigh} shows the \textsf{Update} and \textsf{Search} protocol of \textsf{BF-SRE}.
{\color{black}
\textsf{BF-SRE} uses BF for storing the Distinct State and the FP-DSE $\Sigma_{add}$ to upload modified keyword/value pairs, while the SRE is used to revoke tags associated with deleted or repetitive values in the EDB.
}
In \textsf{Update} protocol, the BF determines the real or dummy tag generation for SRE ciphertexts \textcolor{black}{based on} the first or repetitive inputs, respectively.
The SRE ciphertext is uploaded with real tags if the BF outputs "false." 
Otherwise, \textsf{BF-SRE} uploads the SRE ciphertext with a dummy tag manufactured by revoking the corresponding SRE key.
{In the Search stage, if the SRE ciphertext is unrevoked, it will be decrypted by the constrained sub-key from SRE.} 
Finally, the returned symmetric ciphertexts are decrypted locally.

{\color{black}

The combination of BF and SRE can sufficiently deal with large-scale datasets, even for those with table structures.
According the BF required bit size $b=-n \ln{p}/(\ln{2})^2$, where $n$ and $p$ are the maximum count and the tolerated false-positive probability,
BF can expense a small storage cost to record whether  keyword/value/op pairs appear (e.g., for $n=2^{20}$ and $p=10^{-5}$, $b$ size is just 3MB\cite{sun2021practical}).\label{BFsize}
}

\subsection{BF-SRE Description}

\floatname{algorithm}{Protocol}
\setcounter{algorithm}{0}
\begin{algorithm}[!h]
	\scriptsize
	\label{d-DSE_setup}
	\caption{\textsf{BF-SRE}: Setup, bold lines 2-3 are Distinct State.}
	\underline{\textsf{Setup}$(\lambda)$:}
	\begin{algorithmic}[1]
        \State Initialize the DSE scheme $\left({\rm EDB},\sigma,K_\Sigma\right)\gets\Sigma_{add}.Setup(\lambda)$
        \textbf{
        \State Initialize the Bloom Filter $\mathbf{H},\mathbf{B} \gets \Phi.Gen(\lambda)$
        \State Initialize empty maps \textbf{MSK}, \textbf{UpCnt}, \textbf{C}, \textbf{D}, $\mathbf{\rm EDB_{cache}}$
        }
        \State Randomly generate keys $K_s,K_t,K_c \stackrel{\$}{\gets} \left\{0,1\right\}^\lambda$
        \State Send $\rm EDB$ and the cache $\rm EDB_{cache}$ to the server
	\end{algorithmic}
 \end{algorithm}

\textbf{Setup.}
The \textsf{Setup} protocol generates the encrypted database $\mathrm{EDB}$, its internal state $\sigma$, secret key $K_\sigma$ from a \textcolor{black}{FP-}DSE scheme $\Sigma_{add}$, and constructs secret keys $K_s,K_t,K_c$, the BF hash collection $\mathbf{H}$, the BF bit array $\mathbf{B}$, and four empty maps \textbf{MSK}, \textbf{UpCnt}, \textbf{C}, \textbf{D}.
The $K_c$ is used to encrypt values.
The $K_s$ \textcolor{black}{generates} the cache token $tkn$ to get previous search results from the cache EDB \textcolor{black}{($\rm EDB_{cache}$)}.
The $K_t$ is used to generate tags from inputs and \textbf{UpCnt}.
The \textbf{C} \textcolor{black}{counts} the search on each $w$.
The \textbf{MSK} and \textbf{D} are used to store each master key and revoked key \textcolor{black}{structure associated} keyword $w$, respectively.

 \begin{algorithm}
     \scriptsize
	\label{d-DSE_Update}
	\caption{\textsf{BF-SRE}: Update, bold lines 11-19 are the Distinct Classifier.}
    \underline{\textsf{Update}$\left(K_\Sigma, \textbf{st},op,\left(w,v\right);\mathrm{EDB}\right)$:}
    
	Client:
	\begin{algorithmic}[1]
		\State Read $msk$, $D$, $i$, and $cnt$ from $\textbf{MSK}\left[w\right]$, $\textbf{D}\left[w\right]$, $\textbf{C}\left[w\right]$, and $\textbf{UpCnt}[w]$, respectively.
		\If { $msk$ is not initialized}
			\State Set $msk \gets \textsf{SRE}.KGen(\lambda)$, where $msk = \left(sk,D\right)$
			\State Set \textbf{MSK}[w]$\gets msk$, \textbf{D}[w]$\gets D$
			\State Set $\textbf{UpCnt}\left[w\right] \gets 0$, \textbf{C}[w] $\gets 0$
                \State Set $cnt \gets 1$, $\textbf{UpCnt}[w] \gets cnt$
		\EndIf
		\State Gen. the real/dummy tags $t\gets$ $F\left(K_t,w||v||0\right)$, $l\gets$ $F\left(K_t,w||v||cnt\right)$
            \State {\color{black}Gen. the retrieval $s\gets \mathcal{E}\left(K_c,v||cnt\right)$}
		\If {op == $add$}
                \textbf{
			\If {$\Phi.Check\left(\mathbf{H},\mathbf{B},t\right)$ is false}
				\State Update BF $\Phi.Upd\left(\mathbf{H},\mathbf{B},t\right)$
				\State $ct \gets \textsf{SRE}.Enc\left(msk,s,t\right)$
				\State Insert the EDB $\Sigma_{add}.Update$ $\left(K_\Sigma,add,w||i,\left(ct, t \right);\mathrm{EDB}\right)$
			\Else
				\State $ct \gets \textsf{SRE}.Enc\left(msk,s, l\right)$
				\State Insert the EDB $\Sigma_{add}.Update$ $\left(K_\Sigma,add,w||i,\left(ct, l \right);\mathrm{EDB}\right)$
				\State Puncture the dummy tag $D^\prime \gets \textsf{SRE}.Comp\left(D, l \right)$, $\textbf{D}\left[w\right] \gets D^\prime$
			\EndIf
                }
		\Else
			\State Puncture the real tag $D^\prime \gets \textsf{SRE}.Comp\left(D,t\right)$, $\textbf{D}\left[w\right] \gets D^\prime$
		\EndIf
		\State $\textbf{UpCnt}[w] \gets cnt+1$
	\end{algorithmic}
 \end{algorithm}
 
\textbf{Update.}
This protocol updates the EDB with the new internal state.
At lines 1-7, the client loads the internal state on keyword $w$.
At line 8, the client uses PRF to generate tags $t$ and $l$.
The $t$ is the \textcolor{black}{\textit{real}} tag for the first input of the $w||v$.
\textcolor{black}{The $l$ is the \textit{dummy} tag for the repetitive input.}
At line 9, the client uses the \textcolor{black}{s}ymmetric \textcolor{black}{e}ncryption to generate the retrieval $s$ from the value $v$ and the unique count $cnt$.
At lines 10 and 22, the client chooses to generate ciphertext or revoke the SRE key according to the input \textcolor{black}{operation} ('add' or 'del').
For addition, at line 11, the client executes the Distinct Classifier program to check whether the $w||v$ input appears for the first time.
If so, the client updates the BF, generates SRE ciphertexts $ct$ tagged $t$ as the encrypted ciphertexts, and uploads it by the instance of DSE.
Otherwise,
the client generates dummy SRE ciphertexts, revokes the corresponding revoked key \textcolor{black}{structure} $D$, and uploads the dummy encrypted ciphertexts.
For deletion, at line 21, the client revokes the SRE key \textcolor{black}{structure} corresponding to the tag $t$ of the first input $w||v$. 
At line 23, the client updates the count \textcolor{black}{\textbf{UpCnt}[w]} for keyword $w$.

 \begin{algorithm}
    \scriptsize
    \caption{\textsf{BF-SRE}: Search, bold lines 3-4 are the Distinct Constraint.}
    \underline{\textsf{Search}$\left(K_\Sigma,w,\textbf{st};\mathrm{EDB}\right)$:}
    \begin{algorithmic}[1]
	\Statex Client:
		\State Read $i$, \textcolor{black}{$sk$}, and $D$ from $\textbf{C}\left[w\right]$, $\textbf{MSK}[w]$, and $\textbf{D}[w]$, respectively.
            \State \textbf{if} $i$ is not valid, \textbf{then} $\textcolor{black}{\perp}$.
            \textbf{
		\State ${sk}_R \gets \textsf{SRE}.ckRev\left(sk,D\right)$
		\State Send $\mathbf{\left({sk}_R,D\right)}$ and $\mathbf{tkn = F\left(K_s,w\right)}$ to the server
            }
		\State $msk = (sk,D) \gets \textsf{SRE}.KGen(\lambda)$
		\State Renew $\textbf{C}\left[w\right] \gets i+1$, $\textbf{MSK}[w] \gets msk$, $\textbf{D}[w] \gets D$
	\Statex Client and Server:
		\State Run $\Sigma_{add}.Search\left(K_\Sigma,w||i,\sigma;\mathrm{EDB}\right)$, and the server gets the list $L=$ $\left(\left(ct_1,t_1\right),\left(ct_2,t_2\right)...,\left(ct_l,t_l\right)\right)$
	\Statex Server:
        \State Set the new value list $NV \gets \emptyset$
		\For {$j\in\left[1,l\right]$}
			\State Get the encrypted value $V_j \gets \textsf{SRE}.Dec\left(\left(sk_R,D\right),ct_j,t_j\right)$
			\If {$V_j$ is valid}
				\State Update $NV$ $\gets$ $NV$ $\bigcup V_j$
			\Else
			    \State Delete this ciphertext in EDB 
			\EndIf
		\EndFor
		\State Retrieve cache $OV$ $\gets$ $\rm EDB_{cache}$$\left[tkn\right]$
		\State Return $S \gets NV \bigcup OV$, and update $\rm{EDB}_{cache}$$\left[tkn\right] \gets S$
    \Statex Client:
    \State Extract each symmetric cipher $s$ from $S$, and decrypt the $s$ to get the value $v$ from $v||cnt$ $\gets $ $\mathcal{D}\left(K_c, s\right)$ as the search result
	
	\end{algorithmic}
\end{algorithm}

\textbf{Search.}
The \textsf{Search} protocol finds the ciphertexts by keyword $w$ and returns the distinct values.
At lines 1-4, the client reads \textcolor{black}{internal state}, computes the Distinct Constraint program to get the SRE revoked key $sk_R$ and the cache token $tkn$, and sends them to the server.
At lines 5-6, the client renews the internal state for the next search on $w$.
At line 7, the client and server run the Search protocol of \textcolor{black}{$\Sigma_{add}$}, and then the server gets the search list $L$.
At lines 8-16, the server uses the $sk_R$ to decrypt each $ct$ in list $L$ and remains the symmetric ciphertexts from the unrevoked SRE ciphertext. 
At line 17, the server retrieves previous results from the encrypted database cache $\rm EDB_{cache}$ by $tkn$.
At lines 18-19, the client obtains the search result $S$ from the server and decrypts all symmetric ciphertexts to get the distinct values.

{
\color{black}
\textbf{Remark on BF-SRE}. To simplify the explanation, we illustrate the situation of retrieving distinct values from the specified column $\textbf{T}.y$ based on the searched keyword in $\textbf{T}.x$. 
To retrieve distinct values on other columns in $\textbf{T}$, we do the following: (1) extend the value input as a vector (i.e., v to \textbf{v}); (2) expand the dimension of \textbf{B} (i.e., \textbf{B}[w] to \textbf{B}[w].$\star$) to record more state on other columns; (3) expand the dimension of \textbf{D} (i.e., \textbf{D}[w].$\star$) for their corresponding revoked key structure; and (4)  finally extend $ct$. 
After that, we switch to the right column's revoked key structure in \textbf{D}[w].$\star$ for distinct values on arbitrary columns in $\textbf{T}$. 
More functional extensions for \textsf{BF-SRE} are possible through replacing $\Sigma_{add}$ with other FP-DSE schemes like range search\cite{zuo2018range}.   
}

\subsection{Analysis on BF-SRE}

\textbf{Correctness.} The \textsf{BF-SRE} \textcolor{black}{scheme} uses the \textcolor{black}{FP-}DSE to upload the modified keyword/value pairs in encrypted databases. 
With the polynomial-time algorithm of the PRF $F$, it always outputs a $tkn$ from $K_s$.
The BF's possibility of losing correctness is false positive and acceptable \cite{sun2021practical}. 
Hence, our \textcolor{black}{scheme} can correctly update and search distinct values.

\textbf{Security Analysis.}
\textsf{BF-SRE} securely generates the ciphertexts, search tokens attached to the \textcolor{black}{FP-}DSE, the SRE, the PRF, and the symmetric encryption. 
As for the Update stage, the client modifies the input by \textbf{UpCnt} and the Distinct State, and the ciphertext is uploaded by the FP-DSE $\Sigma_{add}$.
With regard to the Search stage, the client uploads the revoked SRE key to exactly retrieve the values attached to the real tags, which does not show auxiliary information about the volume.

\begin{theorem}[Adaptive Security of \textsf{BF-SRE}]
\label{theorem1}
    {
    \itshape
    Let $F$ (the PRF) with a specific key be modelled as the random oracle $\mathcal{H}_F$, $\mathcal{L}_D=\left(\mathcal{L}_D^{Upt},\mathcal{L}_D^{Srch}\right)$ \textcolor{black}{is defined} as:
			\begin{equation*}
            \footnotesize
			{
			\begin{split}
			\mathcal{L}_D^{Upt}\left(w,v,op\right)&= add \\
			\mathcal{L}_D^{Srch}\left(w\right)&= sp(w), {\rm TimeDTS}(w), {\rm Update}(w),
			\end{split}
			}
			\end{equation*}
			\textsf{BF-SRE} is $\mathcal{L}_D$-adaptively-secure.
    }
\end{theorem}

\emph{Proof Sketch.} The proof is conducted in the game hop like \textsf{AURA} \cite{sun2021practical}.
We gradually replace the real cryptographic tools used in the \textsf{BF-SRE} \textcolor{black}{scheme} with the bookkeeping tables or the corresponding simulators and obtain the \textsf{BF-SRE} simulator.
The addition list $L_{add}$ and the deletion set $D$ are used in the $\mathcal{L}_D^{Srch}$ leakage function to retrieve the search result.
We note that the BF determines the workflow of the simulator and does not affect the indistinguishability between each game hop.

\begin{figure*}[ht]
    \centering
    \subfloat[\textcolor{black}{Crime}]{\label{test_dataA}\scalebox{0.5}{\includegraphics[width=12cm,height=4.8cm]{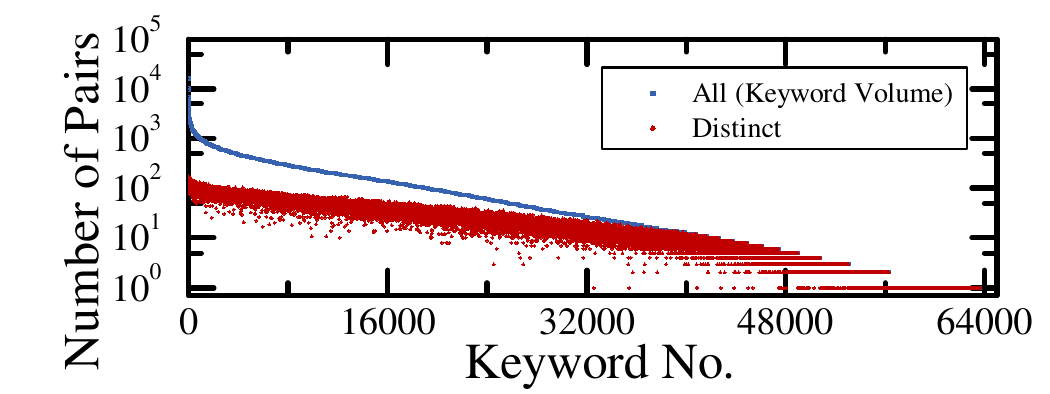}}}
    \subfloat[Wikipedia]{\label{test_dataC}\scalebox{0.5}{\includegraphics[width=12cm,height=4.8cm]{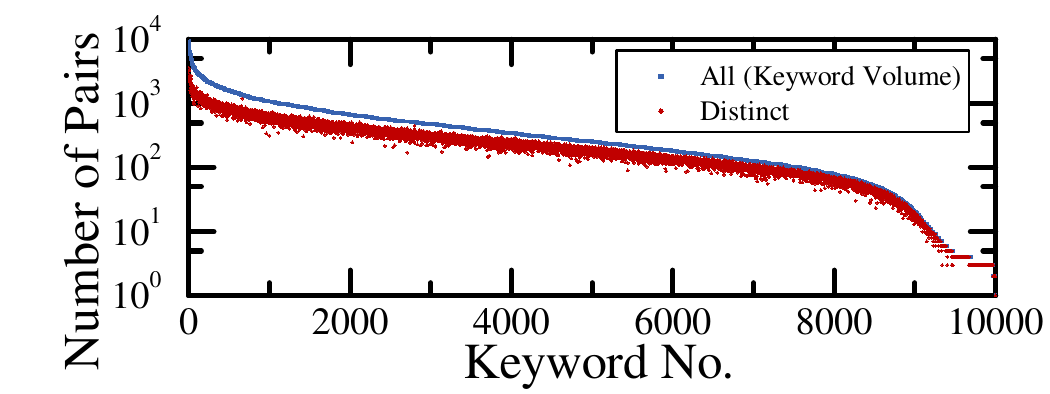}}}
    \subfloat[Enron]{\label{test_dataB}\scalebox{0.5}{\includegraphics[width=12cm,height=4.8cm]{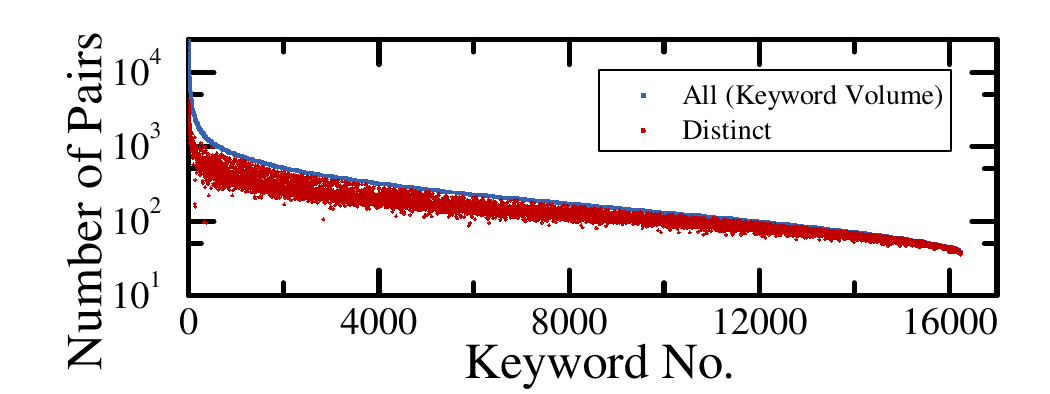}}}
    \vspace{-10pt}
    \caption{The number of keyword/value pairs associated with each keyword on Crime, Wikipedia and Enron dataset. The blue and red nodes represent the number of all and distinct keyword/value pairs, respectively.}
    \vspace{-10pt}
    \label{test_data_scale}
\end{figure*}

\begin{theorem}[DwVH Security of \textsf{BF-SRE}] 
\label{theorem2}
    The leakage function of \textsf{BF-SRE} $\mathcal{L}_D$ = $(\mathcal{L}_D^{Upt}, \mathcal{L}_D^{Srch})$ is Distinct with Volume-Hiding.
\end{theorem}

\emph{Proof Sketch.} The proof is analyzed from the leakage in the \textsf{BF-SRE} forward and backward privacy. 
To prove that $\mathcal{L}_D$ is Distinct with Volume-Hiding, we need to construct two signatures with the same size and the upper bound of maximum volume.
Note that $\mathcal{L}_D^{Upt}$ consists the addition operations, and $\mathcal{L}_D^{Srch}$ includes \textcolor{black}{ value's type and timestamps of updates.}
The $\mathcal{L}_D^{Srch}$ is independent of the input keyword/value pairs because the real and dummy tags generation can always maintain identical response \textcolor{black}{on value's type under the two signatures.}
Therefore, the leakage function reveals nothing about the volume. 
\textcolor{black}{Both proofs are given in Appendix \ref{AdpBS}.} 

\textbf{Complexity.} In the Setup stage, \textsf{BF-SRE} spends a constant time initializing  $\Sigma_{add}$, internal state, and secret keys.
In the Update stage, it consumes a constant time to update the internal state, execute $\Sigma_{add}$, test the BF, and generate ciphertexts.
Hence, the \textsf{Update} protocol computation complexity is $\mathcal{O}\left(1\right)$ with a one-way interaction and an $\mathcal{O}\left(1\right)$ length ciphertext.
In the Search stage, \textsf{BF-SRE} spends a constant time performing $\Sigma_{add}$ and spends $\mathcal{O}\left(a_w\right)$ time on search results, where $a_w$ denotes all ciphertexts found by \textcolor{black}{$\Sigma_{add}$}.
Hence, the \textsf{Search} protocol computation complexity is set as $\mathcal{O}\left(n_w\right)$ with a one-way interaction, where $n_w$ denotes the number of returned distinct values.

\subsection{Optimization}\label{Greedy_OPT}

In \textsf{BF-SRE}, the cryptographic tools determine the actual computation and communication costs.
Recall that in the \textsf{Search} protocol, \textsf{BF-SRE} uses the DSE scheme to find all matched ciphertexts.
A possible solution to improve the efficiency of the \textsf{Search} is to leverage {the parallel Forward Private DSE \cite{kamara2013parallel,kim2017forward}}. 

On the other hand, the SRE decryption affects the efficiency of finding the distinct values. 
In practice, SRE requires significant computational cost to determine the derived sub-key corresponding to SRE ciphertext.
Fortunately, we can identify and implement the \textit{Greedy} and \textit{Default} optimization for the SRE decryption: \\
    $\bullet $ \emph{Greedy}: Temporarily store the sub-keys that are not applied in the previous SRE decryption and try using them in subsequent decryption.
In practice, we can use a stack to store unused sub-keys, enabling a systematic approach to ejecting the keys one by one for the decryption of SRE ciphertexts.
After decryption, the used sub-keys are discarded, while the (new) derivatives are added to the stack. \\
 $\bullet $ \label{Pre_OPT}\emph{Default}: {\color{black}Pre-compute the sub-key space size of SRE decryption through the keyword frequency in the target dataset. In practice, we can utilize the estimated keyword frequencies, which are close to the real ones\cite{vo2021shielddb}, to set the appropriate limitation of revoke operation corresponding to each keyword.
 After that, SRE decryption reduces the total computation of sub-keys, thereby achieving faster decryption speed.\\
 }
Note we combine the two optimizations for \textsf{BF-SRE}.

\section{Evaluation}

{\color{black}
We first introduce datasets with different keyword/value distributions and experimental setup.
Then, we evaluate the comparison between \textsf{BF-SRE} and \textsf{MITRA}* \cite{ghareh2018new}, \textsf{AURA}\cite{sun2021practical}, \textsf{SEAL}\cite{demertzis2020seal}, and \textsf{ShieldDB}\cite{vo2021shielddb}.
Our scheme, under equivalent security parameters, exhibits comparable time costs and significant communication improvement.
Our codes are publicly available in \href{https://github.com/jd89j12dsa/ddse}{https://github.com/jd89j12dsa/ddse}.

\begin{figure*}[htb]
    \centering
        \centerline{
            \subfloat[\textcolor{black}{Crime}]{\label{DatawodA}\scalebox{0.6}{{\includegraphics[width=10cm,height=6.5cm]{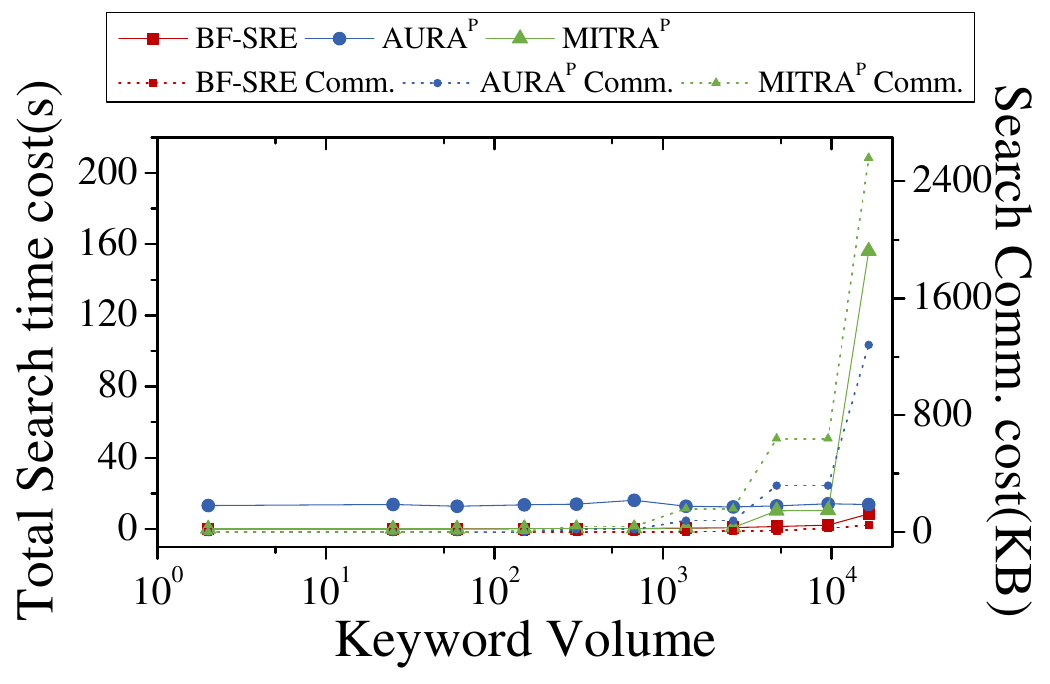}}}}
            \quad
            \subfloat[Wikipedia]{\label{DatawodB}\scalebox{0.6}{{\includegraphics[width=10cm,height=6.5cm]{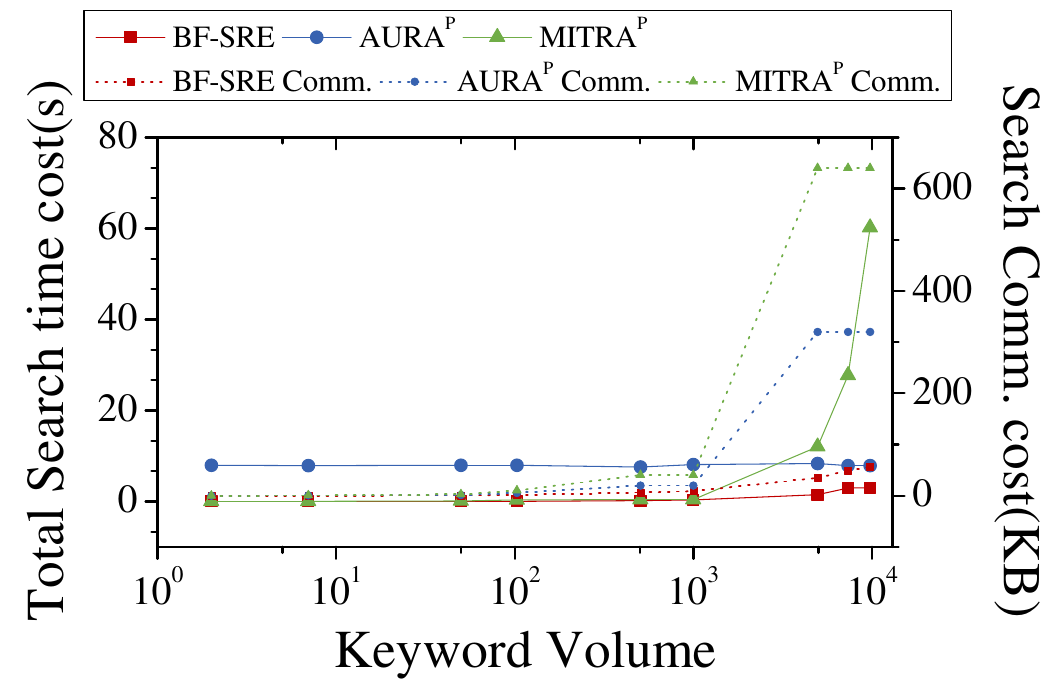}}}}
            \quad
            \subfloat[Enron]{\label{DatawodC}\scalebox{0.6}{{\includegraphics[width=10cm,height=6.5cm]{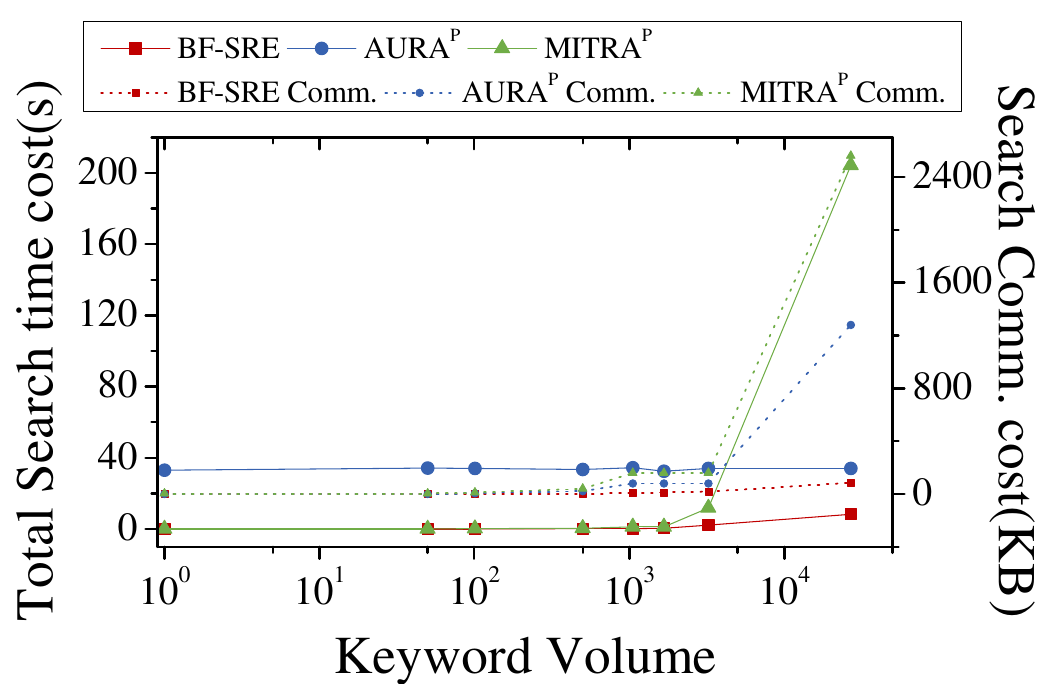}}}}
        }
    \vspace{-5pt}
    \caption{Comparison of the \textbf{total} search \textbf{time and communication} costs of \textsf{BF-SRE}, \textsf{AURA}$^P$, and \textsf{MITRA}$^P$ \textbf{without deletion}.}
    \vspace{-15pt}
    \label{Datawod}
\end{figure*}

\subsection{Dataset}\label{Dataset}

\textbf{Chicago Crimes Reports\footnote{This dataset is available at: \url{https://data.cityofchicago.org/}}}.
The Crime dataset is suitable for keyword/value pairs (in the context of SQL and EDB)\cite{demertzis2020seal,cash2021improved}.
It includes \textcolor{black}{7,989,987} keyword/value pairs extracted from the reports (spanning from 1-1-\textcolor{black}{1999} to \textcolor{black}{2-4-2024}). 
We use street names and the corresponding IUCR codes as keywords and values in the Crime table, respectively. 
For testing join queries, we extract the whole IUCR table to search the PRIMARY\_DESCRIPTION by IUCR codes from the Crime results.

\noindent \textbf{Wikipedia\footnote{The Wikipedia: \url{https://dumps.wikimedia.org/}}}.
This is a large-scale document dataset for DSE evaluation~\cite{chen2023power,kim2017forward}.
Through the Wikipedia extractor and Python NLTK package \cite{kim2017forward}, we collect 4,565,948 pairs \textcolor{black}{in our table structure}, and we aim to find the distinct words by document names. 
}

\noindent \textcolor{black}{
\textbf{Enron}\footnote{Enron email dataset:\url{https://www.cs.cmu.edu/~enron/}.}.
Enron email dataset \cite{DBLP:conf/uss/ZhangKP16,kim2017forward} stores the texts converted from its email system.
We use NLTK package and Porterstemmer \cite{DBLP:journals/program/Porter80} to extract 5,190,199 pairs, aiming to identify distinct words through email names.
}

\textcolor{black}{
Fig. \ref{test_data_scale} illustrates the number of keyword/value pairs associated with each keyword \textcolor{black}{in descending order}.
We count the number of all keyword/value pairs as \textit{Keyword Volume}.
The repetitive-percentage/keyword-space/highest-keyword-volume of the Crime, Wikipedia, and Enron dataset are \textcolor{black}{80.15\%/63,659/16,644}, 42.04\%/10,000/9,738, and 45.53\%/16,241/26,946, respectively.}

\subsection{Experimental Setup}

{\color{black}
Our test platform is $\mathrm{Intel}^\circledR$ Xeon Gold 5120 CPU @ 2.20GHz, 128GB, Dell RERC H730 Adp SCSI Disk Device, Windows Server 2016 Standard.
In our Python (v3.60) code, we utilize the PYmysql package to perform at most 100 parallel queries on MySQL (Ver 14.14 Distrib 5.7.33).

We first compare the performance of \textsf{BF-SRE} with \textsf{MITRA}$^P$ and \textsf{AURA}$^P$. 
Note \textsf{BF-SRE} uses \textsf{DIANA} \cite{bost2017forward} as the Forward Private DSE implementation.
Our implementation makes Python modules from the SRE\cite{sun2021practical} source code with the \textit{greedy} and \textit{default} optimization mentioned in Sec. \ref{Greedy_OPT}.
We also implement \textsf{MITRA}$^P$ and \textsf{AURA}$^P$ from \textsf{MITRA}*\cite{ghareh2018new} and \textsf{AURA}\cite{sun2021practical}, along with the \textsf{SEAL}'s adjustable padding strategy\cite{demertzis2020seal} for fair comparisons.
Specifically, when initializing the EDB, we ensure that \textsf{MITRA}$^P$ and \textsf{AURA}$^P$ pad each keyword with dummy pairs until the volume of each keyword reaches the exponentiation of $x=4$, respectively. 
For correctness, we prepare a translation map between keyword/value pairs and keyword/id pairs
to restore search results mentioned in Sec. \ref{trival_sol}.

Then, we test the time and communication costs of \textsf{BF-SRE}, \textsf{SEAL}, and \textsf{ShieldDB} when processing queries in SQL syntax.
\textsf{SEAL} is implemented through the tree-based ORAM \cite{ghareh2018new}\footnote{The tree-based ORAM:  \url{https://github.com/jgharehchamani/SSE}} 
with the factor $a=20,x=4$. 
For \textsf{ShieldDB} in NL mode ($\alpha = 256$), we transform keyword/value datasets to the related file dataset and further compose keyword clusters uploading the stream data. 
We say that they can perform the queries via our basic Query Planner established from the PYmysql package and the transformation in Sec. \ref{trival_sol}. 

}

\subsection{Compare with DSE}

{
\color{black} 
For \textsf{BF-SRE}, \textsf{AURA}$^P$, and \textsf{MITRA}$^P$, we document total/client search time, communication, and highest-volume keyword search time costs in deletion.
Based on the equation in Sec. \ref{BFsize}, the BF's size reaches an appropriate setting to delete on the highest-volume keyword, preserving the correctness of \textsf{BF-SRE} and \textsf{AURA}$^P$.
}

\subsubsection{Search performance}

{
\color{black}
\textbf{Total search time and communication costs.} Fig. \ref{Datawod} illustrates that \textsf{BF-SRE} significantly outperforms \textsf{AURA}$^P$ and \textsf{MITRA}$^P$ in time and communication costs when the Keyword Volume > $10^3$. 
Particularly, with regards to time costs, in Fig. \ref{DatawodC}, \textsf{BF-SRE} exhibits a slower increase in cost, extending up to 8.25 seconds. 
This is more efficient, with a 4.12x and 29.27x reduction in time compared to \textsf{AURA}$^P$ and \textsf{MITRA}$^P$. 
Meanwhile, the communication costs for \textsf{BF-SRE} are 83.82 KB with 15.27x and 30.54x advantage over \textsf{AURA}$^P$ and \textsf{MITRA}$^P$.
Both compared DSE schemes display significant ladder like mutations (of \textsf{SEAL}) in communication cost.
All trends of Fig. \ref{DatawodA},\ref{DatawodB},\ref{DatawodC} have similar increments.
This result shows that the padding strategy requires \textsf{AURA}$^P$ to initialize the large size BF to log the deletion of dummy pairs, while \textsf{MITRA}$^P$ has to perform more "clean-up" operations (i.e., removing deleted pairs and re-encrypting the remaining pairs).

}

\begin{figure}[!h]
    \vspace{-10pt}
	\centering
        \centerline{
            \subfloat[\textcolor{black}{Crime}]{\label{DataCwodA}\scalebox{0.38}{{\includegraphics[width=8cm,height=7cm]{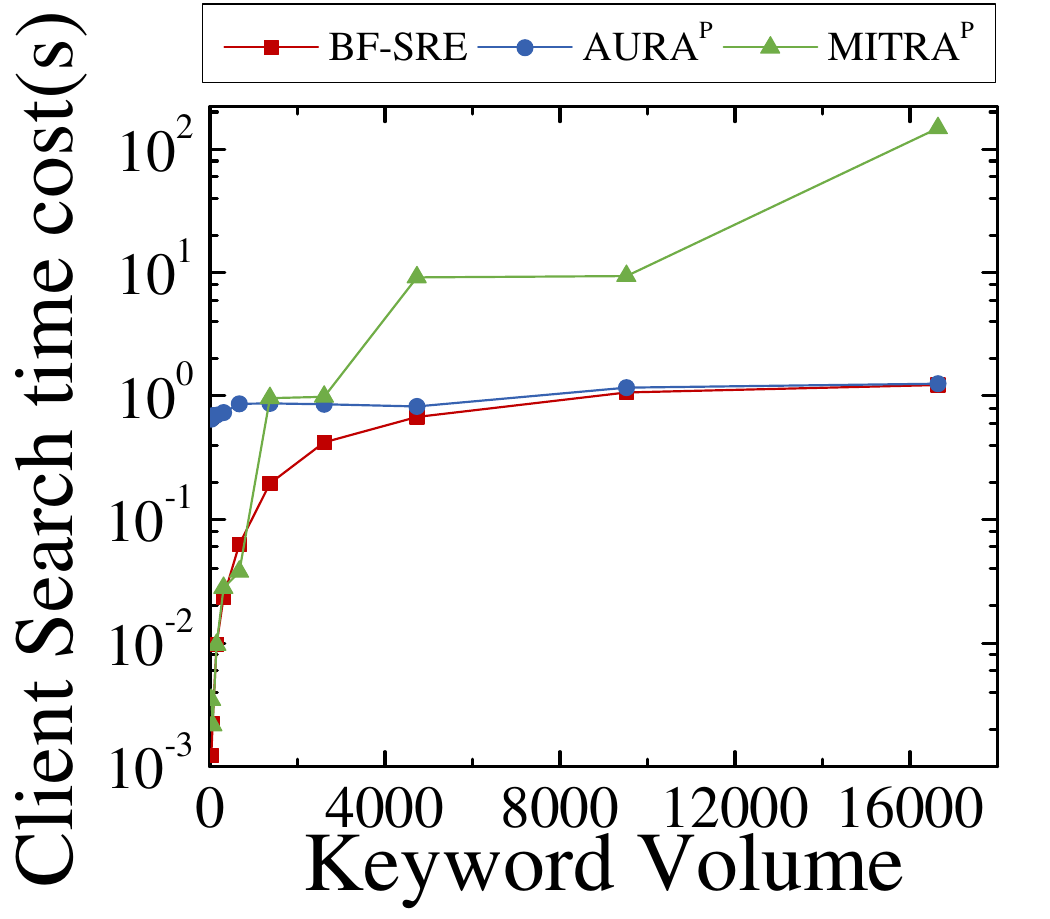}}}}
            \subfloat[Wikipedia]{\label{DataCwodB}\scalebox{0.38}{{\includegraphics[width=8cm,height=7cm]{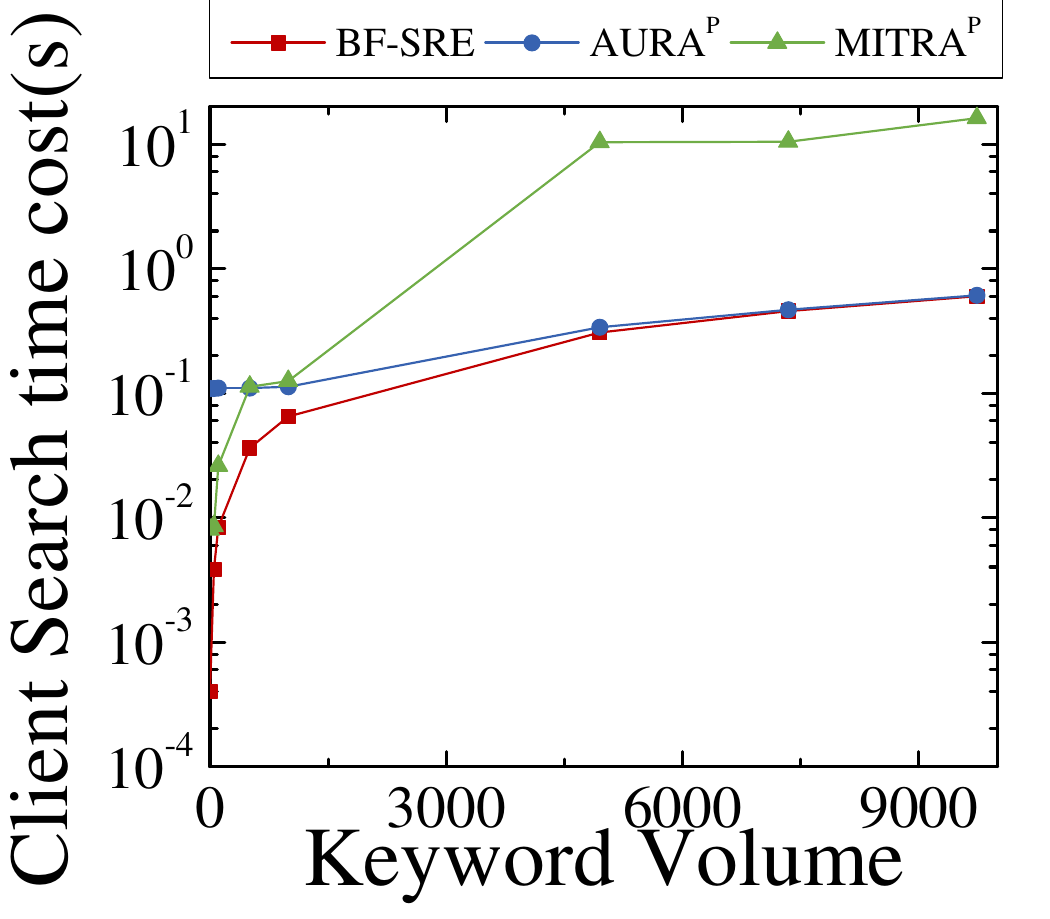}}}}
            \subfloat[Enron]{\label{DataCwodC}\scalebox{0.38}{{\includegraphics[width=8cm,height=7cm]{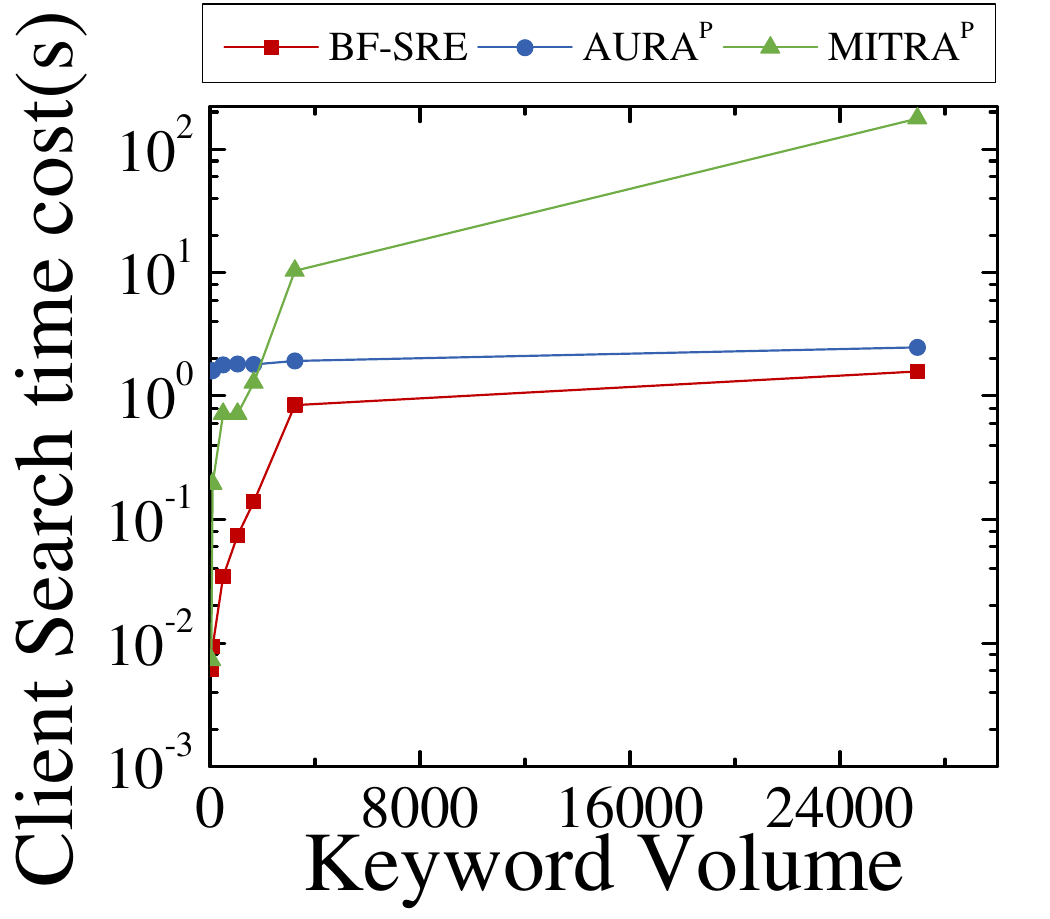}}}}
        }
    \vspace{-5pt}
    \caption{Comparison of search \textbf{time} costs of \textsf{BF-SRE}, \textsf{AURA}$^P$, and \textsf{MITRA}$^P$ on \textbf{client without deletion}.}
    \vspace{-10pt}
    \label{DataCwod}
\end{figure}

{\color{black}
{\color{black}
\textbf{Search time cost on client.} Fig. \ref{DataCwod} shows that \textsf{BF-SRE} outperforms \textsf{MITRA}$^P$ and presents a small advantage over \textsf{AURA}$^P$. 
For example, in Fig. \ref{DataCwodA}, the cost gap between \textsf{BF-SRE} and \textsf{AURA}$^P$ is roughly \textcolor{black}{0.64-0.03}s.}  
We state that \textsf{BF-SRE} does not process dummy data while searching, thereby improving the client's performance.  
}

\begin{figure}[htb]
    \vspace{-10pt}
    \centerline{
            \subfloat[\textcolor{black}{Crime}]{\label{DatawA}\scalebox{0.38}{{\includegraphics[width=8cm,height=7cm]{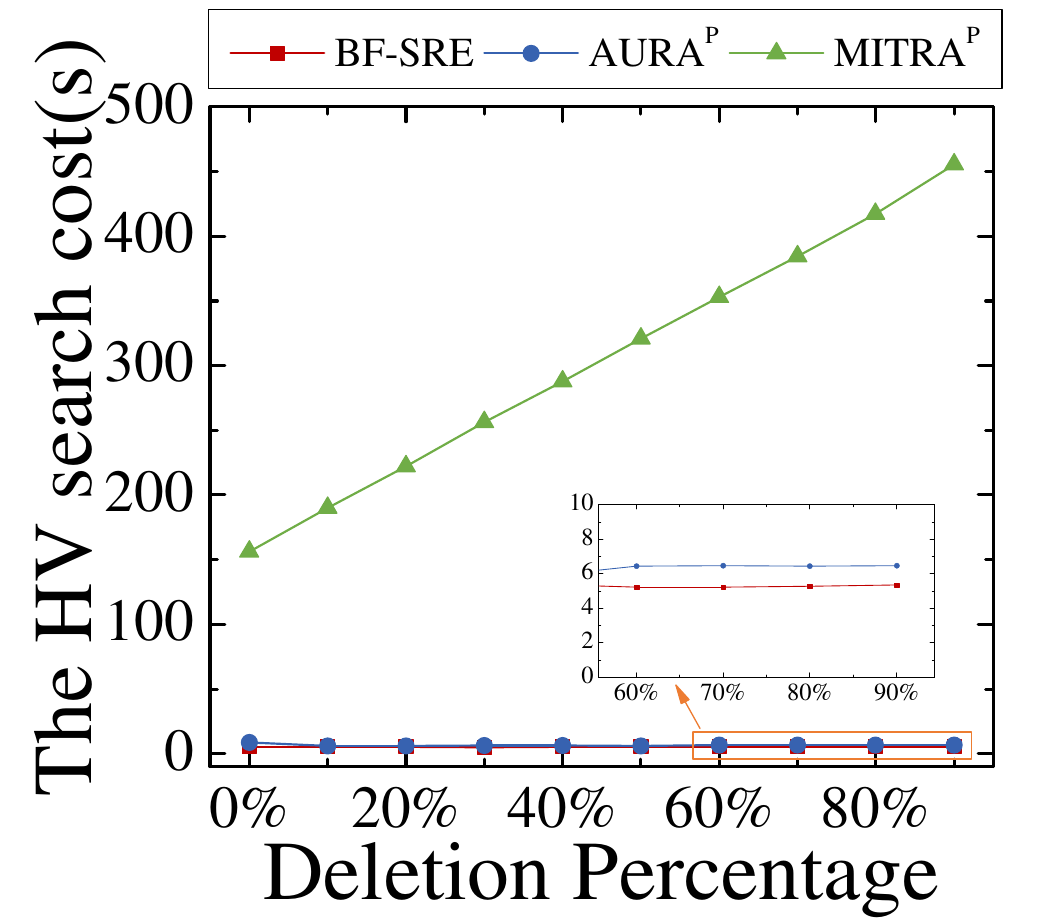}}}}
            \subfloat[\textcolor{black}{Wikipedia}]{\label{DatawB}\scalebox{0.38}{{\includegraphics[width=8cm,height=7cm]{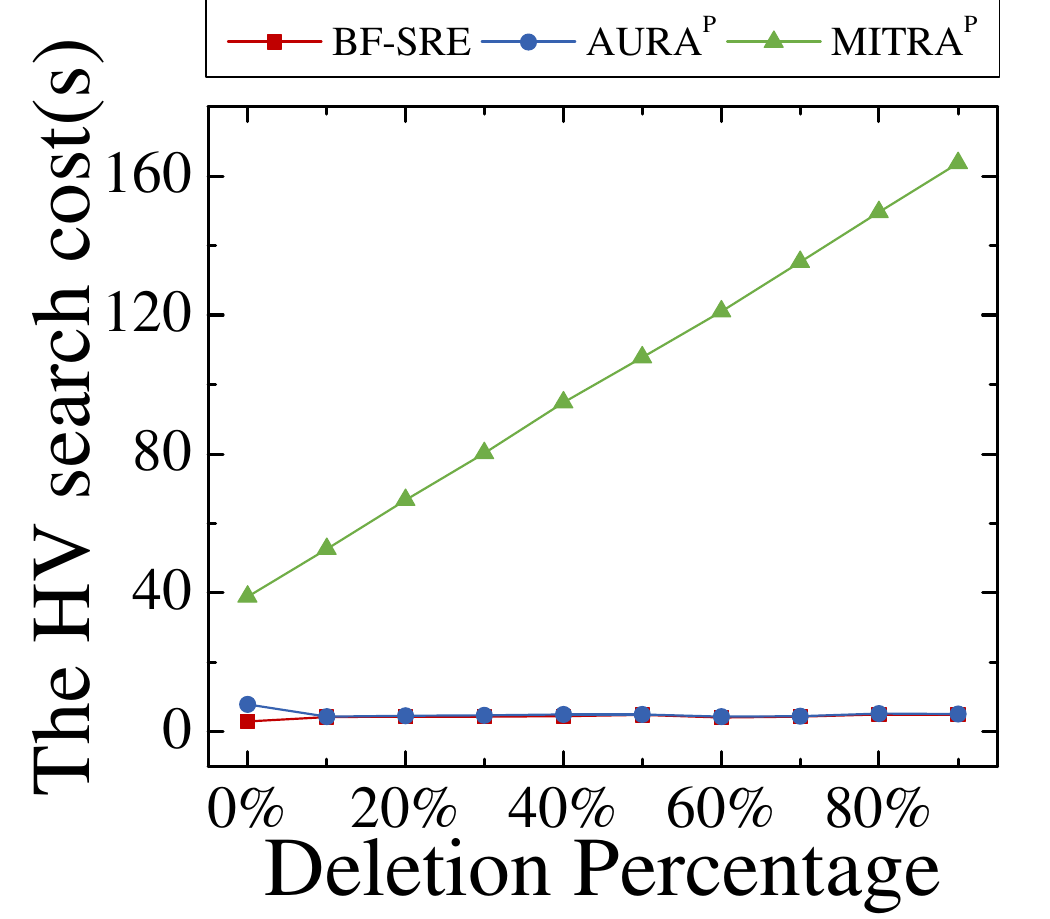}}}}
            \subfloat[\textcolor{black}{Enron}]{\label{DatawC}\scalebox{0.38}{{\includegraphics[width=8cm,height=7cm]{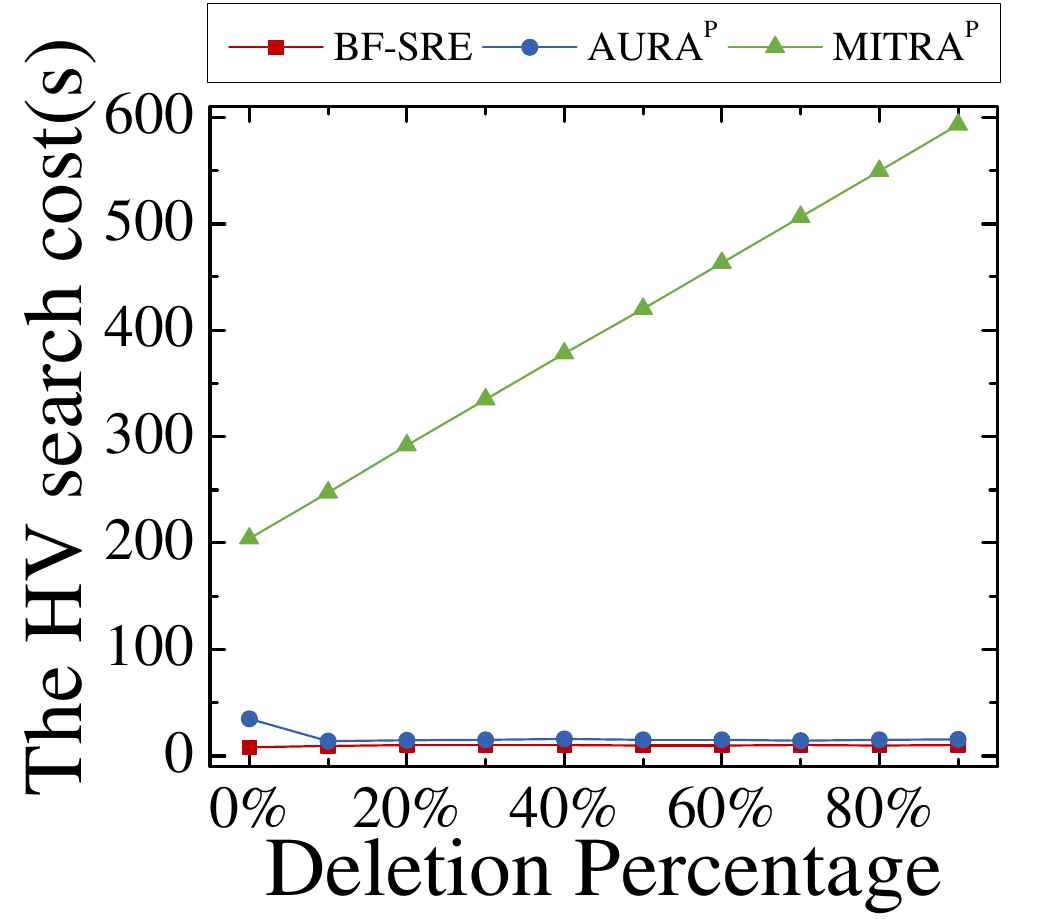}}}}
    }
    \vspace{-5pt}
    \caption{Comparison of \textbf{search time} costs of \textsf{BF-SRE}, \textsf{AURA}$^P$, and \textsf{MITRA}$^P$ on the highest-volume (HV) keyword \textbf{with deletion} percentage (namely, delete 0-\textcolor{black}{90}\% pairs). }
    \label{Dataw}
    \vspace{-10pt}
\end{figure}

\textbf{Highest-Volume (HV) Search time costs.} Fig. \ref{Dataw} provides the fact that \textsf{BF-SRE} is well-perform across all datasets.
We notice that the costs associated with \textsf{BF-SRE} exhibit a slight increase in correlation with the rise in the deletion percentage. 
In Fig. \ref{DatawB},\ref{DatawC}, the costs slowly increase from \textcolor{black}{2.96s} and \textcolor{black}{8.14s} to \textcolor{black}{4.97s} and \textcolor{black}{10.02s}, respectively.
The costs of \textsf{MITRA}$^P$ climb up abruptly as its deletion requires traversing and re-encrypting all dummy data (under padding).  
\textsf{AURA}$^P$ presents a declining cost trend on the costs but still cannot outperform \textsf{BF-SRE}. 
In Fig. \ref{DatawC}, the cost of \textsf{AURA}$^P$ decreases from \textcolor{black}{34.31s} to \textcolor{black}{15.33s} {\color{black} since the deletion operations indirectly reduce its cost on decryption key generation.}

\subsubsection{Update performance}

{\color{black}
We compare update time costs for \textsf{BF-SRE}, \textsf{AURA}$^P$, and \textsf{MITRA}$^P$. 
Due to \textsf{AURA}$^P$ and \textsf{MITRA}$^P$ inheriting the batch update of padding strategies, we record the costs based on Keyword Volume.
For \textsf{BF-SRE}, we record the average addition cost of all pairs associated with the keyword at a specific Keyword Volume.
\textcolor{black}{We also record the client storage cost when adding keyword/value pairs in the sequence of Keyword No.
In other words, for the Crime dataset, we add the corresponding pairs in the order of Keyword No.1 to No.63659 shown in Fig. \ref{test_dataA}.}
}

\begin{figure}[htb]
        \vspace{-10pt}
	\centering
        \subfloat[\textcolor{black}{Crime}]{\label{EncC}\scalebox{0.38}{{\includegraphics[width=8cm,height=7cm]{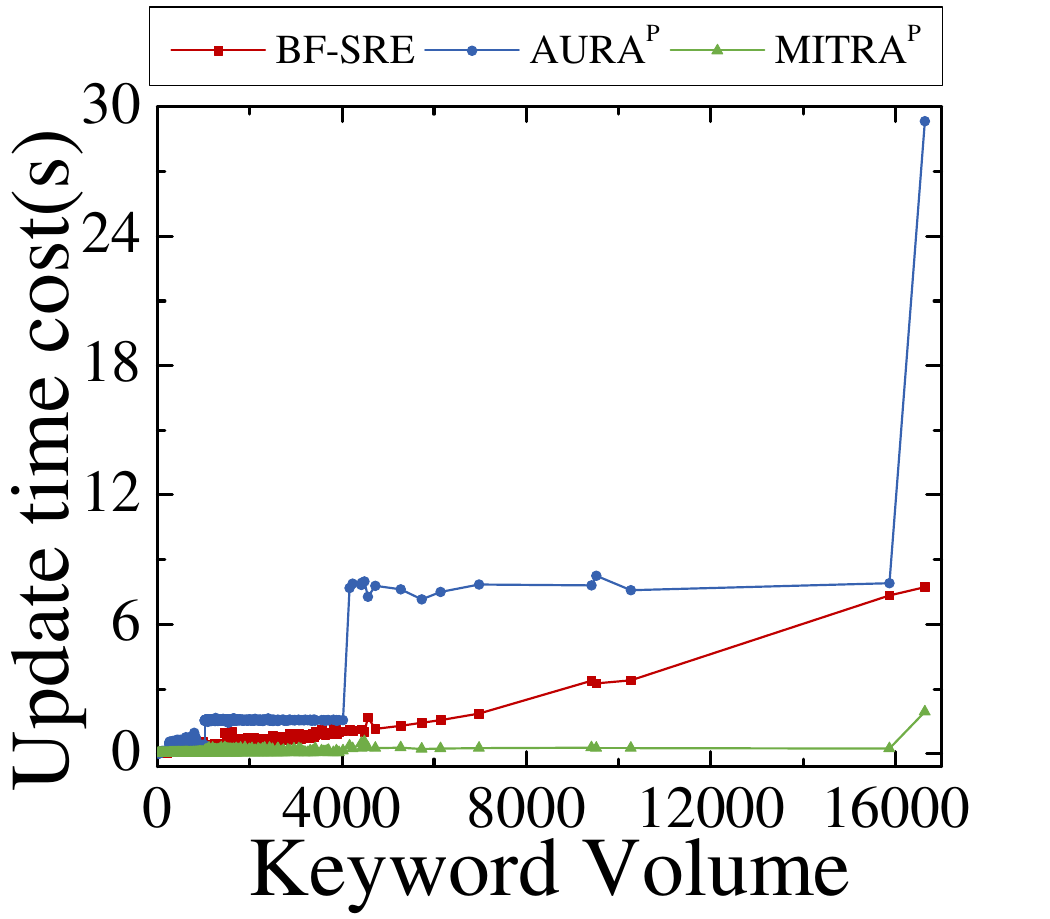}}}}
        \subfloat[Wikipedia]{\label{EncW}\scalebox{0.38}{{\includegraphics[width=8cm,height=7cm]{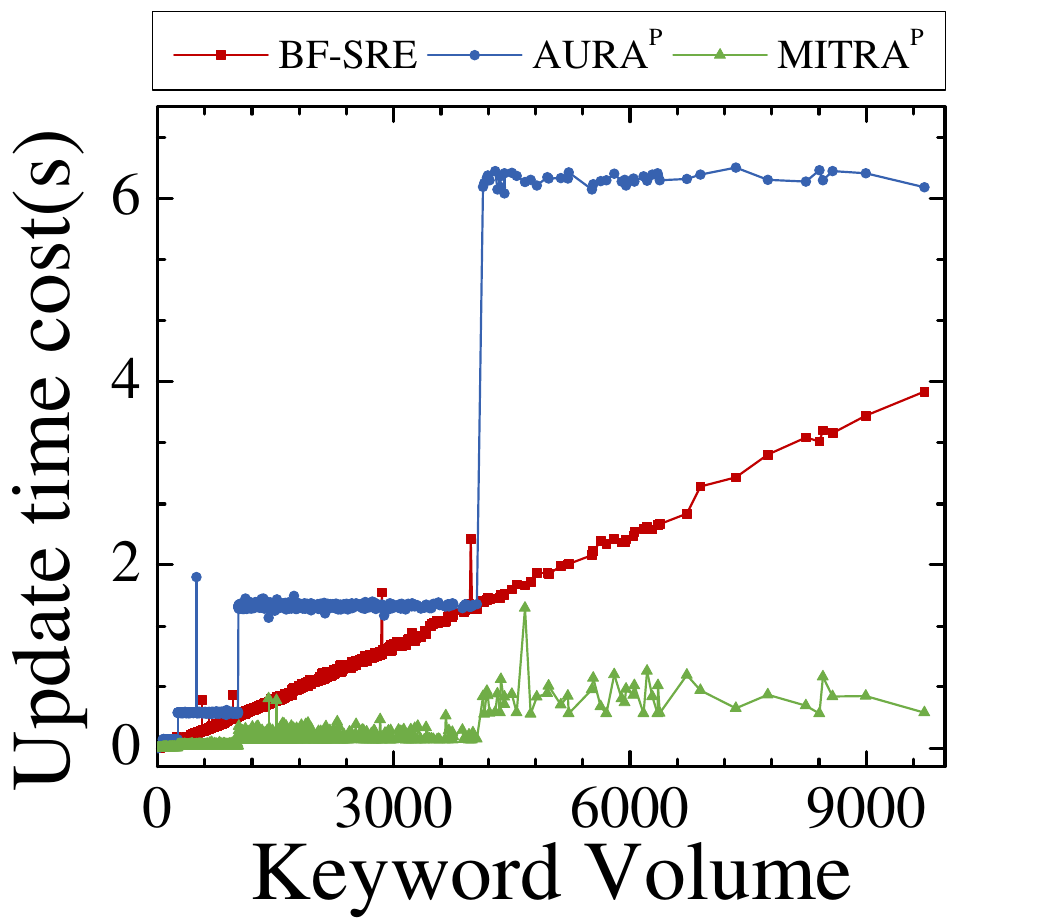}}}}
        \subfloat[Enron]{\label{EncE}\scalebox{0.38}{{\includegraphics[width=8cm,height=7cm]{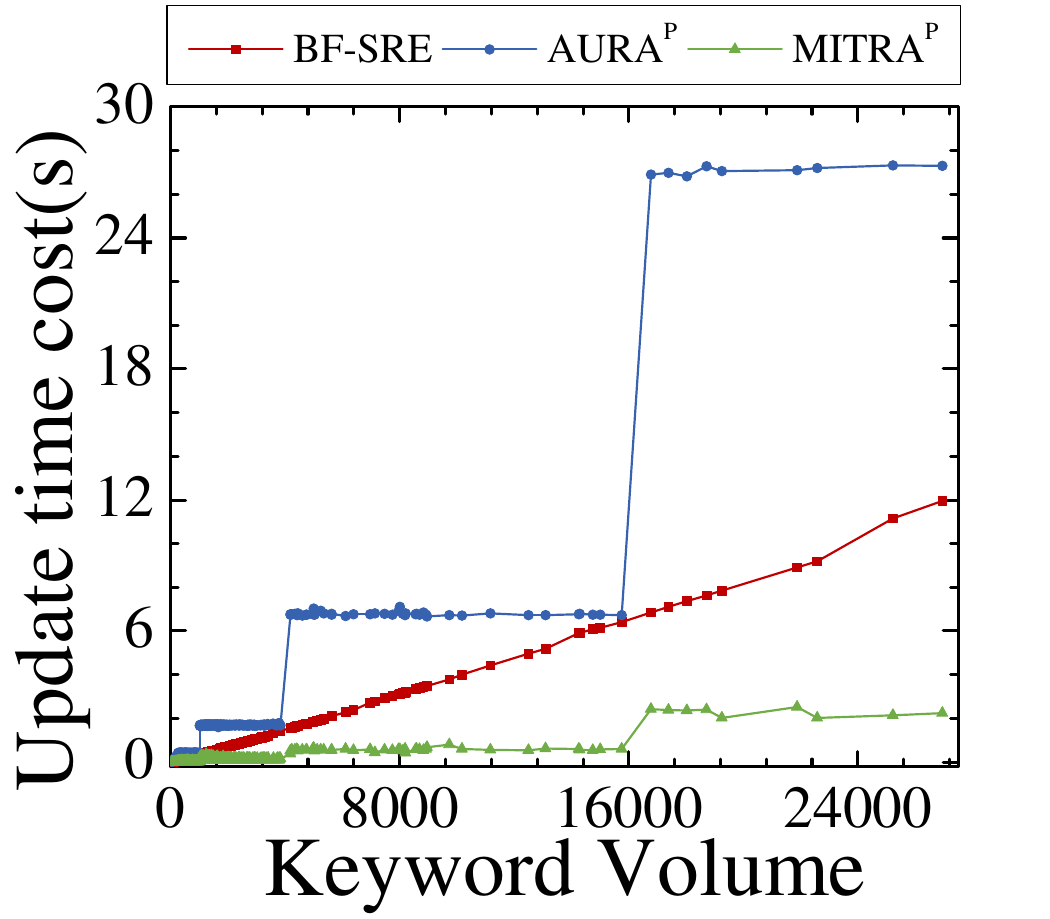}}}}
    \vspace{-5pt}
    \caption{Comparison of \textbf{update time} costs of \textsf{BF-SRE}, \textsf{AURA}$^P$, and  \textsf{MITRA}$^P$.}
    \label{Encryption_time}
    \vspace{-10pt}
\end{figure}

{\color{black}
\textbf{Update time costs.} Fig. \ref{Encryption_time} depicts that the costs of \textsf{BF-SRE} are at the same magnitude in all datasets.
\textsf{BF-SRE} yields linear costs that grow in proportion to the Keyword Volume, while the \textsf{AURA}$^P$ and \textsf{MITRA}$^P$ both experience step-wise increases. 
In detail, for Fig. \ref{EncC}, the costs of \textsf{BF-SRE}, \textsf{AURA}$^P$, and \textsf{MITRA}$^P$ climb from \textcolor{black}{1.52$\cdot 10^{-4}$s, 2.72$\cdot 10^{-4}$s, and 8.8$\cdot 10^{-5}$s to 7.62s, 29.24s, and 1.94s}, respectively.
Both \textsf{AURA}$^P$ and \textsf{MITRA}$^P$ have the mutation at the Keyword Volume of \textcolor{black}{4096 and 15868}, {\color{black}which is due to the use of \textsf{SEAL}'s adjustable padding}.

\begin{figure}[htb]
        \vspace{-10pt}
	\centering
        \subfloat[Crime]{\label{Cstorage_Upt_C}\scalebox{0.38}{{\includegraphics[width=8cm,height=7cm]{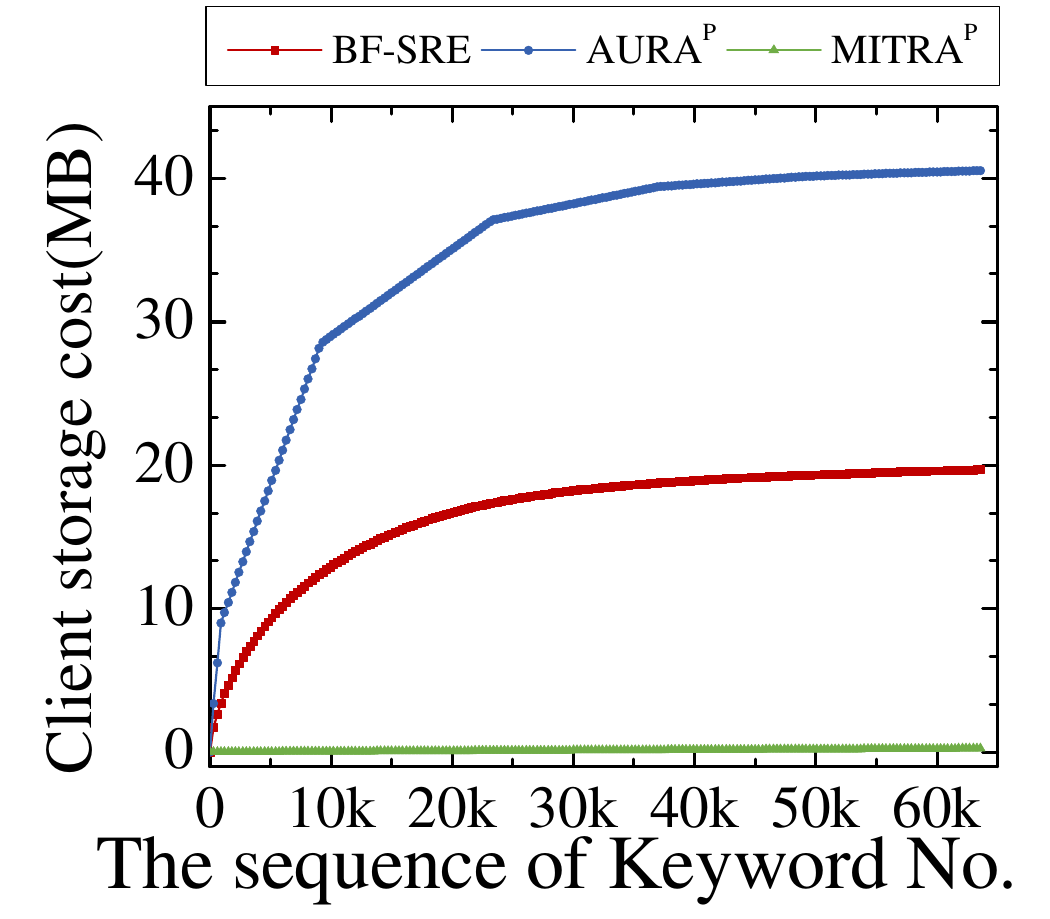}}}}
        \subfloat[Wikipedia]{\label{Cstorage_Upt_W}\scalebox{0.38}{{\includegraphics[width=8cm,height=7cm]{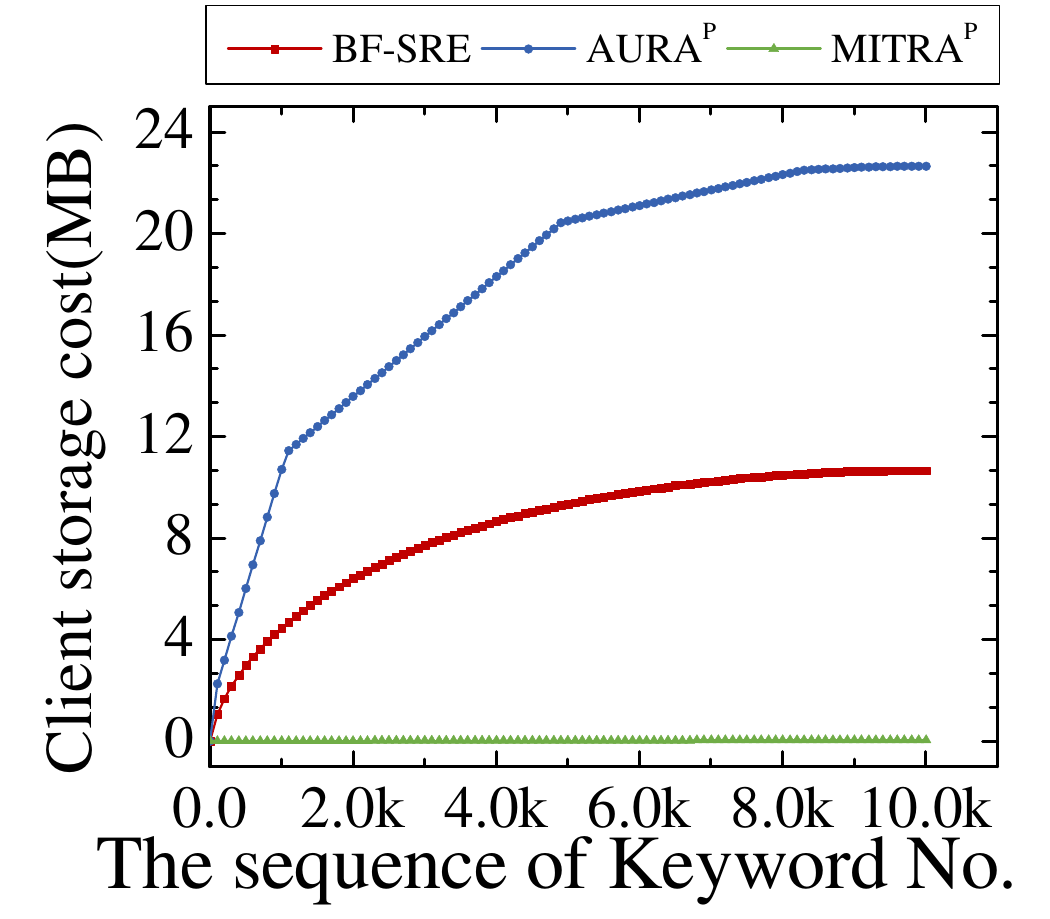}}}}
        \subfloat[Enron]{\label{Cstorage_Upt_E}\scalebox{0.38}{{\includegraphics[width=8cm,height=7cm]{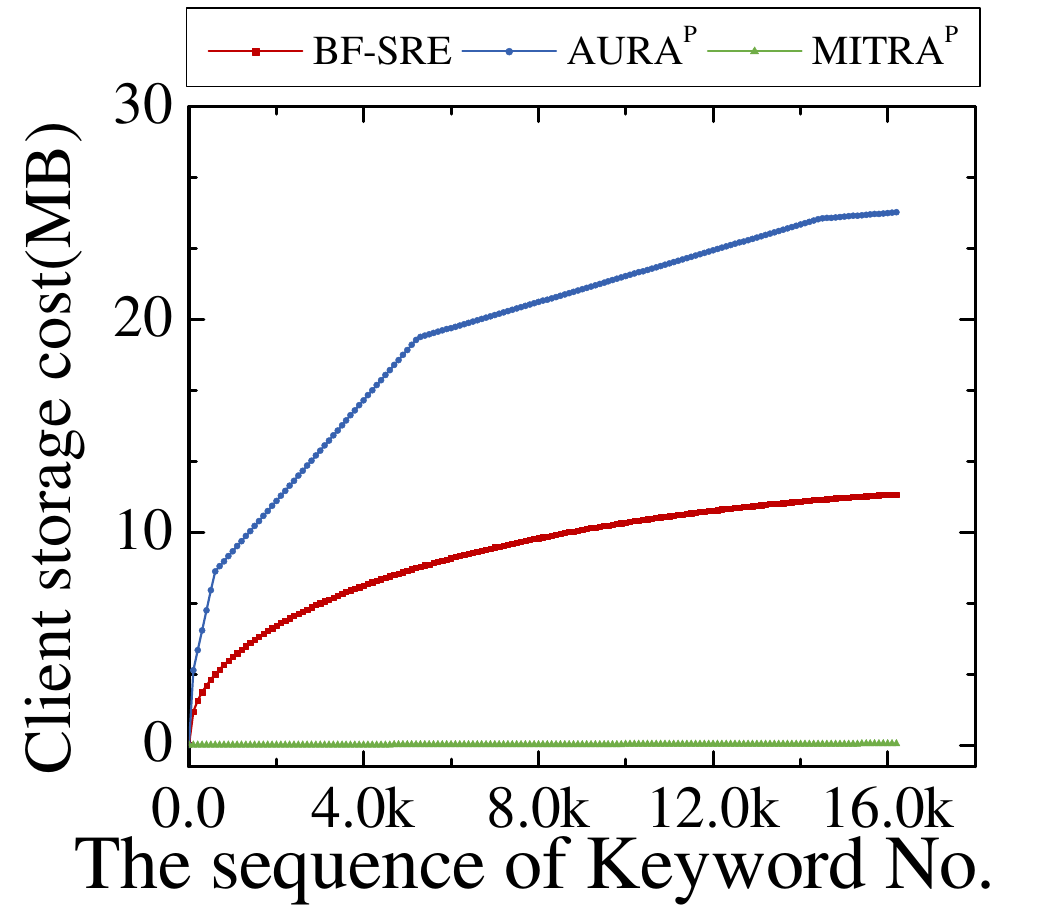}}}}
    \vspace{-5pt}
    \caption{ \textcolor{black}{Comparison of \textbf{client storage} costs of \textsf{BF-SRE}, \textsf{AURA}$^P$, and  \textsf{MITRA}$^P$ for update.}}
    \label{Cstorage_Upt}
        \vspace{-10pt}
\end{figure}

{\color{black}
\textbf{Client storage costs.} Fig. \ref{Cstorage_Upt} illustrates that the costs for \textsf{BF-SRE} fall between those of \textsf{AURA}$^P$ and \textsf{MITRA}$^P$ in all datasets.
Due to the $Default$ optimization, the \textsf{BF-SRE} costs eventually trend toward \textsf{MITRA}$^P$, roughly 1.47x less than \textsf{AURA}$^P$ in Fig. \ref{Cstorage_Upt_C}, \ref{Cstorage_Upt_W}, \ref{Cstorage_Upt_E}, respectively.
In Fig. \ref{Cstorage_Upt_C}, \textsf{AURA}$^P$ exhibits noticeable inflections when the sequence of Keyword No. reaches 899, 9192, and 23351 (corresponding to Keyword Volumes of 1023, 255 and 63) due to the \textsf{SEAL}'s adjustable padding.
}

\subsection{Compare with Padding Strategy}

\begin{figure*}[htb]
    \centering
        \centerline{
            \subfloat[\textcolor{black}{Crime}]{\label{SQLTCrime}\scalebox{0.54}{{\includegraphics[width=10cm,height=6cm]{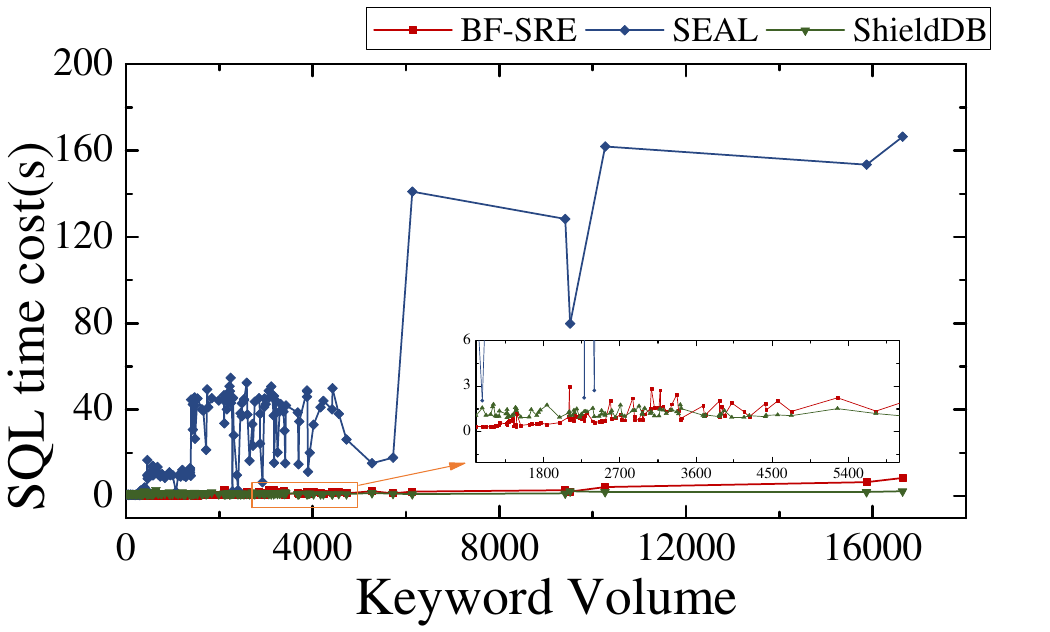}}}}
            \quad
            \subfloat[Wikipedia]{\label{SQLTWiki}\scalebox{0.54}{{\includegraphics[width=10cm,height=6cm]{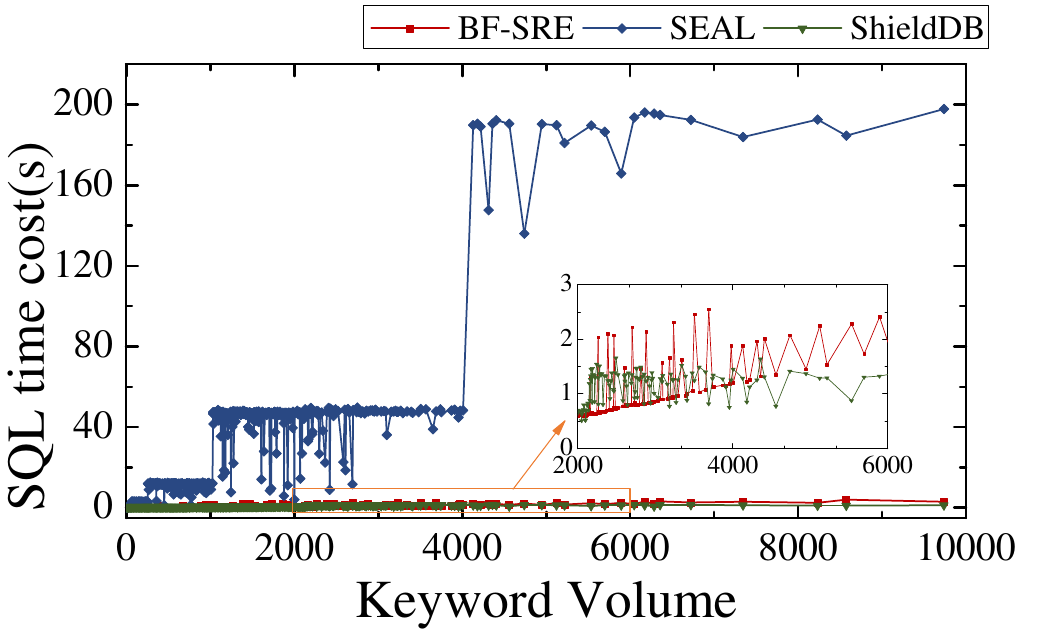}}}}
            \quad
            \subfloat[Enron]{\label{SQLTEnron}\scalebox{0.54}{{\includegraphics[width=10cm,height=6cm]{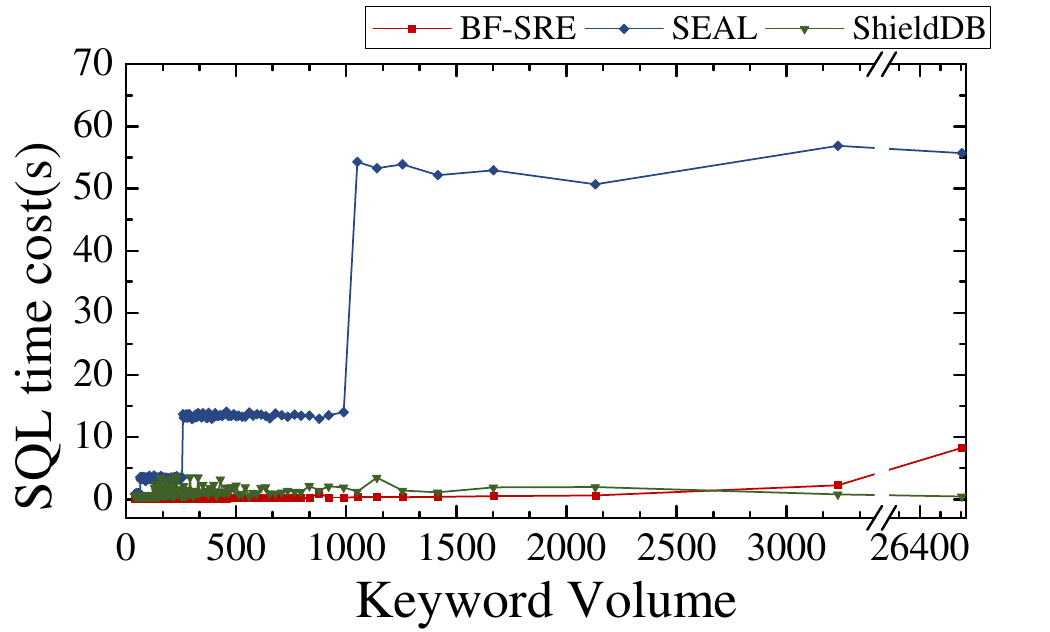}}}}
        }
    \vspace{-10pt}
    \caption{Comparison of the SQL \textbf{time} cost for keyword queries on \textsf{BF-SRE}, \textsf{SEAL}, and \textsf{ShieldDB}.}
    \vspace{-10pt}
    \label{SQLT}
\end{figure*}

\begin{figure*}[htb]
    \centering
    \vspace{-10pt}
        \centerline{
            \subfloat[\textcolor{black}{Crime}]{\label{SQLCCrime}\scalebox{0.54}{{\includegraphics[width=10cm,height=6cm]{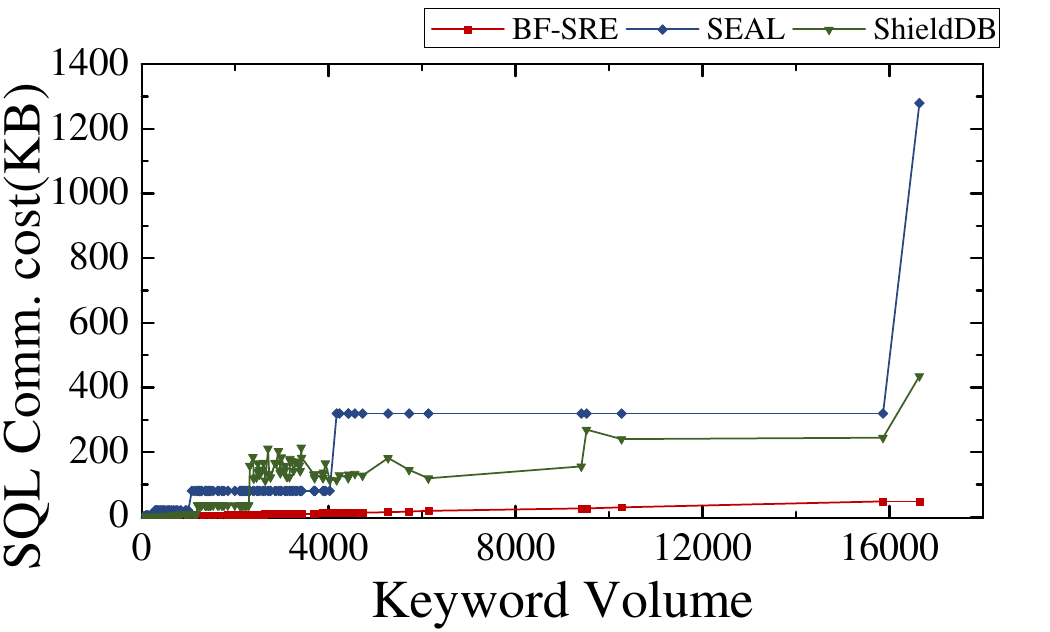}}}}
            \quad
            \subfloat[Wikipedia]{\label{SQLCWiki}\scalebox{0.54}{{\includegraphics[width=10cm,height=6cm]{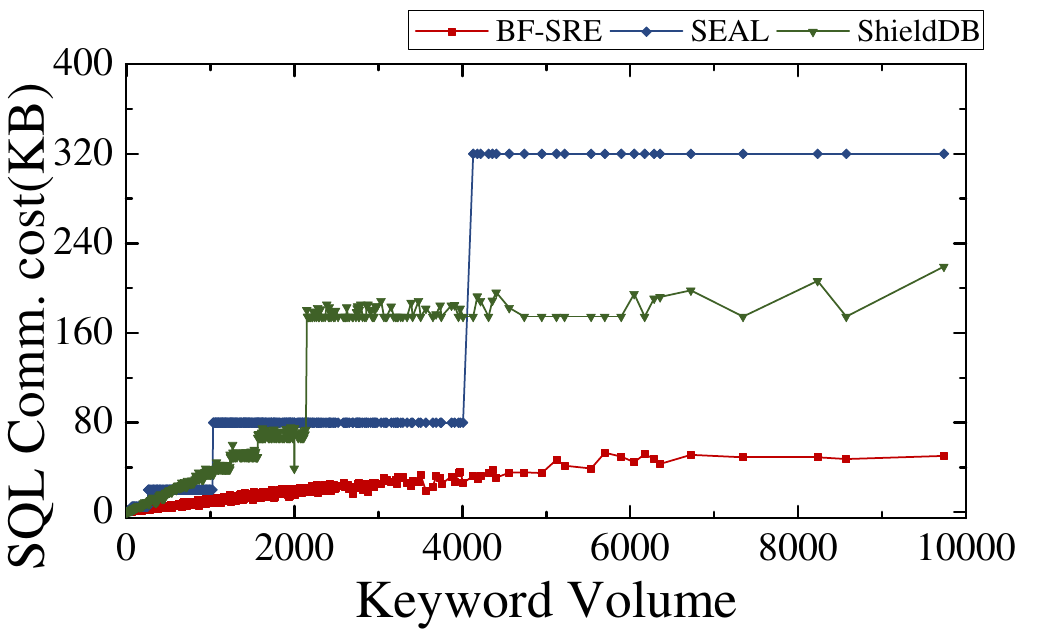}}}}
            \quad
            \subfloat[Enron]{\label{SQLCEnron}\scalebox{0.54}{{\includegraphics[width=10cm,height=6cm]{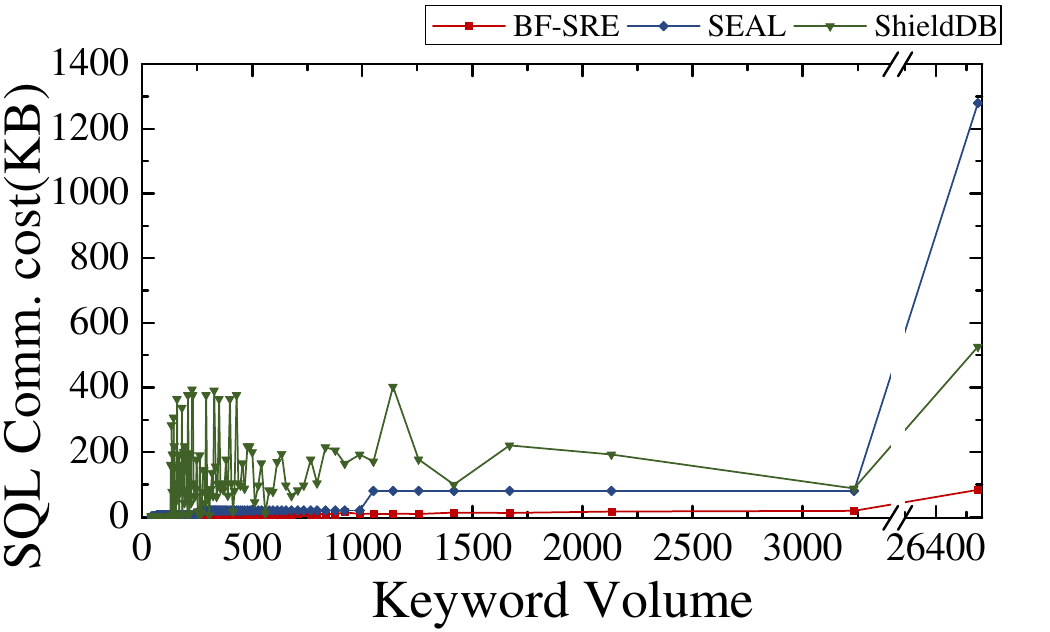}}}}
        }
    \vspace{-10pt}
    \caption{Comparison of the SQL \textbf{communication} cost for keyword queries on \textsf{BF-SRE}, \textsf{SEAL}, and \textsf{ShieldDB}.}
    \vspace{-10pt}
    \label{SQLC}
\end{figure*}

We test the time and communication costs for keyword and join queries in SQL syntax. 
As mentioned in Sec. \ref{AppinEDB}, we first set up the constructions of \textsf{BF-SRE}, \textsf{SEAL}, and \textsf{ShieldDB}, respectively.
Then, we use the Query Planner to extract original queries for each construction to retrieve values from the EDB.

For join queries, we use the Crime dataset to test time and communication costs about its Keyword Volume.
We use Query Planner to request retrieving the relevant IUCR code through keyword (street name, see Sec. \ref{Dataset}) in the Crime table, and then leverage each IUCR code result to find the PRIMARY\_DESCRIPTION in the IUCR table. 
Finally, we let the planner restore the query combined from the two-stage results and log the related costs. 

{\color{black}
\textbf{Keyword query time cost.} Fig. \ref{SQLT}  presents the costs of \textsf{BF-SRE}, \textsf{SEAL}, and \textsf{ShieldDB}.  
From the query results covering the entire keyword space, we see that  \textsf{BF-SRE} provides advantage, nearly 10.89-180x, for \textsf{SEAL} when the Keyword Volume $>10^3$.
Its cost slightly exceeds that of \textsf{ShieldDB} when the Keyword Volume is greater than \textcolor{black}{$4728$}, $4219$, and $2712$ in the Crime, Wikipedia, and Enron dataset, respectively. 
From our analyze, the numbers of keywords below these thresholds are accounted for approximately \textcolor{black}{95.0\%}, 96.5\%, and 98.5\%, revealing that the \textsf{BF-SRE}'s cost is competitive among large keyword spaces.
Note that some `spikes' come from batch read operations\cite{ghareh2018new} and cache retrieval\cite{vo2021shielddb} in \textsf{SEAL} and \textsf{ShieldDB}, respectively.

\textbf{Keyword query communication costs.} In Fig. \ref{SQLC}, \textsf{BF-SRE} consumes much less communication costs than the padding strategies. 
Specifically, it reduces the highest-volume keyword cost up to \textcolor{black}{53.14}x, 6.36x, and 15.27x on the Crime, Wikipedia, and Enron dataset, respectively.
We recall that \textsf{BF-SRE} mainly concentrates on the retrieval of the values, while the padding has to deal with the dummy data.  
}

\begin{figure}[htb]
    \centering
    \vspace{-20pt}
        \centerline{
            \subfloat[\textcolor{black}{Time costs}]{\label{JOINT}\scalebox{0.5}{{\includegraphics[width=8cm,height=7cm]{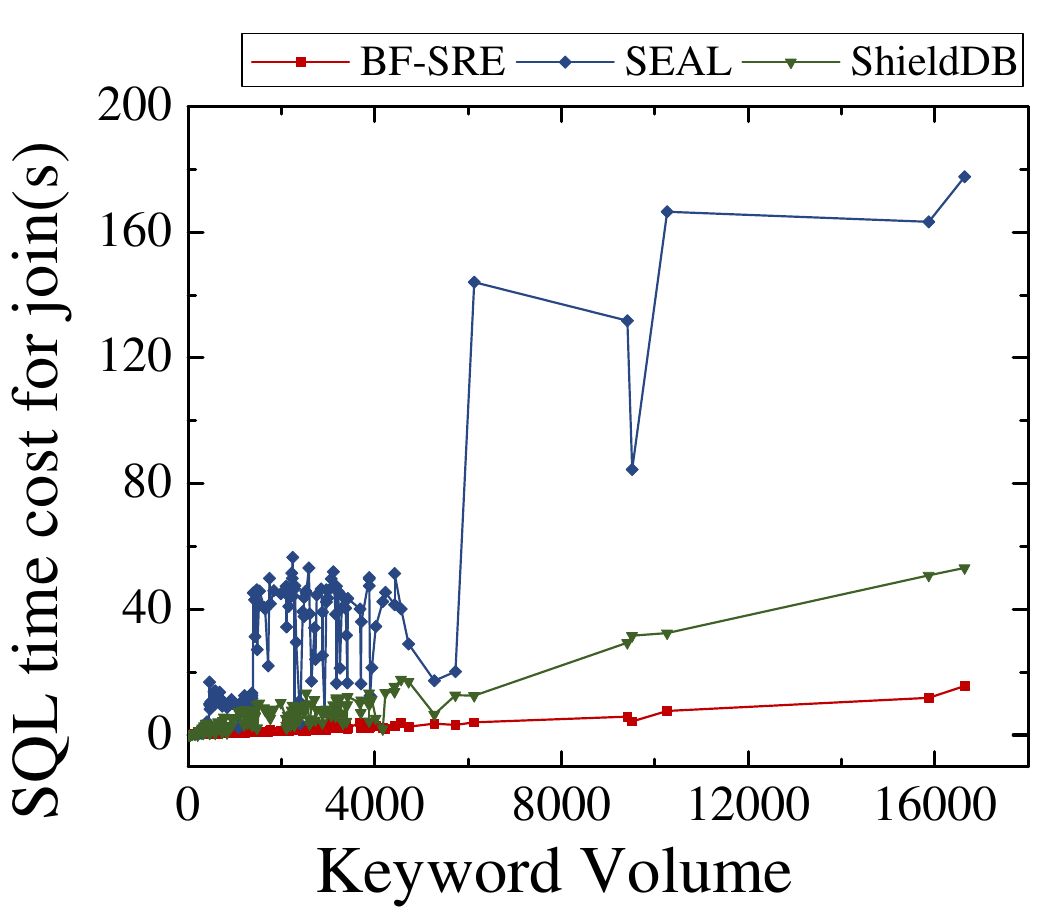}}}}
            \subfloat[\textcolor{black}{Communication costs}]{\label{JOINC}\scalebox{0.5}{{\includegraphics[width=8cm,height=7cm]{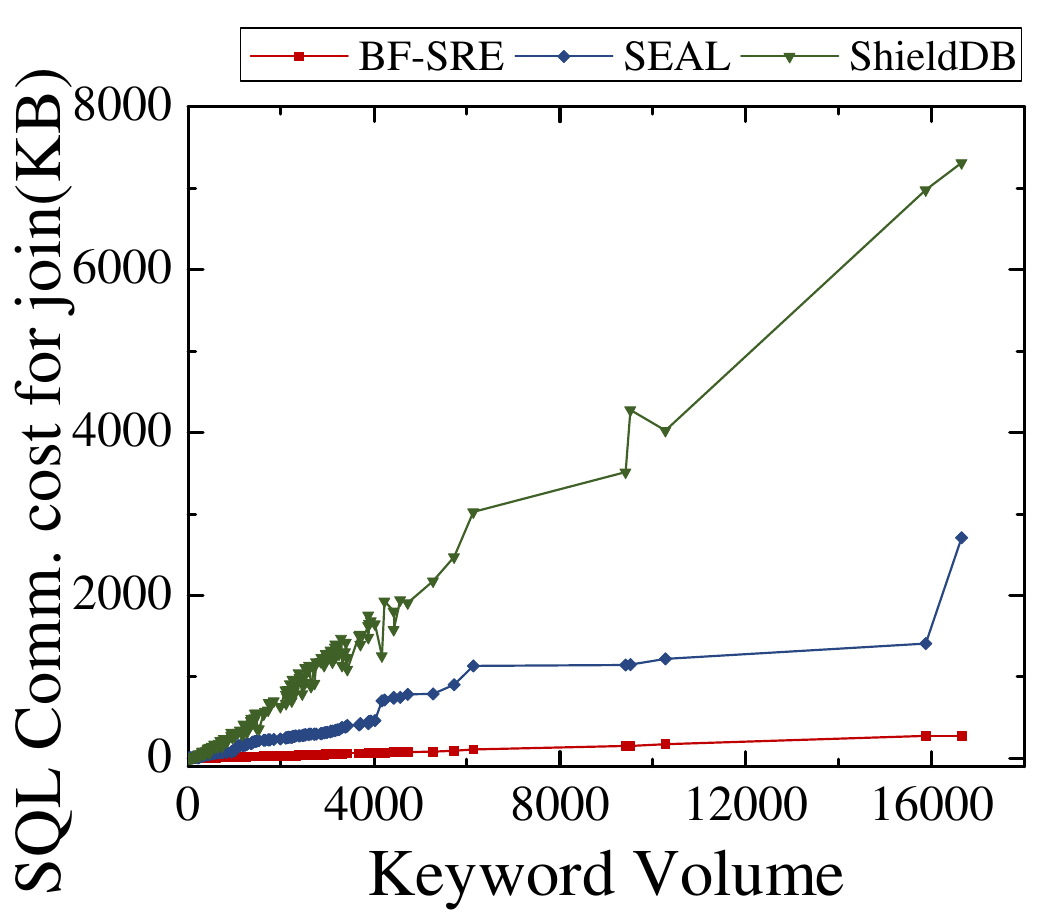}}}}
        }
    \vspace{-10pt}
    \caption{Comparison of \textbf{time and communication} costs for SQL join queries on Crime.}
    \vspace{-10pt}
    \label{Join}
\end{figure}

\textbf{Join query cost.} Fig. \ref{Join} shows that \textsf{BF-SRE} performs the best in the Crime dataset.
Compared to executing SQL keyword queries, \textsf{BF-SRE} does not require a significant rise in the metrics. 
In contrast, \textsf{SEAL} and \textsf{ShieldDB} repeatedly eliminate the dummy data from the IUCR table locally, thereby sharply increasing their costs. 
The costs associated with \textsf{BF-SRE} are roughly \textcolor{black}{26.6x} lower than those of \textsf{ShieldDB} w.r.t. the highest-volume keyword.

When considering overall costs, \textsf{BF-SRE} stands as a competitive performer in comparison to trivial-transformed instances from DSE. 
Moreover, it provides comparable time costs to the padding strategies while delivering significant advantages. 
It demonstrates the ability to efficiently handle various search queries and offers practical communication costs under the tested security parameters.
}

\vspace{-5pt}
\section{Discussions}
\vspace{-5pt}

\textbf{Dataset.} {We notice that there are still a few datasets that have been used for DSE evaluations, e.g., Apache Lucene\cite{zhang2023high}. 
We can apply them to locate distinct words by document names, similar to the approach used in Wikipedia or Enron. 
We argue that using these datasets will not seriously affect the performance of our proposals because Crime/Wikipedia/Enron has a familiar data distribution to other table/document/email datasets \cite{kim2017forward}.  
} 

\textbf{Highest-volume keyword search in deletion.} One may argue that using "highest-volume keyword" may not yield extensive results. 
It is important to note that the search results obtained from representative keywords exhibit minimal fluctuation. 
This is because the search process for these keywords consumes sufficient time and ensures accurate results.
These keywords are more frequently used in practical databases and we use them to clearly highlight the advantage of our proposal. 
We note that interested readers may choose to use less-frequent keywords in the evaluation.  

{\textbf{Other SQL databases.} We use MySQL to store the ciphertexts of the EDB.  
We say that interested readers may use other SQL databases, such as PostgreSQL and Microsoft SQL Server, to implement the \textit{d}-DSE schemes via the Python interfaces.
The parallel capability of SQL databases is not a crucial aspect to consider when choosing among them, as their performance is largely dependent on disk storage.
}

{\color{black}\textbf{The $d$-DSE storage cost.} \textsf{BF-SRE} needs more client storage than \textsf{MITRA}$^P$, which is required by the underlying  SRE revoke structure.
We can apply other DSE schemes as building blocks consuming lower client storage.
\footnote{\textcolor{black}{In Appendix \ref{BF-IPE-example}, we provide an example based on \textsf{MITRA} and the Inner Product Encryption \cite{kim2018function}, obtaining distinct values via the distinct inner products between ciphertexts and search tokens.}}}

\label{ThreatFaccess}
\textbf{Threat from access pattern.} 
Similarly to JOIN-\textit{d}DSE, d-KW-\textit{d}DSE could potentially be vulnerable to access pattern leakage \cite{DBLP:conf/ndss/BlackstoneKM20}, which is a concern associated with practical DSE constructions \cite{demertzis2020dynamic, sun2021practical, patranabis2021forward, chen2023power}.
In certain scenarios, the retrieval process could involve accessing memory, such as acquiring a constant-size file-identifier from a column, which could inadvertently reveal access patterns.
One straightforward strategy to mitigate this issue is to combine d-KW-\textit{d}DSE with ORAM \cite{ghareh2018new} to obscure memory access, albeit at extra costs of time and storage. 
{\color{black} Specifically, we can first retrieve identifiers via d-KW-\textit{d}DSE and then employ ORAM to retrieve each identifier's data. 
When processing queries in batch, we leverage the multi-path ORAM (e.g., OBI \cite{wu2023obi}) to mitigate the high throughput.
}

{\color{black} 

\textbf{Against frequency attacks.} Although we mainly focus on volume leakage, a recent research \cite{xu2023leakage} alerts the frequency(-matching) attack.
The success of the attack hinges on the diversity of the query frequency that is exposed in \textsf{BF-SRE} and prior FP\&BP DSE schemes \cite{ghareh2018new,sun2021practical} through persistent frequency detection.
To mitigate the attack, we can apply a general countermeasure - frequency-smoothing - also used in PANCAKE \cite{grubbs2020pancake}, on the top of $d$-DSE.
It is compatible for combining $d$-DSE and frequency smoothing as: 1) they both leverage keyword/value pairs (KV pairs in their description); 2) their contexts (i.e., static and dynamic frequency distributions) are consistent.
}

\textbf{Ciphertext de-duplication.}
Recall that ciphertext de-duplication  
\cite{Ren2021dedup,liu2015securede,came2019oblivious,Yang2022dedup} is to eliminate the repetitive ciphertexts. 
In particular, this approach, inherited from convergent encryption\cite{liu2015securede}, emphasizes the elimination of ciphertexts with identical contents across multiple clients in order to minimize storage costs.
{But this technique does not focus on ensuring secure value searches.} 

\section{Conclusion}

We explore the distinct search and propose \textit{d}-DSE. 
Following the concept and definition of \textit{d}-DSE, we propose \textcolor{black}{the $d$-DSE designed EDB} and develop \textcolor{black}{\textsf{BF-SRE,}} which satisfies the forward and backward privacy and DwVH security. 
We conduct extensive experiments to highlight the practical performance of our designs in run time, communication, storage, and effectiveness on volume leakage. 

%
%

%
%
%

\begin{footnotesize}
{\color{black}

\bibliographystyle{plain}
\bibliography{refabrv}
}
\end{footnotesize}

\appendix

\vspace{-5pt}
\section{Summary of Notations \& Concepts}\label{Ntandtools}\label{Notation}
\vspace{-5pt}

\begin{table}[ht]
\vspace{-5pt}
    \centering
    \caption{{\color{black}Notations for \textit{d}-DSE}.}
    \vspace{-10pt}
    \scalebox{0.9}{
    \scriptsize
    \label{notation_tables}
    \begin{tabular}{p{2.6cm}p{6.0cm}}
        \hline
            \textbf{\textit{d}-DSE designed EDB} & \textbf{Description} \\
        \hline
            \textcolor{black}{$\mathcal{EDB}_S$} & \textcolor{black}{The collection of encrypted tables stored on disk}\\
             \textcolor{black}{$\texttt{EDB}_C$}  & \textcolor{black}{The replicated encrypted data from $\mathcal{EDB}_S$ on memory}\\
             $q$                        &The query\\
             $\textbf{q}_*$             &The query sequence\\
             $SS$                       &The functionally equivalent syntax\\
             ${\rm t}_*$                &The encrypted data \textcolor{black}{generated from $SS$ }for queries\\
             $syn_*$                    &The specific \textit{d}-DSE construction name\\
             $\mathcal{T}$              &The table collection in EDB storage\\
             $\textbf{T}_*$             &The table name\\
             $\textbf{T}_*.\star$       &The table column\\ 
             $\hat{\textbf{L}}_*$       &The leakage function of constructions\\
             $\textbf{d}$               &The vector containing value's quantities\\
             $ulen(w)$                  &The update length pattern\\
             $drlen(w)$                 &The distinct response length pattern\\
             $qeq(w)$                   &The query equality pattern\\
        \hline
        \hline
            \textbf{BF-SRE} & \textbf{Description} \\
        \hline
            $\rm EDB$& The encrypted database\\
            $\rm EDB_{cache}$ & The cache encrypted database\\
            $\mathbf{C}[w]$ & The map counts the number of search on keyword $w$\\
            $\mathbf{MSK}[w]$   & The map records the master secret key about  $w$\\
            $\mathbf{UpCnt}[w]$ & The map counts the update on  $w$\\
            $\mathbf{D}[w]$ & The map records revoked key structure about $w$\\
            $K_c$           &The encryption key for values\\
            $K_s$           &The encryption key for $\rm EDB_{cache}$\\
            $K_t$           &The encryption key for tag\\
            $t$ &  The real tag \\
            $l$ &  The dummy tag \\
            $\mathbf{H}$ & The Bloom Filter hash collection and bit array\\
            $\mathbf{B}$ & The Bloom Filter bit array\\
            $\Sigma_{add}$ & The Forward Private DSE scheme \\
            $\Phi$       &  The Bloom Filter scheme\\
            $\mathcal{E}\&\mathcal{D}$ & The encryption and decryption algorithm of the symmetric encryption \\
        \hline
        \end{tabular}
        \vspace{-10pt}
    }
\end{table}

\textbf{Symmetric encryption.} Given a security parameter $\lambda\in\mathbb{N}$, message space $\mathcal{M}=\left\{0,1\right\}^*$, ciphertext space $\mathcal{C}=\left\{0,1\right\}^*$ and the key space $\mathcal{K}=\left\{0,1\right\}^\lambda$, a symmetric encryption scheme consists of two algorithm $\left(\mathcal{E},\mathcal{D}\right)$ under the following syntax: \\
$\bullet$ $\mathcal{E}\left(k,m\right)$: Input a symmetric key $k\in\mathcal{K}$ and a message $m\in\mathcal{M}$, output a ciphertext $c \in\mathcal{C}$. \\
$\bullet$ $\mathcal{D}\left(k,c\right)$: Input a symmetric key $k\in\mathcal{K}$ and a ciphertext $c\in \mathcal{C}$, recover a message $m\in\mathcal{M}$.

For correctness, with each message $m\in\mathcal{M}$ and secret key $k\in\mathcal{K}$, the equation $c\gets\mathcal{E}\left(k,m\right)$ should always make sense, and $m\gets\mathcal{D}\left(k,c\right)$ can always recover the message $m$ from $c$ using the secret key $k$.
In the security definitions, a popular requirement of symmetric encryption is the 
INDistinguishability against Chosen Plaintext Attack (IND-CPA).	

\textbf{Bloom Filter\cite{sun2021practical}.}
A Bloom Filter (BF) $\Phi$ is a probabilistic data structure, which can rapidly and space-efficiently perform set membership test.
A BF consists of three polynomial-time algorithms: \\
$\bullet$ $Gen\left(\lambda\right)$: 
    It takes \textcolor{black}{$\lambda$ parsed as} two integers $b,h\in \mathrm{N}$, and samples a collection of universal hash functions $\mathbf{H}$ = $\left\{H_j\right\}_{j\in [h]}$, where $H_j:\left\{0,1\right\}^* \to \left[b\right]$, $\left[b\right]$ denotes a finite set.
    Finally, it outputs $\mathbf{H}$ and an initial b-bit array $\mathbf{B}=0^b$ with each bit $\mathbf{B}\left[i\right]$ for  $i \in \left[b\right] $ set to 0. \\
$\bullet$ $Upd\left(\mathbf{H},\mathbf{B},x\right)$:
    It takes $\mathbf{H}$ = $\left\{H_j\right\}_{j\in [h]}$, $\mathbf{B} \in \left\{0,1\right\}^b$ and an element $x \in \mathcal{X}$, updates the current array $\mathbf{B}$ by setting $\mathbf{B}\left[H_j\left(x\right)\right] \gets 1$ for all $j \in \left[h\right]$ and finally outputs the updated $B$.
    We use $\mathbf{B}_\textsf{S} \gets $$Upd\left(\mathbf{H},\mathbf{B},\textsf{S}\right)$ to denote the final array after inserting all elements in the set \textsf{S} one-by-one. \\
$\bullet$ $Check\left(\mathbf{H},\mathbf{B},x\right)$:
    It takes $\mathbf{H},\mathbf{B}$ and an element $x$, and checks if $\mathbf{B}\left[H_j\left(x\right)\right]$ = 1 for all
    $j\in \left[h\right]$. If true, it outputs 1 and 0 otherwise.

For correctness, a BF is \emph{perfectly complete} if for all integers $b,h \in \mathbb{N}$, any set $\textsf{S} \in \mathcal{X}$, and
$\left(\mathbf{H},\mathbf{B}\right) \gets Gen\left(\lambda\right)$ as well as $\mathbf{B}_\textsf{S} \gets $ $Upd\left(\mathbf{H},\mathbf{B},\textsf{S}\right)$, it holds:
\begin{equation*}
\mathbb{P}\left[\textsf{Check}\left(\mathbf{H},\mathbf{B}_\textsf{S},x\right)=1\right] = 1.
\end{equation*}
\textbf{\quad Symmetric Revocable Encryption\cite{sun2021practical}.}
Symmetric Revocable Encryption (SRE) is a primitive resembled from Symmetric Puncturable Encryption (SPE).
With key space $\mathcal{K}_{SRE}$, message space $\mathcal{M}$ and tag space $\mathcal{T}$, SRE includes four polynomial-time algorithms: \\
$\bullet$ $KGen\left(\lambda\right)$:
    It takes a security parameter $\lambda$ as input and outputs a master secret key $sk \in \mathcal{K}_{SRE}$. \\
$\bullet$ $Enc\left(sk,m,T\right)$:
    It takes as input a $sk$ and a message $m \in \mathcal{M}$ with a list of tags $T \subseteq \mathcal{T}$, and outputs a ciphertext $ct$ for $m$ under tags $T$. \\
$\bullet$ $KRev\left(sk,R\right)$:
    It takes as input $sk$ and a revocation list $R=\left\{t_1,t_i,...,t_\tau\right\} \subseteq T$ , and outputs a revoked secret key $sk_R$, which can be used to decrypt only the ciphertext that has no tag belonging to $R$.
    For \texttt{compress revocation}, it takes $D^\prime \gets $$Comp\left(D,t_i\right)$ to renew revoked key structure $D$ and $sk_R \gets $$ckRev\left(sk,D\right)$ to make the revoked secret key $sk_R$. \\
$\bullet$ $Dec\left(sk_R,ct,T\right)$: It takes as input $sk_R$ and $ct$ encrypted under tags $T$, and outputs $m$ or a failure symbol $\bot$.
    

For correctness, an SRE scheme is correct if the decrypt algorithm returns
the correct result for every input of $sk_R,ct,T$, except with negligible probability.
For security, SRE should provide the adaptive security of IND-REV-CPA and the selective security of IND-sREV-CPA\cite{sun2021practical}.

\vspace{-5pt}
\section{DSE} \label{RDSSE}
\vspace{-5pt}

We notice that the majority of DSE \textcolor{black}{schemes} fail to consider the search and retrieval of the distinct values (which is a fundamental feature in relational databases). 
And none of existing works adequately address distinct search with well-defined security notions. 
DSE is commonly used for retrieving \textcolor{black}{file-identifiers} by keywords\cite{curtmola2011searchable}, and it is assumed that the client does not query the addition of the same pair of keyword and identifier\cite{stefanov2014practical,cash2014dynamic}. 
Kamara et al. \cite{Kamara2012dynamic} conceptualized a database as a collection of files, each represented by a unique identifier.
This assumption has been adopted by subsequent works, e.g.,  \cite{hahn2014searchable}.
Stefanov et al. \cite{stefanov2014practical} used DSE to search the inverted index data structure for keyword/identifier pairs, 
while Cash et al. \cite{cash2014dynamic} revisited definitions from \cite{curtmola2011searchable} to facilitate the storage of document-type data into databases using DSE.
This philosophy has been widely embraced by subsequent works  \cite{chamani2022dynamic,xu2021searching,bost2016ovarphiovarsigma,bost2017forward,sun2018practical,sun2021practical,demertzis2020dynamic}. 
Xu et al. \cite{Xu2022rose} proposed a robust DSE with an extension of backward security. 
Chen et al. \cite{chen2023power} introduced a solution against key compromise in the context of DSE. 
Wang et al. \cite{wang22uno} proposed a solution for keyword search on multi-writer encrypted databases.
{\color{black}
Recently, a volume-hiding DSE construction \cite{Zhao2021VHDSSE} was  built on top of a padding strategy called \textsf{dprfMM} \cite{patel2019mitigating} (involving a dummy dataset). 
But its search complexity is proportional to the maximum response length\cite{zhang2023high}.
} 

\vspace{-5pt}
\section{Padding Strategies} \label{RPDS}
\vspace{-5pt}

The padding strategy leverages false positive information as a means of performing obfuscation. 
Cash et al. \cite{DBLP:conf/ccs/CashGPR15} initially proposed padding for keyword search to counter volume leakage. 
It is important to note that this technique relies on the distribution of the input dataset, making it susceptible to potential leakage even before a search query is initiated.
Many works have been proposed to refine the padding \cite{bost2017thwarting,kamara2018structured,kamara2019computationally}.
{\textsf{SEAL} \cite{demertzis2020seal} introduces an adjustable searchable encryption scheme that provides control over the amount of leaked access pattern information. 
This control is implemented through fine-tuning padding parameters to adjust leakage level. 
}
{But \textsf{SEAL} does not support update operations, which is impractical in the EDB scenarios.} 
\textsf{ShieldDB} \cite{vo2021shielddb} introduces a padding strategy that is explicitly tailored to accommodate a more realistic adversarial model, particularly within the context of databases undergoing continuous updates. 
This solution requires a large amount of dummy files to perform the padding on the batched data. 
We highlight that padding strategies, while effectively mitigating volume leakage, come with extra costs, impacting search efficiency and response time of queries. 
For example, during the initialization stage, they generate dummy data for the encrypted database, which can be computationally intensive, especially for large-scale databases. 
Padding strategies also  tend to increase the amount of data transferred during search operations, thereby raising communication cost.  

\vspace{-5pt}
\section{Encrypted Databases} \label{REDB}
\vspace{-5pt}

Over the past few years, notable efforts have been made in enhancing the security of search operations within encrypted databases. 
These advancements span a spectrum of dimensions, including bolstering data security, refining security schemes tailored to diverse functionalities, and optimizing the overall design of encrypted databases \cite{fuller2017sok}. 
%
Encrypted databases have demonstrated capacity to execute a spectrum of secure functionalities, such as conjunctive search \cite{patranabis2021forward}, keyword range search  \cite{zuo2018range}, and order-revealing encryption \cite{wang2018order}.
Researchers have leveraged structural encryption (e.g., \cite{kamara2018sql} and \cite{george2021structured}), private set intersection \cite{pinkas2014fasterPSI}, private set union  \cite{Jia2022PSU}, and secure hardware \cite{priebe2018enclavedb} to delve into the realm of universal search within encrypted databases.

\begin{table}[htb]
    \centering
    \vspace{-10pt}
    \caption{Comparison of related EDBs.}
    \vspace{-10pt}
    \label{EncDB}
    \scalebox{0.7}{
    \begin{tabular}{|c|c|c|c|}
    \hline
    \textbf{EDB} & \textbf{Volume leakage} & \textbf{Tools for Security} & \textbf{Query mode} \\
    \hline
    CryptDB \cite{Popa2011CyptDB}   &  \Checkmark &  SQL aware Encryption &  SQL \\
    \hline
    Arx \cite{poddar2016arx} & \Checkmark & {DSE} &  SQL \\
    \hline
    EnclaveDB\cite{priebe2018enclavedb} & \Checkmark & SGX & SQL \\
    \hline
    ShieldDB \cite{vo2021shielddb} & $\times$ & DSE + Padding & Keyword based \\
    \hline
    \end{tabular}
    \vspace{-10pt}
    }
\end{table}

The SQL aware encryption used in CryptDB\cite{Popa2011CyptDB} is deterministic, which makes the EDB vulnerable to volumetric attacks. 
At present, researchers have used trusted hardware (SGX)\cite{priebe2018enclavedb} and {DSE}\cite{poddar2016arx,vo2021shielddb} to safeguard the EDBs. 
Tab. \ref{EncDB} shows that current EDBs either overlook volume leakage or strongly rely on padding strategies (with significant storage and communication costs).  

\vspace{-5pt}
\section{Security analysis for BF-SRE} \label{AdpBS}
\vspace{-5pt}

\setcounter{theorem}{0}
\begin{theorem}[Adaptive Security of \textsf{BF-SRE}]
    {
    \itshape
    We define $\mathcal{L}_D=\left(\mathcal{L}_D^{Upt},\mathcal{L}_D^{Srch}\right)$ as:
    \begin{center}
		{
        \footnotesize
		$\mathcal{L}_D^{Upt}\left(w,v,op\right)= op$
  
		$\mathcal{L}_D^{Srch}\left(w\right)= sp(w), {\rm TimeDTS}(w), {\rm Update}(w)$,
		}
		\end{center}
    {
        \vspace{-5pt}
        \textsf{BF-SRE} is $\mathcal{L}_D$-adaptively-secure.
        \vspace{-5pt}
    }
    }
\end{theorem}
\begin{proof}
    We analyze the indistinguishability between \textsf{BF-SRE} and simulator $\mathcal{S}_{BF-SRE}$, and we use the game hop method to analyze indistinguishability.

	\textbf{Game $G_0$.} This game is identical with the real \textsf{BF-SRE}: $\mathbb{P}\left[{\rm Real}_\mathcal{A,S,L}^{\rm BF-SRE}\left(\lambda\right)=1\right] = \mathbb{P}\left[G_0=1\right]$.
	
	\textbf{Game $G_1$}. We replace the calls of $F$ with random strings. 
	When each time a previous unseen call is input, we select a random output from this range space, and record it in tables \textsf{Tokens}, \textsf{Tags}, and \textsf{Counts} for $F\left(K_s,w\right)$, $F\left(K_t,w||v\right)$, and $F\left(K_t,cnt\right)$, respectively.
	Whenever $F$ is recalled on the same input, the output value is retrieved directly from these tables.
	The distinguishing advantage between $G_0$ and $G_1$ is equal to that of PRF against an adversary making at most $N$ calls to $F$:
		$\mathbb{P}\left[G_1=1\right]-\mathbb{P}\left[G_0=1\right]\leq  3Adv_{F,\mathcal{B}_1}^{prf}\left(\lambda\right)$.
  
	\textbf{Game $G_2$}. We replace the Forward Private DSE instance $\Sigma_{add}$ with the associate simulator $\mathcal{S}_{add}^{DSE}$.
	To construct this simulator, we use some bookkeeping to keep track of all the \textsf{Update} queries as they come, and postpone all addition and deletion operations to the subsequent \textsf{Search} query.
	This variant can be done because the additions leak nothing about their contents guaranteed by the forward privacy of $\Sigma_{add}$ and the obliviousness of deletions to server.
	
	Moreover, a list \textsf{Uphist} is initialized and used in this game.
        The list \textsf{Uphist}
        in fact corresponds to the update history on $w$ for the scheme $\Sigma_{add}$ and will be taken as the input of the simulator.
	The distinguishing advantage between $G_1$ and $G_2$ is reduced to the $\mathcal{L}_{FS}$-adaptive forward privacy of  $\Sigma_{add}$. 
	Therefor, there exists a PPT adversary $\mathcal{B}_2$ such that:
        $\mathbb{P}\left[G_2=1\right]-\mathbb{P}\left[G_1=1\right]\leq  Adv_{\Sigma_{add},\mathcal{S}_{add}^{DSE},\mathcal{B}_2}^{\mathcal{L_{FS}}}\left(\lambda\right)$.	    
	
	\quad \textbf{Game $G_3$}. We only modify the generation of the ciphertext deletion.
	More precisely,  we replace the values that were inserted previously and punctured later with constant 0.
	
	Since the modification above works only on the ciphertexts with revoked tags, we can see that the distinguishing advantage between $G_2$ and $G_3$ is the IND-sREV-CPA security of the SRE scheme.
	The selective security is sufficient for the application here, because the reduction algorithm $\mathcal{B}_3$ can obtain the revoked tags from \textsf{Uphist}$\left(w\right)$ and simulate the encrypting process of non-deleted values with the revoked secret key.
	There for, there exists a reduction algorithm $\mathcal{B}_3$ such that:
        $\mathbb{P}\left[G_3=1\right]-\mathbb{P}\left[G_2=1\right]\leq  Adv_{SRE,\mathcal{B}_3}^{\mathrm{IND-sREV-CPA}}\left(\lambda\right)$.
    
	\quad \textbf{Game $G_4$}. We modify the way of constructing addition list $L_{add}$ and the way of updating the compressed data structure $D$.
        $L_{add}$ contains the addition entries and corresponds to the update history \textsf{Uphist} on $w$ for the scheme $\Sigma_{add}$ and is taken as the input of the simulator $\mathcal{S}_{add}^{DSE}$.
	In detail, we first compute the leakage information $\rm TimeDTS$ and $\rm Update$ from the table \textsf{Uphist}, and then base the information to construct $L_{add}$ and update $D$.
	This has no influence to the distribution of $G_3$:
	    $\mathbb{P}\left[G_4=1\right]=\mathbb{P}\left[G_3=1\right]$.
     
	\textbf{Game $G_5$}. We modify the generation of tags in a different way.
	In detail, we replace the tags with random strings directly, instead of computing them from keyword-value pairs and storing them in the table $Tags$.
	The distinguish between $G_4$ and $G_5$ is whether the tags will repeat.
	It is sufficient because each keyword/value/count pair was inserted/deleted at most once during the updates.
	We do not need to record every distinct tag for keeping consistence.
	Therefor, we have:
	    $\mathbb{P}\left[G_5=1\right]=\mathbb{P}\left[G_4=1\right]$.
     
        \textbf{Game $G_6$}. We replace the outputs from symmetric encryption  to random strings on retrievals.
        The difference between $G_5$ and $G_6$ comes from the advantage:
        $\mathbb{P}\left[G_6=1\right]-\mathbb{P}\left[G_5=1\right]\leq  Adv_{\mathcal{E},\mathcal{B}_4}^{\mathrm{IND-CPA}}\left(\lambda\right)$.
        
	\textbf{Simulator.}  When building a simulator from $G_6$, we need to avoid directly using the keyword $w$ as the protocols input. 
	This can be done by replacing the input $w$ with $min\ sp\left(w\right)$.
	To construct $L_{add}$ and collection $D$, we can properly take the leakage $\rm TimeDTS$ and $\rm Update$ as the input of $\textsf{Search}$, and the simulator does not need to keep the track of the updates anymore.
	In this condition, $G_6$ can be efficiently simulated by the simulator with the leakage function $\mathcal{L}$, so we have:
	    $\mathbb{P}\left[G_6=1\right]=\mathbb{P}\left[{\rm Ideal}_\mathcal{A,S,L}^{\rm BF-SRE}\left(\lambda\right)=1\right]$.
     
	\quad \textbf{Conclusion.} By combining all contributions from all games, there exists the adversary such that:
	\begin{equation*}
	    \footnotesize
        \begin{aligned}
	    |\mathbb{P}\left[{\rm Real}_\mathcal{A }^{\rm BF-SRE}\left(\lambda\right)=1\right]-
	    \mathbb{P}\left[{\rm Ideal}_\mathcal{A,S,L}^{\rm BF-SRE}\left(\lambda\right)=1\right]| \leq\\
	    3Adv_{F,\mathcal{B}_1}^{prf}\left(\lambda\right)+Adv_{\Sigma_{add},\mathcal{S}_{add}^{DSE},\mathcal{B}_2}^{\mathcal{L_{FS}}}\left(\lambda\right)+Adv_{SRE,\mathcal{B}_3}^{\mathrm{IND-sREV-CPA}}\left(\lambda\right) + Adv_{\mathcal{E},\mathcal{B}_4}^{\mathrm{IND-CPA}}\left(\lambda\right).
	    \end{aligned}
	\end{equation*}
\end{proof}

{\color{black}
\begin{theorem}[BF-SRE's DwVH Security]
    {
    \itshape
    The leakage function of \textsf{BF-SRE} $\mathcal{L}_D$ = $(\mathcal{L}_D^{Upt}, \mathcal{L}_D^{Srch})$ is Distinct with Volume-Hiding.
    }
\end{theorem}

\begin{proof}

Suppose an adversary $\mathcal{A}_D$ who can the distinguish signature from $\mathcal{L}_D$ and win the DwVH game with advantage $Adv_{BF-SRE,\mathcal{A}_D}^{DwVH}$. 
Then there exist a PPT algorithm $\mathcal{B}_D$ efficiently breaking all security guarantees from PRF $F$, $\Sigma_{add}$, \textsf{SRE}, and the symmetric encryption $\mathcal{E}$, i.e.:

\begin{center}
\small

$MIN(Adv_{F,\mathcal{B}_D}^{prf},Adv_{\Sigma_{add},\mathcal{S}_{add}^{DSE},\mathcal{B}_D}^{\mathcal{L_{FS}}},Adv_{SRE,\mathcal{B}_D}^{\mathrm{IND-sREV-CPA}},Adv_{\mathcal{E},\mathcal{B}_D}^{\mathrm{IND-CPA}}) \geq Adv_{BF-SRE,\mathcal{A_D}}^{DwVH}$.

\end{center}

Note that $\mathcal{B}_D$ should perform like the simulator $\mathcal{S}$ with $\mathcal{L}_D$ that $\mathcal{A_D}$ can access in the DwVH game.
Specifically, it contains four sub-program $\mathcal{B}_D^1$, $\mathcal{B}_D^2$, $\mathcal{B}_D^3$, and $\mathcal{B}_D^4$ to break the secure primitives from $F$, $\Sigma_{add}$, \textsf{SRE}, and $\mathcal{E}$, respectively.
With the guess from $\mathcal{A}_D$, the sub-programs leverage bookkeeping and embed the challenge message to break the above security primitives. 
We describe the process of $\mathcal{B}_D^1$, $\mathcal{B}_D^2$, $\mathcal{B}_D^3$, and $\mathcal{B}_D^4$ as follow.

$\mathbf{\mathcal{B}_D^1}$ gets information from the DwVH game and try to break the PRF.
It separately records three lists $ls_1,ls_2,ls_3$ for tuple $\langle ip_F/iq_F,rs_{w}, rs \rangle$, $\langle ip_F/iq_F,rs_{w}||rs_v, rs \rangle$, and $\langle iq_F,rs_{w}, rs \rangle$, where $ip_F$/$iq_F$ represent the invoke times of PRF in the \underline{Prepare}/\underline{Queries} stage, $rs$ represent the random output string, and $rs_{w}$ and $rs_{v}$ represent the random keyword strings and random value strings.

At the point of $\mathcal{A}_D$ guess the signature $S_{b^*}$ with probability $\epsilon$.
$\mathcal{B}_D^1$ checks the lists to find the input for the challenged $rs^*$ in PRF game. 
In this case, $\mathcal{B}_D^1$ has $ \frac{(1+\epsilon)|ip_F|+(2+\epsilon)|iq_F|}{2|ip_F|+3|iq_F|}$ probability to break the PRF like $\mathcal{B}_1$, where $|ip_F|$ and $|iq_F|$ denotes the number of invoke times of PRF in the \underline{Prepare} and \underline{Queries} stages.
Therefore we have: $Adv_{F,\mathcal{B}_D}^{prf} \geq Adv_{BF-SRE,\mathcal{A_D}}^{DwVH}$.

$\mathbf{\mathcal{B}_D^2}$ inherits $\mathcal{B}_2$ in $G_2$ and additionally receives the guessed signature $S_{b^*}$ to break the forward privacy. At this point, $\mathcal{B}_D^2$ can split the \textsf{Uphist} list as follow: 1) Resize records in each \textsf{Uphist}$(rs_{w})$ until its length equals $t_{b^*}(w)$; 2) Randomly select $t_{b^*}(w)$ values containing $l_{b^*}(w)$ distinct values for \textsf{Uphist}$(rs_{w})$; 3) append the values in tuples of each \textsf{Uphist}$(rs_{w})$ and output the result.

Assume $\mathcal{A}_D$ has $\epsilon$ probability to guess $S_{b^*}$ in the DwVH game. We see that $\mathcal{B}_D^2$ has $\epsilon$ probability to break the state-less leakage of forward privacy in $\mathcal{S}_{add}^{DSE}$ from the guess occurred at the 1-st step of the \underline{$Queries$} stage.
Therefore: $Adv_{\Sigma_{add},\mathcal{S}_{add}^{DSE},\mathcal{B}_D}^{\mathcal{L_{FS}}} \geq Adv_{BF-SRE,\mathcal{A_D}}^{DwVH}$.

$\mathbf{\mathcal{B}_D^3}$ inherits $\mathcal{B}_3$ and additionally record some operation times (i.e., $ip_s,iq_s$) of the encryption and revocation list (i.e., $\langle ip_s/iq_s, rs_v, tags\rangle$ and $\langle ip_s/iq_s, revoke\_tags\rangle$ ) in the DwVH game.  In the \underline{Prepare} stage, $\mathcal{B}_D^3$ probabilistically embeds one IND-sREV-CPA challenge message $m_\mu, \mu\stackrel{\$}{\gets} \{0,1\}$ in the $t_{b}(w)-l_{b}(w)$ repetitive values of keyword $w$. When $\mathcal{B}_D^3$ receives the $S_{b^*}$ from $\mathcal{A}_D$, $\mathcal{B}_D^3$ output 1 if $b^* = b$.

We assume for the sake of contradiction that $\mathcal{A}_D$ can distinguish the outputs from $\mathcal{B}_D$ containing challenge with non-negligible advantage.
Meanwhile, $\mathcal{B}_D^3$ leverages the $\mathcal{A}_D$'s distinguish on the embedding position of the $t_{b}(w)-l_{b}(w)$ repetitive values sequence to break the IND-sREV-CPA security. Therefore: $Adv_{SRE,\mathcal{B}_D}^{\mathrm{IND-sREV-CPA}} \geq Adv_{BF-SRE,\mathcal{A_D}}^{DwVH}$.

\textbf{$\mathcal{B}_D^4$.} Like the embedding in $\mathcal{B}_D^3$, we can also bookkeeping the invoke times and embed the challenge message in the DwVH game. In the \underline{Prepare} stage, $\mathcal{B}_D^4$ will probabilistically embed one IND-CPA challenge message $m_{\mu^\prime}, \mu^\prime\stackrel{\$}{\gets} \{0,1\}$ in the $t_{b}(w)-l_{b}(w)$ repetitive values of keyword $w$. When $\mathcal{B}_D^4$ receives the $S_{b^*}$ from $\mathcal{A}_D$, $\mathcal{B}_D^4$ output 1 if $b^* = b$.

We see that $\mathcal{B}_D^4$ leverages $\mathcal{A}_D$ to break the IND-CPA security of the underling symmetric encryption. Therefore: $Adv_{\mathcal{E},\mathcal{B}_D}^{\mathrm{IND-CPA}} \geq Adv_{BF-SRE,\mathcal{A_D}}^{DwVH}$.

\textbf{Conclusion.} By combining all contributions from all sub-programs of $B_D$, we conclude:

\begin{center}
\small
$MIN(Adv_{F,\mathcal{B}_D}^{prf},Adv_{\Sigma_{add},\mathcal{S}_{add}^{DSE},\mathcal{B}_D}^{\mathcal{L_{FS}}},Adv_{SRE,\mathcal{B}_D}^{\mathrm{IND-sREV-CPA}},Adv_{\mathcal{E},\mathcal{B}_D}^{\mathrm{IND-CPA}}) \geq Adv_{BF-SRE,\mathcal{A_D}}^{DwVH}$.
\end{center}
    
\end{proof}

\vspace{-10pt}
\section{High-level design for $d$-DSE} \label{High-d-DDSE}

\quad
\begin{figure}[htb]
    \vspace{-20pt}
    \centering
    \scalebox{0.32}{\includegraphics[width=26.5cm,height=10.5cm]{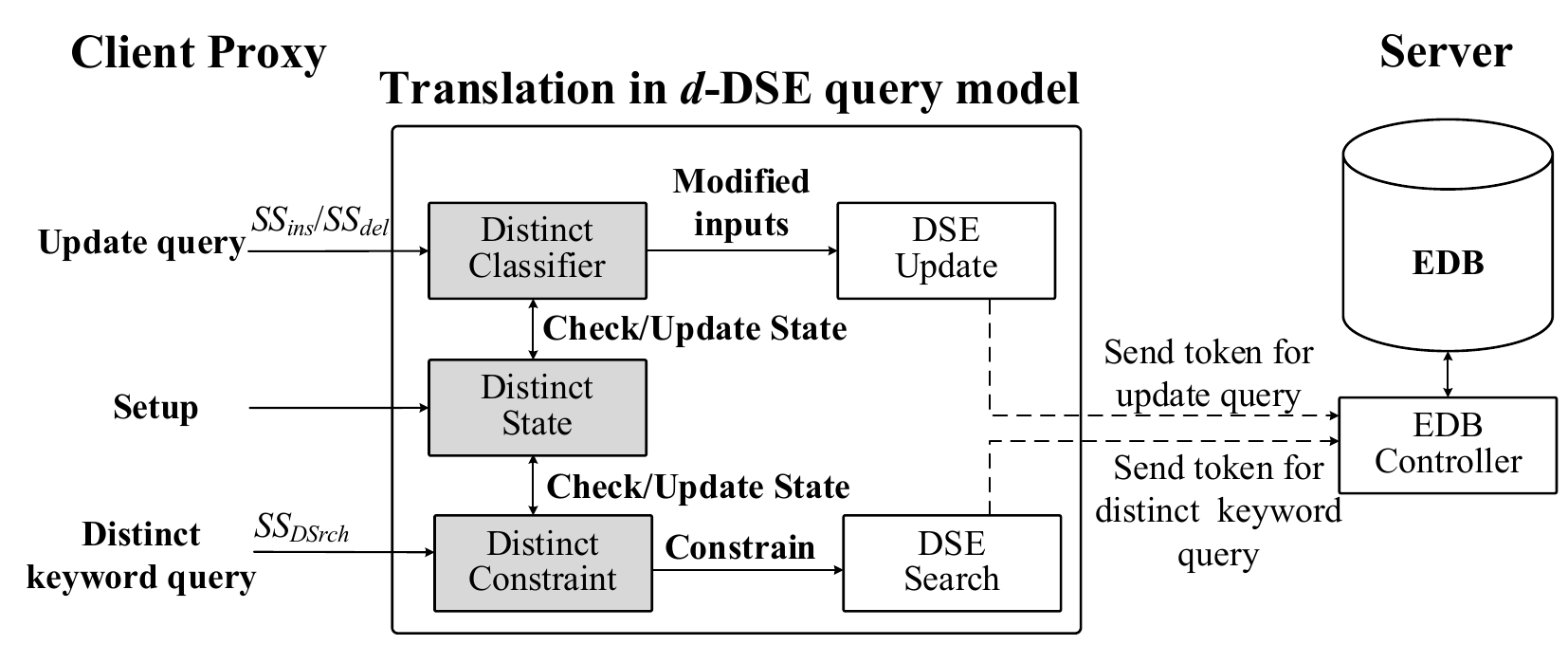}}
    \caption{{The high-level design of $d$-DSE construction.}}
    \label{Apphigh_DSE}
    \vspace{-10pt}
\end{figure}

Fig. \ref{Apphigh_DSE} shows a high-level design for the $d$-DSE construction. 
The concrete construction should provide the components of Distinct Classifier, Distinct State, and Distinct Constraint in the \textsf{Setup}, \textsf{Update}, and \textsf{Search} protocols, respectively. 
The core of achieving  secure distinct search is to:  
 \emph{promptly update the state of (w,v,op) inputs at the initial occurrence, and subsequently conceal the repetitive (w,v,op) inputs.}

\subsection{Store the state of distinct inputs}
The \textsf{Setup} protocol allocates local memory as Distinct State to record the state of whether the input has appeared.
In previous work, the state is commonly used to achieve the forward and backward privacy and special properties such as non-interactivity \cite{sun2021practical} and robustness \cite{Xu2022rose}.
To clarify whether the input is distinct, Distinct State should efficiently store all updates that include distinct $(w,v,op)$ pairs.

We find that the BF \cite{sun2021practical} can satisfy the requirement.
It can expense a small amount of storage overhead to record whether a large number of keyword/value/op pairs appear.
{We initialize BF with the parameters $b,h,n$, which represent the size of BF (in bits), the number of different hash functions, and the number of distinct inputs, respectively.}
Given the input count $n=2^{20}$, to achieve the tolerated false-positive probability $p=10^{-5}$, we can maintain the BF with just 3MB of local storage according to the requirement of the BF size $b=-n \ln{p}/(\ln{2})^2$. 
{This shows that a small local storage is sufficient for BF even dealing with a large-scale dataset.}

\subsection{Tag distinct values} 

The \textsf{Update} protocol employs a program called Distinct Classifier to generate real or dummy tags for each input.
The tags are used to identify whether the matched value is represented as distinct in the \textsf{Search} protocol.
Based on the Distinct State, the inputs attach the real tag as the distinct.
The repetitive inputs are replaced with bogus keyword/value pairs and tagged the dummy tag.
The processed data will be updated to the EDB system.

We state that traditional dummy data generation like that in \cite{vo2021shielddb,demertzis2020seal} is applicable to our requirement.
The inputs marked by dummy tags will not affect the results of distinct values.
Based on the Distinct State, the repetitive inputs can be replaced with random duplicated pairs to construct the arbitrary volume distribution on the EDB system.

\subsection{Retrieve distinct values}

The \textsf{Search} protocol employs the Distinct Constraint program to generate the constrained key and retrieves the distinct values.
Under the requirement from the $d$-DSE security definitions, the tag can only reveal whether the corresponding value is distinct in the \textsf{Search} protocol.

Below, we conduct an analysis on the \textsf{Search} protocols of \textsf{SEAL}\cite{demertzis2020seal} and \textsf{ShieldDB}\cite{vo2021shielddb}.
\textsf{SEAL} designs the SE-based queries with its adjustable Searchable Encryption construction to support conditions (point, range) and join queries.
\textsf{ShieldDB} develops an independent search method for searching the keyword cluster to achieve high search efficiency.
{Both of them require the padding strategy to ensure a uniform volume for search results. 
Thus, they need to retrieve dummy data for every query, consistently increasing the communication cost of the \textsf{Search} protocol. }
{Furthermore, they do not provide a delete operation because this can definitely incur the uneven volume from the \textsf{Search} protocol, which brings an advantage to the adversary in distinguishing the volume.} 

{Since there are drawbacks of the above techniques, we turn to another construction roadmap - using SRE\cite{sun2021practical} and Function-hiding Inner Product Encryption (\textsf{IPE}) \cite{kim2018function}.}
The SRE scheme can pre-compute the tag's authenticity offline and use the Distinct State to generate related search tokens.
{This offline method can minimize ciphertext storage for the server because tags can be checked offline based on the revoked key}.  
The \textsf{IPE} scheme can use auxiliary ciphertext to represent all tags' relationships and classify the real tag online.
This stores more ciphertexts to obtain the corresponding distinct value in the searched sequence. 
But it can reduce the pre-computation overheads and achieve lower client storage cost than the offline approach.

{\color{black}
\section{Example: The BF-IPE scheme}\label{BF-IPE-example}
}

We construct \textsf{BF-IPE} from \textsf{MITRA}\cite{ghareh2018new}, the BF, the \textsf{IPE}, and the Forward Private DSE scheme.
The scheme uploads auxiliary ciphertexts and helps the server to clarify the distinct values.
It is also able to maintain forward and backward security and DwVH security, and obtain the sub-linear search efficiency with low client storage.

\subsection{Inner Product Encryption}

Inner product encryption is a type of functional encryption and computes the inner products $\left<\mathbf{x},\mathbf{y}\right>$ of a ciphertext for a vector $\mathbf{x}$ and a secret key for a vector $\mathbf{y}$.
It reveals no additional information about both $\mathbf{x}$ and $\mathbf{y}$ beyond the inner product.
We adopt the formalism and definitions as follows.

An \textsf{IPE} scheme includes four algorithms: \\ 
$\bullet$ $\textsf{IPE}.Setup\left(1^\lambda,n\right)$: It inputs a security parameters $1^\lambda$, a vector lenghth $n$ and outputs a master secret key $Imsk$ and public parameters $pp$. \\
$\bullet$ $\textsf{IPE}.Enc\left(Imsk, pp, \mathbf{x}\right)$: It inputs the $Imsk$, the $pp$, a vector $\mathbf{x}$ and outputs a ciphertext $ct_{\mathbf{x}}$. \\
$\bullet$ $\textsf{IPE}.KeyGen\left(Imsk, \mathbf{y} \right)$: It inputs the $Imsk$, a vector $\mathbf{y}$ and outputs a secret key $ct_{\mathbf{y}}$. \\
$\bullet$ $\textsf{IPE}.Dec\left(pp,ct_{\mathbf{x}},ct_{\mathbf{y}}\right)$: It inputs the $pp$, the ciphertext $ct_{\mathbf{x}}$, the secret key $ct_{\mathbf{y}}$ and outputs either a value $\left<\mathbf{x},\mathbf{y}\right>$ or the dedicated symbol $\bot$.

The correctness requires that $\textsf{IPE}.Dec\left(pp,ct_{\mathbf{x}},ct_{\mathbf{y}}\right)$ is sure to output $\left<\mathbf{x},\mathbf{y}\right>$ and not $\bot$ when  $\left<\mathbf{x},\mathbf{y}\right>$ is from a fixed polynomial range of values.
The \textsf{SIM}-security captures the security of \textsf{IPE}\cite{kim2018function}.

\subsection{Overview} \label{HighIPE}
\begin{figure}[htb]
	\centering
	\centerline{\scalebox{0.7}{{\includegraphics[width=12.7cm,height=3.8cm]{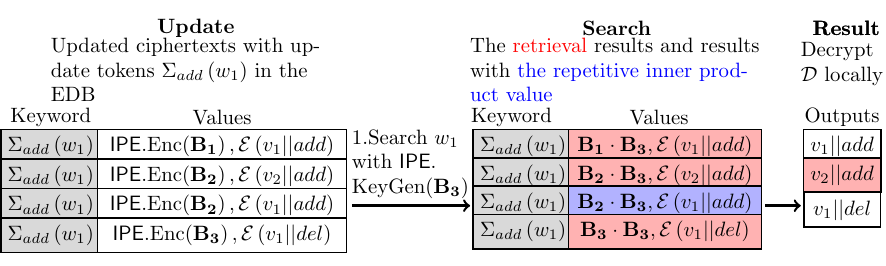}}}}
    \caption{\textsf{BF-IPE}: the scheme overview.}
    \label{d-DSEHighIPE}
    \vspace{-10pt}
\end{figure}

Fig. \ref{d-DSEHighIPE} shows the construction of \textsf{BF-IPE} in the \textsf{Update} and \textsf{Search} protocol.
\textsf{BF-IPE} uploads modified keyword/value pairs in EDB by a Forward Private DSE scheme $\Sigma_{add}$.
To retrieve the distinct values, it leverages the \textsf{IPE} decryption to find the distinct inner products from the past and current vectors' form of the BF, which regards the BF's bit array as the dimensions\footnote{If the value of the BF's bit array is 0011, it can be represented by the vector $(0, 0, 1, 1)$.}.

In the \textsf{Update} protocol, \textsf{BF-IPE} updates $\mathbf{B}$ with the inputs $\left(w,v,op\right)$, gets the corresponding tag by the \textsf{IPE} encryption, and uploads the tag with the symmetric ciphertext of $v||op$.
Note that the BF only renews its value when the input is never seen before, so the new corresponding inner product value can represent the new distinct values. 
For example, in Fig. \ref{d-DSEHighIPE}, \textsf{BF-IPE} inputs $\left(w_1,v_2,add\right)$ and gets $\mathbf{B_2}$ from $\mathbf{B_1}$, and the repetitive input $\left(w_1,v_1,add\right)$ does not change $\mathbf{B_2}$.

In the \textsf{Search} protocol, \textsf{BF-IPE} sends the \textsf{IPE} sub-key for the current vector form of the BF (e.g., $\mathbf{B_3}$) to the server, and the server decrypts the tags to get a result list of inner products.
The first occurrence of the specific inner product values reveals the position of the distinct values. 
Finally, the retrieval is decrypted by the symmetric encryption locally.

The repetitive values cannot be inferred as equal to the values returned in retrieval.
{To this end, \textsf{BF-IPE} performs a dummy addition for the first addition. 
In practice, we can update the first distinct value in batch {so that} the dummy values can be eliminated in retrieval. 

\subsection{BF-IPE Description}

\setcounter{algorithm}{3}
\floatname{algorithm}{Protocol}
\begin{algorithm}
	\scriptsize
	\caption{\textsf{BF-IPE}: Setup, bold lines 3-4 are the Distinct State.}
	\underline{\textsf{Setup}$\left(1^\lambda\right)$:}
 
        //$\mathbf{H}$ and $\mathbf{\mathbf{B}}$ denote the collection of hash functions in BF and the BF bit array, respectively.
        
	\begin{algorithmic}[1]
		\State Initialize the DSE scheme $\left({\rm EDB},\sigma,K_\Sigma\right)\gets\Sigma_{add}.Setup\left(1^\lambda\right)$
		\State $K_t,K_c,K_h \stackrel{\$}{\gets} \left\{0,1\right\}^\lambda$
            \textbf{
		\State Initialize the Bloom Filter $\mathbf{H},\mathbf{B} \gets \Phi.Setup\left(1^\lambda\right)$ 
		\State $\mathbf{UpCnt\gets MAP}$
            }
		\State Initialize the \textsf{IPE} scheme $\left(Imsk,pp\right)\gets \textsf{IPE}.Setup\left(1^\lambda\right)$
        \State Send $\rm EDB$ and $pp$ to the server.
	\end{algorithmic}
\end{algorithm}

\textbf{Setup.}
This protocol initializes the encrypted database $\mathrm{EDB}$, internal state $\sigma$, secret key $K_\sigma$ from a Forward Private DSE $\Sigma_{add}$.
It also generates secret keys $K_t$ and $K_c$, the BF hash collection $\mathbf{H}$, the bit array $\mathbf{B}$, and the empty map \textbf{UpCnt}.
The $K_c$ is used to encrypt values.
The $K_t$ is used to generate tags for the hash of the BF inputs.
Finally, the protocol initializes the \textsf{IPE} scheme and sends the EDB and public parameters $pp$ to the server.

\begin{algorithm}
	\scriptsize
	\caption{\textsf{BF-IPE}: Update, bold lines 2-9 are the Distinct Classifier.}
	\label{BF-IPE-Upt}
        \underline{\textsf{Update}$\left(K_\Sigma,\mathbf{st},op,\left(w,v\right);EDB\right)$:}
        
        //$\stackrel{\to}{\cdot}$ denotes the vector representing each bit of the data in each dimension.
        
	\textbf{Client}:
	\begin{algorithmic}[1]
	\State \textbf{$ cnt \gets \textbf{UpCnt}[w]$
	\If {cnt == $\bot$}
		\State $cnt \gets 0$, $\textbf{UpCnt}[w] \gets cnt$
            \State {\itshape Update a dummy addition along with the first update on a keyword $w$.}
	\EndIf
	\State $\mathbf{t\gets F\left(K_t,w||v||op\right)}$
	\State Update the BF $\Phi.Upd\left(\mathbf{H},\mathbf{B},t\right)$
        \State Random $r \stackrel{\$}{\gets} R$, get the $\mathbf{ct}_{\mathbf{x}}$ as tag $\mathbf{ct}_{\mathbf{x}}\gets\textsf{IPE}.\mathbf{Enc(Imsk},\overset{\xrightarrow{\hspace{0.6cm}}}{\mathbf{B}||r||0^\lambda})$ 
        \State $\mathbf{ct_s \gets \mathcal{E}\left(K_c, v||op||cnt\right),ct \gets \left(ct_{\mathbf{x}}, ct_s \right)}$}
        \State Update the EDB $\Sigma_{add}.Update\left(K_\Sigma,add,w,ct;\mathrm{EDB}\right)$
        \State $\textbf{UpCnt}[w] \gets cnt+1$
	\end{algorithmic}
	
\end{algorithm}

\textbf{Update.}
The protocol updates the EDB.
At lines 1-5, the client reads the internal state and performs the dummy addition for the first input on the keyword $w$.
At lines 6-7, the client uses PRF to generate the BF's input $t$ and updates the BF.
The \textsf{IPE} encrypts the vector form of the BF concatenated with $r||0^\lambda$ (line 8).
To generate the retrieval, the client encrypts the concatenation $v||op||cnt$ by the symmetric encryption (line 9).
At line 10, the client uses the DSE scheme to upload the ciphertexts $\left(ct_x, ct_s \right)$ in the EDB.
Finally, the client updates the \textbf{UpCnt} of $w$ (line 11).

\begin{algorithm}
        \scriptsize
	\caption{\textsf{BF-IPE}: Search, bold lines 1-2 are the Distinct Constraint.}
        \label{BF-IPE-Srch}
	\underline{\textsf{Search}$\left(K_\Sigma,w,\mathbf{st};EDB\right)$:}
 
	\begin{algorithmic}[1]
        \Statex \textbf{Client}:
        \State \textbf{$r \stackrel{\$}{\gets} R$, $\mathbf{ct_{\mathbf{y}}\gets\textsf{IPE}.KeyGen(Imsk,}\overset{\xrightarrow{\hspace{0.6cm}}}{\mathbf{B}||0^\lambda||r})$
	\State Send $ct_{\mathbf{y}}$ to the server}
        \Statex \quad
	\Statex \textbf{Client and Server:}
		\State Run $\Sigma_{add}.Search\left(K_\Sigma,w;EDB\right)$, and the server gets the list $ L=\left(ct_1,ct_2...,ct_l\right)$
	\Statex \textbf{Server}:
	\State Set a list $L^\prime \gets \emptyset$ and a set $S \gets \emptyset$
		\For {$j\in\left[1,l\right]$}
		\State parse $ct_j = \left({ct_{\mathbf{x}}}_j, {ct_s}_j\right)$ from the list $L$
		\If{ the first $ q_j \gets \textsf{IPE}.Dec\left(pp,{ct_{\mathbf{x}}}_j,ct_{\mathbf{y}}\right)$ not in the list $L^\prime$}
        \State Insert $q_j$ into the list $L^\prime$ and  ${ct_s}_j$ into the set $S$
        \EndIf
		\EndFor
        \State Return $S$ to client
        \Statex \quad
        \Statex \textbf{Client:}
        \State Get plaintexts $V = \left(v_1||op_1||cnt,...,v_n||op_n||cnt\right)$ from $\mathcal{D}\left(K_c, s\right)$, where $s \in S$
        \State Get undeleted $V^\prime = \{v_j\ |\ v_j||add||cnt \in V\ and\ v_j||del||cnt \notin V\}$
        \end{algorithmic}
\end{algorithm}

\textbf{Search.}
The protocol identifies the ciphertexts by a keyword $w$ and returns the distinct values.
First, the client generates the \textsf{IPE} decryption key $ct_{\mathbf{y}}$ from the current vector form of the $\mathbf{B}$ concatenated with $0^\lambda||r$ and sends it to the server (lines 1-2).
At line 3, the client and server perform the DSE search protocol, and then the server gets the list $L$.
The server decrypts \textsf{IPE} ciphertexts and retains ${ct_s}_j$ when the new inner product result appears in the list (lines 5-11). 
Finally, the client decrypts the symmetric ciphertexts returned from the server and retains the undeleted distinct values (lines 12-13).  

\subsection{Analysis on BF-IPE}

\textbf{Correctness.} The construction uses the Forward Private DSE to upload the modified keyword/value pairs in encrypted databases. 
With the $F$, it always outputs a tag $t$ from $K_s$.
A probability of losing correctness probability comes from the false positive property of the BF, which is acceptable in practice\cite{sun2021practical}.
Hence, the \textsf{BF-IPE} can correctly update and search distinct values.

\textbf{Security Analysis.}
\textsf{BF-IPE} securely generates ciphertexts, search tokens attached to the Forward Private DSE,  \textsf{IPE}, PRF, and the symmetric encryption. 
The simulator in $\Sigma_{add}$ and the simulator in \textsf{IPE} capture the forward and backward privacy. 
The server only knows the number of update and distinct values on the searched keyword through the \textsf{Search} protocol.

\begin{theorem}[Adaptive Security of \textsf{BF-IPE}]
    {
    \itshape
    Let $F$ with a specific key be modelled as the random oracle $\mathcal{H}_F$, we define $\mathcal{L}_D=\left(\mathcal{L}_D^{Upt},\mathcal{L}_D^{Srch}\right)$ as:
			\begin{equation*}
                \footnotesize
			{
			\begin{split}
			\mathcal{L}_D^{Upt}\left(w,v,op\right)&= op \\
			\mathcal{L}_D^{Srch}\left(w\right)&= sp(w), {\rm TimeDTS}(w), {\rm Update}(w),
			\end{split}
			}
			\end{equation*}
			\textsf{BF-IPE} is $\mathcal{L}_D$-adaptively-secure.
    }
\end{theorem}

\emph{Proof Sketch.} The proof is similar to that in Theorem \ref{theorem1}.
We gradually replace the real cryptographic tools used in the \textsf{BF-IPE} construction with the bookkeeping tables or the corresponding simulators and obtain the \textsf{BF-IPE} simulator.
We first replace the outputs of PRF and the symmetric encryption with random strings and the \textsf{IPE} scheme with the simulator $\mathcal{S}_{IPE}$.
Then the input of the $\textsf{Search}$ protocol is replaced with the $\mathcal{L}_D^{Srch}$ leakage function.
Finally, we use the leakage function \textsf{UpHist}$(w)$, $sp(w)$ as the input of the simulator. The proof is given in Appendix \ref{ProoFIPE}. 

\begin{theorem}[DwVH Security of \textsf{BF-IPE}]
    The leakage function of \textsf{BF-IPE} $\mathcal{L} = (\mathcal{L}_D^{Upt}, \mathcal{L}_D^{Srch})$ is Distinct with Volume-Hiding.
\end{theorem}

\emph{Proof Sketch.} Similar to Theorem \ref{theorem2}. 
The $\mathcal{L}_D^{Srch}$ is also independent of the input keyword/value pairs because the update for real and dummy \textsf{IPE} tags keep identical distinct response lengths.
Therefore, the leakage function reveals nothing about the volume.

\textbf{Complexity.} The complexity is identical to that of \textsf{BF-SRE} except that the client storage is $\mathcal{O}(W)$ based on the keyword counter map (i.e. \textbf{UpCnt}).
\subsection{Improve the Performance} \label{HighP}
The \textsf{IPE} entails significant computational overhead and warrants efforts to enhance its efficiency. 
We have optimized the randomness {via the auxiliary input (line 8 in Protocol \ref{BF-IPE-Upt} and line 1 in Protocol \ref{BF-IPE-Srch})} and the key generation to accelerate the comparison of the inner product values. {The optimized key only decrypts \textsf{IPE} ciphertexts and obtains the related element on cyclic group, without affecting the feasibility of the comparison}.

\begin{figure}[htb]
	\centering
	\centerline{\scalebox{0.8}{{\includegraphics[width=9.2cm,height=3.0cm]{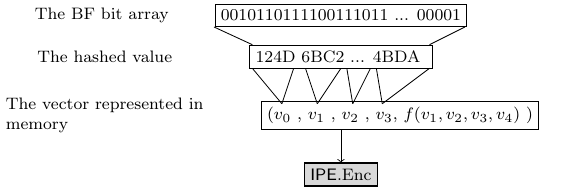}}}}
    \caption{The dimension optimization for \textsf{IPE}.}
    \vspace{-10pt}
    \label{Opt_IPE}	
\end{figure}

We reduce the encrypted vector's dimension as a further optimization for the \textsf{IPE}.
As shown in Fig. \ref{Opt_IPE}, we use the hash function to map the arbitrary-length BF's bit array to the fixed-length hash value. 
Then, we divide the hash value into small fixed-length bit arrays representing the corresponding vector's dimension.
{To ensure comparability between inner products, we use a polynomial to represent the order of these small bit arrays.} 
In practice, one may turn to the Pseudo Random Number Generator \cite{blum1986simple} (PRNG) that uses the hash value as the seed to generate the `checksum' of the order.

\section{Security analysis for BF-IPE} \label{ProoFIPE}

\begin{theorem}[Adaptive Security of \textsf{BF-IPE}]
    {
    \itshape
    Let $F$ with a specific key be modelled as the random oracle $\mathcal{H}_F$. $\mathcal{L}_D=\left(\mathcal{L}_D^{Upt},\mathcal{L}_D^{Srch}\right)$ is defined as:
			\begin{equation*}
			{
			\begin{split}
			\mathcal{L}_D^{Upt}\left(w,v,op\right)&= op \\
			\mathcal{L}_D^{Srch}\left(w\right)&= sp(w), {\rm TimeDTS}(w), {\rm Update}(w),
			\end{split}
			}
			\end{equation*}
			\textsf{BF-IPE} is $\mathcal{L}_D$-adaptively-secure.
    }
\end{theorem}

\begin{proof}
    We analyze the indistinguishability between \textsf{BF-IPE} and simulator $\mathcal{S}_{BF-IPE}$, and we use the game hop to analyze indistinguishability.

	\textbf{Game $G_0$.} This game is identical with the real \textsf{BF-IPE}, so that:
		$\mathbb{P}\left[{\rm Real}_\mathcal{A,S,L}^{\rm BF-IPE}\left(\lambda\right)=1\right] = \mathbb{P}\left[G_0=1\right]$.
	
	\textbf{Game $G_1$}. We replace the calls of PRF $F$ with random strings picked from $F$'s range space. 
	When each time a previous unseen call is input, we select a random output from this range space, and record it in tables \textsf{Tokens}, \textsf{Tags}, and \textsf{Counts} for $F\left(K_s,w\right)$, $F\left(K_t,w||v\right)$, and $F\left(K_t,cnt\right)$, respectively.
	Whenever $F$ is recalled on the same input, the output value is retrieved directly from these tables.
	The distinguishing advantage between $G_0$ and $G_1$ is equal to that of PRF against an adversary making at most $N$ calls to $F$.
	Therefor, we have:
		$\mathbb{P}\left[G_0=1\right]-\mathbb{P}\left[G_1=1\right]\leq  Adv_{F,\mathcal{B}_1}^{prf}\left(\lambda\right)$.
	
	\textbf{Game $G_2$}. We replace the Forward Private DSE instance $\Sigma_{add}$ with the associate simulator $\mathcal{S}_{add}^{DSE}$.
	To construct this simulator, we use some bookkeeping to keep track of all the \textsf{Update} queries as they come, and postpone all addition and deletion operations to the subsequent \textsf{Search} query.
	This variant can be done because the additions leak nothing about their contents guaranteed by the forward privacy of $\Sigma_{add}$ and the obliviousness of deletions to server.
	
	Moreover, a list \textsf{Uphist} is initialized and used in this game.
	The list \textsf{Uphist} contains the encryption of the inserted indices for the subsequent search on $w$, their associated tags and the addition timestamps, which in fact corresponds to the update history on $w$ for the scheme $\Sigma_{add}$ and will be taken as the input of the simulator.
	The distinguishing advantage between $G_1$ and $G_2$ is reduced to the $\mathcal{L}_{FS}$-adaptive forward privacy of  $\Sigma_{add}$. 
	Therefor, there exists a PPT adversary $\mathcal{B}_2$ such that:
        $\mathbb{P}\left[G_1=1\right]-\mathbb{P}\left[G_2=1\right]\leq  Adv_{\Sigma_{add},\mathcal{S}_{add}^{DSE},\mathcal{B}_2}^{\mathcal{L_{FS}}}\left(\lambda\right)$.
 
	\textbf{Game $G_3$}. We modify the generation of the symmetric ciphertexts.
	More precisely,  we replace the encrypted values as random strings from the range of $\mathcal{E}$.
	Since the modification above works only on the symmetric ciphertexts, we can see that the distinguishing advantage between $G_2$ and $G_3$ is the IND-CPA security of the symmetric encryption scheme.
	There exists a reduction algorithm $\mathcal{B}_3$ such that:
        $\mathbb{P}\left[G_2=1\right]-\mathbb{P}\left[G_3=1\right]\leq  Adv_{\mathcal{E},\mathcal{B}_3}^{\mathrm{IND-CPA}}\left(\lambda\right)$.
    
	\textbf{Game $G_4$}. We replace the real \textsf{IPE} scheme with the IPE simulator $\mathcal{S}_{IPE}$.
        Since the modification above works only on the \textsf{IPE} scheme, the distinguishing advantage between $G_3$ and $G_4$ is the SIM security of the \textsf{IPE} scheme.
        There is a reduction algorithm $\mathcal{B}_4$ such that:
        $\mathbb{P}\left[G_3=1\right]-\mathbb{P}\left[G_4=1\right]\leq  Adv_{\textsf{IPE},\mathcal{B}_4}^{\mathrm{SIM-security}}\left(\lambda\right)$.
    
        \textbf{Game $G_5$}. We modify the way of constructing addition list $L_{add}$ from the SSE simulator $\mathcal{S}_{add}^{DSE}$.
	In detail, we first compute the leakage information $\rm TimeDTS$ and $\rm Update$ from the update history table \textsf{Uphist}, and then base the information to construct $L_{add}$.
	This has no influence to the distribution of $G_3$, so that:
	    $\mathbb{P}\left[G_4=1\right]=\mathbb{P}\left[G_5=1\right]$.
 
	\textbf{Simulator.}  To build a simulator from $G_5$, we need to avoid directly using the keyword $w$ as the protocols input. 
	This can be done by replacing the input $w$ with $min sp\left(w\right)$.
	To construct $L_{add}$, we can properly take the leakage $\rm TimeDTS$ and $\rm DelTime$ as the input of $\textsf{Search}$, and the simulator does not need to keep the track of the updates anymore.
	In this condition, $G_5$ can be efficiently simulated by the simulator with the leakage function $\mathcal{L}$, so we have:
	    $\mathbb{P}\left[G_5=1\right]=\mathbb{P}\left[{\rm Ideal}_\mathcal{A,S,L}^{\rm BF-IPE}\left(\lambda\right)=1\right]$.
	
	\textbf{Conclusion.} By combining all contributions from all games, there exists the adversary such that:
	\begin{equation*}
	    \small
        \begin{aligned}
	    \mathbb{P}\left[{\rm Real}_\mathcal{A,S,L}^{\rm BF-IPE}\left(\lambda\right)=1\right]-
	    \mathbb{P}\left[{\rm Ideal}_\mathcal{A,S,L}^{\rm BF-IPE}\left(\lambda\right)=1\right] \leq\\
	    Adv_{F,\mathcal{B}_1}^{prf}\left(\lambda\right)+Adv_{\Sigma_{add},\mathcal{S}_{add}^{DSE},\mathcal{B}_2}^{\mathcal{L_{FS}}}\left(\lambda\right)+Adv_{\mathcal{E},\mathcal{B}_3}^{\mathrm{IND-CPA}}\left(\lambda\right)+ \\Adv_{\textsf{IPE},\mathcal{B}_4}^{\mathrm{SIM-security}}\left(\lambda\right).	 
	    \end{aligned}
	\end{equation*}
\end{proof}
}

\end{document}